%% file: gcbp_paper.tex
\documentclass[]{article}

\usepackage{times}
\usepackage{amsmath,amssymb,amsthm,graphicx}
\usepackage{yhmath}
\usepackage{enumerate}
\usepackage{enumitem}
\usepackage{color}
\usepackage{rotating}

\newcommand{\Cy}{\mathcal{C}}
\newcommand{\E}{\mathcal{E}}
\newcommand{\F}{\mathcal{F}}
\newcommand{\G}{\mathcal{G}}
\newcommand{\T}{\mathcal{T}}
\newcommand{\PP}{\mathcal{P}}

\newcommand{\I}{\mathcal{I}}
\def\R{\mathcal{R}}

\newcommand{\V}{\mathcal{V}}

\newcommand{\x}{\>x}

\newcommand{\h}{\>h}
\newcommand{\J}{\>J}
\newcommand{\vs}{{v^\star}}
\newcommand{\s} {{\>s}}
\newcommand{\bm}{\breve m}
\newcommand{\bchi}{\breve \chi}
\newcommand{\bh}{\breve h}
\newcommand{\hm}{\hat m}
\newcommand{\hchi}{\hat \chi}
\newcommand{\Yi}{\Theta}
\newcommand{\bYi}{\breve \Theta}
\newcommand{\Yp}{{\cal E}}

\def\blue{\color[rgb]{0,0,1}}
\def\red{\color[rgb]{1,0,0}}

\newtheorem{prop}{Proposition}[section]
\newtheorem{coro}[prop]{Corollaire}

\newtheorem{lem}[prop]{Lemma}

\makeatletter

\@addtoreset{equation}{section}

\@addtoreset{figure}{section}
\makeatother

\newcommand\egaldef{\stackrel{\mbox{\upshape\tiny def}}{=}}

\newcommand\1{\leavevmode\hbox{\rm \small1\kern-0.35em\normalsize1}}

\def\DD{\displaystyle}

\DeclareMathOperator*{\argmin}{argmin}
\DeclareMathOperator*{\argmax}{argmax}
\DeclareMathOperator{\Tr}{\text{Tr}}

\newcommand\dkl{\text{D}_\text{KL}}

\title{Cycle-based Cluster Variational Method for Direct and Inverse Inference.}

\author{Cyril Furtlehner\thanks{Inria Saclay - LRI, Tao project team, B\^at 660 Universit\'e Paris Sud, Orsay Cedex 91405}\hspace{0.2cm} and 
Aur\'elien Decelle\thanks{LRI, AO team, B\^at 660 Universit\'e Paris Sud, Orsay Cedex 91405}%
}

\begin{document}

\maketitle

\begin{abstract}
We elaborate on the idea that loop corrections to belief propagation 
could be dealt with in a systematic way on 
pairwise Markov random fields,
by using the elements of a cycle basis to define region in a generalized belief propagation setting.
The region graph is specified in such a way as to avoid dual loops as much as
possible, by discarding redundant Lagrange multipliers, 
in order to facilitate the convergence, while avoiding instabilities 
associated to minimal factor graph construction.
We end up with a two-level algorithm, where a belief propagation algorithm is run 
alternatively at the level of each cycle and at the inter-region level.
The inverse problem of finding the couplings of a Markov random field from empirical covariances
can be addressed region wise. It turns out 
that this can be done efficiently 
in particular in the Ising context, where fixed point equations can be derived along with 
a one-parameter log likelihood function to minimize. Numerical experiments confirm the effectiveness of
these considerations both for the direct and inverse MRF inference. 
\end{abstract}

\section{Introduction}
Markov random fields~\cite{Lauritzen} (MRF) are widely used probabilistic models, able to 
represent multivariate structured data in order to perform inference tasks. 
They are at the confluence of probability, statistical physics and machine learning~\cite{WaJo}.
From the formal probabilistic viewpoint they express the conditional independence properties
of a collection of $n$ random variables $\x=\{x_1,\ldots,x_n\}$, in the form of a factorized probability measure, where each factor 
involves a subset of $\x$. In statistical mechanics the Gibbs measure 
takes the form of an MRF, to express
the thermodynamic equilibrium probability of a system of $n$ degrees of freedom in interactions.
The practical use of MRF appears also in various applied fields,  like image processing, bioinformatics,
spatial statistics or information and coding  theory. Recent breathtaking successes  in artificial
intelligence have been obtained by learning deep neural networks which building blocks are
so-called restricted Boltzmann machine i.e. bipartite networks of Ising spins in interaction. 
By stacking them into deep architectures some high level features can be learned recursively~\cite{LeBeHi} 
using schematically Monte-Carlo based learning algorithms in combination with Bragg-Williams mean-field method 
within a Gibbs-sampling loop. The use of more advanced mean-field methods like the 
cavity approach could be possibly helpful in this context~\cite{GaTrKr}. The main difficulty resides in the fact that these MRF are being of 
practical use in a domain of parameters which clearly corresponds to an ordered phase with strong couplings, which is usually not
the most favorable one for applying mean-field methods.   
Letting aside this potential difficulty, let us simply state the 
two main generic problems that have to be commonly dealt with when using MRF in practical applications:
\paragraph{\it Direct inference problems:} 
\begin{itemize}
\item computation of marginal probabilities
\[
p_i(x_i) = \sum_{\x\backslash x_i} P(x_i),
\]
which involves in general an exponential cost w.r.t. $N$ to be done exactly;
\item computing the mode, also referred to as the maximum a posteriori probability (MAP)
\[
\x^\star = \argmax_{\x} P(\x),
\]
which is generally an NP hard problem~\cite{Cooper,Shimony}.
\end{itemize}
\paragraph{\it Inverse problem:} 
learning the parameters of the model, given for example by sufficient statistics when 
the MRF is in the exponential family. For instance the inverse Ising problem~\cite{LeGaKo,HoTi,KaRo,WeTe,MoMe,YaTa,CoMo,NgBe}
consists in to find the set of couplings $\{J_{ij}\}$ 
and external fields $\{h_i\}$  of an Ising model
\[
P(\s) = \frac{1}{Z(\h,\J)}\exp\Bigl(\sum_{i,j}J_{ij}s_is_j+\sum_i h_i s_i\Bigr),
\]
which maximize the associated log likelihood (LL),
given data in form of sequences $\s^{(k)}, k=1\ldots M$ or of empirical marginals $\hat E(s_i)$, $\hat E(s_is_j)$.
Generally the partition function $Z(\h,\J)$ requires an exponential cost w.r.t. $N$ to be computed exactly.

In order to be useful, any approach based on MRF modeling relies therefore strongly on efficient approximate algorithms, 
since both direct and inverse problems have potentially an exponential cost w.r.t. to system size.
Belief propagation (BP)
and its generalizations GBP~\cite{YeFrWe} have opened the possibility for using MRF in 
large scale problems even though many restrictions stand in the way of a systematic use, either from 
convergence problems or from precision performances.  
In particular, the use of GBP is hampered by notoriously difficult 
convergence problems, which have led some authors~\cite{Yuille,HeAlKa} to consider double loop algorithms, at the price of 
some computational costs~\cite{pelizzola}.
In addition the choice to be made for 
region definition is rather open in general, except that a bad choice may lead to poor precision and lack of convergence~\cite{Welling2004},
and too large regions are excluded, as computational cost grows exponentially w.r.t. the size of the largest regions. 
For regular graphs, regions are straightforwardly 
identified for example with square plaquettes or cells of $2-$D and $3-$D  lattices, as in the Kikuchi 
cluster variational methods~\cite{Kikuchi,pelizzola} (CVM). But for general graphs, a systematic choice 
of regions is more difficult to define and also some complexity problem may 
occur if the size of regions is not controlled. 
As suggested in~\cite{WelMinTeh2005} a good choice for the regions to run GBP
might be provided by a cycle basis and possibly a weakly fundamental~\cite{GeWe} cycle basis.
An alternative line of research which has been also followed over recent years consists in to estimate
loop corrections to Bethe-Peierls approximation in order to improve its accuracy, by addressing directly the errors
caused by the presence of loops on multiply connected factor graphs~\cite{MoRi,chertkov1,PaSla,MoKa07,XiZh,Ramezanpour,Fu2013}. 
In the present work, we investigate further along these directions
by generalizing in some way previous considerations ~\cite{LaMuRiRi,Fu2013} concerning the random Ising model in absence of local fields. 
Firstly we analyze in this context convergence problems emerging from canonical definitions of the region graph. This leads us to propose 
a specific construction of the factor graph, which to some extent solves the convergence issue, as is observed experimentally. 
Secondly, we exploit a property of the minimizer of  Kikuchi free energy functional associated to certain cycles basis,
such that the message to be send from one region to another can be computed efficiently with help 
of an internal BP routine to be performed within each (cycle) region, allowing for arbitrary loop sizes to be considered. 
For binary variables in particular, it is worth exploiting the fact that BP has one single fixed point on 
a circle~\cite{Weiss}, and that the loop correction can be computed explicitly on 
this geometry. These considerations apply as well to the inverse problem, which consists in to learn the model.
We show that the aforementioned property of the Kikuchi free energy minimizer can as well be exploited, for the inverse Ising 
problem in particular, in order to learn efficiently the parameters of the model.

The paper is organized as follows: in Section~\ref{sec:CVM} we give a brief introduction on CVM and related GBP 
algorithms. In Section~\ref{sec:gcbp} we specify GBP and the Kikuchi approximation associated to a cycle basis for region definition, 
analyze the Lagrange multiplier structure and propose a mixed region graph, 
which discards all unnecessary constraints. 
Section~\ref{sec:map} details how this framework adapt to the maximum a posteriori probability estimation (MAP) context. 
The problem of choosing a relevant cycle basis is discussed in Section~\ref{sec:cbasis}.
Then Section~\ref{sec:clmsg} is devoted to an efficient computation of messages exchange 
between cycle and links regions which completes 
our generalized cycle based belief propagation (GCBP) formulation for direct inference. Some properties of the 
free energy functional are also discussed at the end of this section.
In Section~\ref{sec:KIC} we reverse the equations of Section~\ref{sec:clmsg} to address the inverse Ising
problem. Finally some numerical tests are presented in Section~\ref{sec:exp} both for the direct and inverse 
inference problem.

\section{Cluster variational method and generalized BP}\label{sec:CVM}
In this Section we give all the necessary material concerning the relation between BP, generalized BP 
and mean-field approximations in statistical physics. Further details and references can be found e.g. in~\cite{pelizzola}.  
\subsection{Belief propagation and the Bethe approximation}
As far as large scale inference is concerned, the Pearl's belief propagation~\cite{Pearl} and related algorithms 
constitute  central tools in  MRF-based inference approaches. The BP algorithm is an iterative algorithm designed 
to solve a set of fixed point equations. Given an MRF, namely a joint  
distribution over a set $\x=\{x_1,x_2\ldots,x_N\}$ of variables endowed with a factorized form 
\[
p(\x) = \prod_{a\in\F}\psi_a(\x_a)\prod_{i\in\V}\phi_i(x_i)
\]
with $\x_a=\{x_i, i\in a\}$, $a\in\F$ a set of factors,
the marginal probabilities associated to each variable and each factor are search in the form
\begin{align*}
\DD b(x_i) &= \frac{1}{Z_i} \phi_i(x_i)\prod_{a \supset i} m_{a\to i}(x_i),\\[0.2cm]
\DD b(\x_a) &= \frac{1}{Z_a}\psi_a(\x_a)\prod_{i\subset a} n_{i\to a}(x_i),
\end{align*}
where the messages $m_{a\to i}$ and $n_{i\to a}$ relating factor to variables and variables to factors 
satisfy the following set of self-consistent equations
\begin{align}
m_{a\to i}(x_i) &= \sum_{\x_a\backslash x_i}\psi_a(\x_a)\prod_{j\in a\backslash i}n_{j\to a}(x_j),\label{eq:bp1}\\[0.2cm]
n_{j\to a}(x_j) &= \phi_j(x_j)\prod_{b\ni j\backslash a}m_{b\to j}(x_j)\label{eq:bp2}.
\end{align}
\begin{figure}
\begin{center}
\resizebox*{!}{3cm}{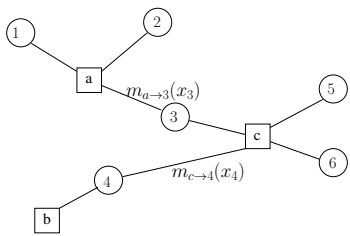}
\end{center}
\caption{Factor graph and message propagation.}\label{fig:fg}
\end{figure}
This algorithm as sketched on Figure~\ref{fig:fg} is exact on a tree, but only approximate on multiply connected factor graphs. When it converges,
it does it empirically in  $O(N\log(N))$ steps on a sparse random graphs, yielding often rather good approximate
marginals.

In~\cite{YeFrWe} was first established the connection between the BP algorithm
of Pearl with a standard mean-field  method - the Bethe approximation~\cite{Bethe} -
used in statistical physics. As is well known in statistical physics, the Gibbs distribution 
associated to the energy function $E(\x)$ and inverse temperature $\beta$, is obtained as a minimizer of the 
free energy functional of a trial distribution $b(\x)$
\begin{align*}
\beta{\cal F}[b]&= \beta E[b]-S[b] = \beta\sum_{\x} b(\x)E(\x)+ \sum_{\x} b(\x)\log\bigl(b(\x)\bigr)\\[0.2cm] 
&= -\log(Z_\text{\tiny Gibbs}) + \sum_{\x} b(\x)\log\frac{b(\x)}{p_\text{\tiny Gibbs}(\x)}\\[0.2cm]
&= -\log(Z_\text{\tiny Gibbs}) + D_\text{\tiny KL}\bigl(b\Vert p_\text{\tiny Gibbs}\bigr)
\end{align*}
as is explicitly seen in the last equality from the non-negativity property of the Kullback Leibler divergence $D_\text{\tiny KL}$.
The mean energy term $E[b]$ can be expressed exactly in terms of marginal distribution obtained from $b$, like e.g. single and pairwise marginals
if $E(\x)$ decomposes over pairwise terms. Instead the entropy term $S[b]$ is in general intractable and mean field methods in 
statistical physics generally correspond to different ways to approximate this term. The Bethe approximation for instance 
corresponds to  
\begin{align*}
S[b]  \approx S_\text{\tiny Bethe} &\egaldef -\sum_i b_i(x_i)\log\bigl(b_i(x_i)\bigr)-
\sum_a b_a(\x_a)\log \frac{b_a(\x_a)}{\prod_{i\in a} b_i(x_i)}\\[0.2cm]
&= \sum_i S_i + \sum_a \Delta S_a,
\end{align*}
i.e. as a sum of individual entropy of each variables, corrected by mutual information among group of variables indexed by $a$. 
The connection with BP is precisely that a BP fixed point of~(\ref{eq:bp1},\ref{eq:bp2}) corresponds to a stationary point of  
the approximate Bethe free energy complemented with compatibility constraints among marginal probabilities
\[
\beta{\cal F}_\text{\tiny Bethe}[b] = \beta E[b]-S_\text{\tiny Bethe}[b] +
\sum_{a\in\F,i\in a\atop x_i}\lambda_{ai}(x_i)\bigl(b_i(x_i)-\sum_{\x_a\backslash x_i}b_a(\x_a)\bigr) 
\]
with help of Lagrange multipliers $\lambda_{ai}(x_i)$. The BP algorithm actually corresponds to performing 
the dual optimization with log messages in~(\ref{eq:bp1},\ref{eq:bp2}) corresponding to an invertible linear transformation 
of the Lagrange multipliers, 
\begin{align}
\lambda_{ai}(x_i) &= \log\bigl(n_{i\to a}(x_i)\bigr),\label{eq:Lnmap} \\[0.2cm]
\log\bigl(m_{a\to i}(x_i)\bigr) &= \frac{1}{d_i-1}\sum_{b\ni i}\lambda_{bi}(x_i) - \lambda_{ai}(x_i),\label{eq:mLmap}
\end{align}
with $d_i$ the number of factors containing $i$.
Moreover, as shown in~\cite{Heskes4} a stable fixed point corresponds to a local minimum of the free energy functional. 

\subsection{Kikuchi approximation and associated message passing algorithms}
In fact as observed in~\cite{Kikuchi,Morita90}, the Bethe approximation is only the first stage
of a systematic entropy cumulant expansion over a poset $\{\alpha\}$ of clusters
\[
S = \sum_\alpha \Delta S_{\alpha},
\]
where $\Delta S_{\alpha}$ is the entropy correction delivered by the cluster $\alpha$ w.r.t. the entropy 
of all its subclusters.
The decomposition is actually valid at the level of each cluster, 
such that with help of some M\"obius inversion formula, the corrections  
\[
\Delta S_\beta 
= \sum_{\alpha\subseteq\beta}  \mu(\alpha,\beta)\ S_\alpha.
\]
and subsequently the full entropy  can be expressed as a weighted sum 
\[
S = \sum_\alpha \kappa_\alpha S_\alpha
\]
of individual cluster entropy 
\[
S_\alpha = -\sum_{\x_\alpha}b_\alpha(\x_\alpha)\log b_\alpha (\x_\alpha),
\]
with $\kappa_\alpha\in{\mathbb Z}$ a set of counting number.  
For example on the $2$D square lattice, the  Kikuchi approximation amounts to retain as cluster the set of nodes $v\in\V$, of links  
$\ell\in\E$ and of square plaquettes $c\in\Cy$ such that on a periodic lattice the corresponding approximate entropy reads 
\[
S = \sum_{c} S_c - \sum_{\ell} S_\ell +\sum_v S_v.
\]
In the CVM, the choice of constraints maybe arbitrary, as long as the clusters hierarchy is  
closed under intersection.

Once identified, the connection between the Bethe approximation and BP leads Yedidia et al. to 
propose in~\cite{YeFrWe} a generalization to BP as an algorithmic counterpart to CVM.
In fact they introduce a notion of region with relaxed constrained w.r.t. the notion of cluster  
used in CVM. In their formulation, any region $R$ containing a factor $a$ 
should contain all variable nodes attached to $a$ in order to be valid. The approximate free energy functional associated to 
a set of region is given by 
\[
\F(b) = \sum_{R\in\R}\kappa_R \F_R(b_R) +\sum_{R'\subseteq R}\sum_{\x_{R'}}\lambda_{RR'}(\x_{R'})
\bigl(b_{R'}(\x_{R'})-\sum_{\x_R\backslash \x_{R'}}b_R(\x_R)\bigr),
\]
with resp. $b_R(\x_R)$ and $\kappa_R$  resp. the  marginal probability  and counting number 
associated to region $R$. The $\lambda_{RR'}$ are again Lagrange multipliers enforcing the constraints 
among regions beliefs. 
The only constraint for the 
counting numbers is that for any variable $i$ or node $a$
\[
\sum_{R\ni i} \kappa_R= \sum_{R\ni a} \kappa_R = 1.
\]
This insures the exactness of the mean energy contribution $E(b)$ to the free energy in general as well as the entropy 
term for uniform distributions in particular. By comparison, there is no freedom in the CVM on the choice of the 
counting numbers once the set of cluster is given.
Additional desirable constraints on the counting numbers are (i) the maxent-normal constraint and 
a (ii) global unit sum rule for counting numbers, 
\begin{equation}\label{eq:sumrule}
\qquad\sum_{R\in\R} \kappa_R = 1.
\end{equation}
Condition (i) means that the approximate region based entropy reaches its maximum for uniform distribution.
Condition (ii) insures exactness of the entropy estimate for perfectly correlated distributions.
As for belief propagation, a set of compatibility constraints among beliefs are introduced 
with help of Lagrange multipliers and generalized belief propagation again amounts to solve 
the dual problem after a suitable linear transformation of Lagrange multipliers hereby defining the messages.
Once a fixed point is found a reparameterization property of the joint measure holds:
\[
P(\x) \propto \prod_{R\in\R} b_R(\x_R)^{\kappa_R}.
\]
When the region graph has no cycle, this factorization involves the true marginals probabilities 
of each region and is exact.

There is some degree of freedom both in the initial choice of Lagrange multipliers and 
messages leading to different algorithm without changing the free energy and associated 
variational solutions. A canonical choice is to connect regions only to their direct ancestor
or direct child regions leading to the parent-to-child algorithm. There is still
in this choice some redundancy in the constraints, some linear dependencies among those, which
can potentially affect the convergence of the algorithm by adding unnecessary loops in the factor graph.
This problem has been addressed in~\cite{PaAn} where for a given region set a construction for a minimal factor graph is 
proposed. 

\subsection{Main contributions}
GBP is a framework corresponding to a wide class of algorithms, which upon a good choice 
of regions can lead to much accurate results than basic BP. Its systematic use 
is however made delicate by the following unsolved issues as far as large scale inference is concerned: 
\begin{itemize}
\item there is no automatic and efficient procedure of choosing the regions able to scale with large scale problems 
for non-regular factor-graph, despite proposals like the region pursuit algorithm~\cite{Welling2004} which potential use seems
however limited to small size systems. 
\item without special care the computational cost grows exponentially w.r.t. region size.  
\item there are difficult convergence problems associated to GBP which have led to consider double loop algorithms~\cite{Yuille,HeAlKa}
at the price of additional computational burden.
\end{itemize}
Concerning inverse problems, we are not aware of any method in the family of region based approximation of  the log likelihood,
going beyond the Bethe approximation  at the exception  
of the exact method proposed in~\cite{CoMo}, which is however 
limited to small systems size from the practical point of view.

The idea of constructing the region graph from a cycle basis is not new, it is already present as a special case 
of CVM in~\cite{Kikuchi} and was first formally proposed in~\cite{WelMinTeh2005} and refined in~\cite{GeWe}, regarding the
choice for the cycle basis, without however explicitly addressing large scale issues listed above.
Our contributions in this context is to settle a certain number of technical problems regarding this construction in order to address the 
above restrictions such that large scale problems can be treated by means of two algorithms GCBP and KIC respectively for direct and inverse
pairwise MRF inference. More specifically,
\begin{itemize}
\item we address convergence problems by proposing a specific construction of the factor graph in Section~\ref{sec:mfg} based on 
some decomposition of single variable counting numbers unraveled in Section~\ref{sec:ldc};
\item our construction leads to a linear cost w.r.t.  region size i.e. large cycles,  instead  of exponential in general as detailed
in~\ref{sec:clmsg};
\item our region graph construction as discussed in Section~\ref{sec:cbasis} relies on a minimal cycle basis optimization, 
which to some extent and thanks to some approximate algorithm
can scale-up to relatively large size as seen experimentally in Section~\ref{sec:exp};
\item we propose in Section~\ref{sec:KIC} a general inverse pairwise MRF method based on the Kikuchi approximation 
which scales linearly w.r.t. system size, once a
cycle basis is given or properly guessed, again without any limitation in cycle's sizes.
\end{itemize}

\section{Generalized cycle based BP (GCBP)}\label{sec:gcbp}
The first motivation for attaching regions to the elements of a cycle basis originate in the observation 
that the Bethe approximation violates the ``global unit sum rule''~(\ref{eq:sumrule}) for counting numbers,
except on singly connected graphs, precisely by an amount corresponding to the cyclomatic number of the graph.
Completing the regions set with elements of a cycle basis restores the unit sum rule property~\cite{WelMinTeh2005}. 

A different motivation comes from statistical physics considerations associated to the duality transformation~\cite{Savit}
which can be performed with certain restrictions on the models like e.g. the Ising model without external fields.
In such cases, one is naturally led to consider a dual belief propagation on the dual graph which nodes 
correspond to the element of a cycle basis~\cite{Fu2013}. The extension of such consideration to arbitrary pairwise
models led us to consider GBP based on such cycle basis.

\subsection{Cycle based Kikuchi approximation}
To set up notations, we consider a pairwise MRF of $n$ random 
variables valued in some arbitrary subset $\x=\{x_1,\ldots,x_n\}\in\I_1\times\ldots\I_n\subset{\mathbb R}^n$,
specified by some undirected graph $\G=(\V,\E)$, with vertex set $\V=\{1,\ldots,n\}$ and edge set 
$\E\subset\V\times\V$. To simplify we also assume $\G$ to be connected.
The reference distribution considered to be pairwise, is of the form
\begin{equation}\label{eq:refP}
P(\x) = \prod_{\ell\in\E}\psi_\ell^0(\x_\ell)\prod_{v\in\V}\phi_v(x_v).
\end{equation}
By definition a cycle of $\G$ is an unoriented  subgraph where each node has an even degree. The set of cycles is 
a vector space over ${\mathbb Z}_2$ of dimension $|\E|-|\V|+1$ for a graph with one single component which is assumed from now on. This means 
that when two cycles are combined, edges are counted modulo $2$.
Examples of cycle basis are shown on Figure~\ref{fig:cyclebasis}. For heterogeneous graphs, a simple way 
to generate a basis consists in first to select a spanning tree of the graph and to associate a cycle 
to each of the $|\E|-|\V|+1$ remaining links of the graph, by adding to each one  the path
on the spanning tree joining the two ends of the link. This yields by definition a fundamental cycle basis, associated to 
the considered spanning tree.
\begin{figure}
\centerline{\resizebox*{0.9\textwidth}{!}{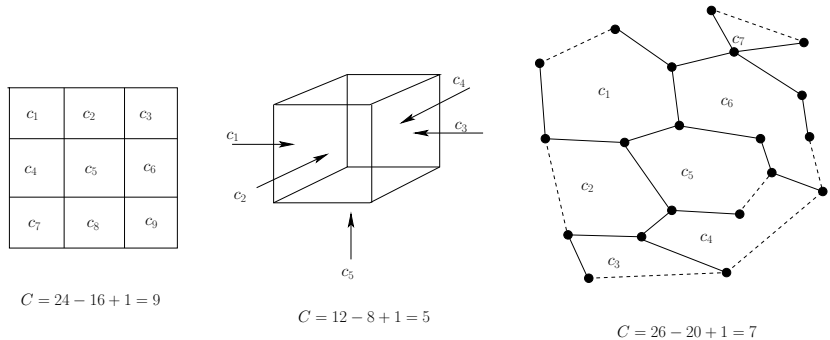}}
\caption{\label{fig:cyclebasis} Example of cycle basis on $2$-D and $3$-D lattices and fundamental cycle basis on an arbitrary graph.}
\end{figure}
Let us assume that a cycle basis of $\G$ is given with cycles indexed by $c\in\Cy =\{1,\ldots |\Cy|\}$. $|\Cy|=|\E|-|\V|+1$ also called 
the cyclomatic number represents the number of independent loops of $\G$.
In the Kikuchi CVM approximation that we consider, the maximal clusters are associated to each element 
of the cycle basis and possibly links which are not contained in any basic cycle.
We assume also that one cycle has at most one edge in common with any other cycle. If this is not the 
case then one edge and one cycle can be added to $\G$ in order to restore this property, for each set of cycles 
having a common group of edges in common (see Figure~\ref{fig:dual_graph}). 
Disconnected intersections can be eliminated by a proper choice of cycle basis.
As explained in Section~\ref{sec:CVM} all mean-field type approximations underlying BP or GBP, consists in assuming a factorized form of the
joint measure in term of some of its marginal distributions.  
Within the CVM and given our choice for the maximal cluster this leads to assuming the following factorization of 
the joint measure:
\begin{equation}\label{eq:gbp}
P_\text{\tiny GBP}(\x) = \prod_{c\in \Cy} p_c(\x_c)\prod_{\ell\in \E} p_\ell(\x_\ell)^{\kappa_\ell}\prod_{v\in\V}p_v(x_v)^{\kappa_v},
\end{equation}
where $p_c$, $p_\ell$ and $p_v$ are marginal probabilities respectively associated to cycles, links and single variables.
As we shall see, and this is an important observation for what follows, the
probability $p_c$ associated to a cycle can be itself expressed as a pairwise MRF, with 
each factor corresponding to one edge of the cycle:
\begin{equation}\label{eq:pcycle}
p_c(x_c) = \prod_{\ell\in c}\varphi_\ell(x_\ell).
\end{equation}
In (\ref{eq:gbp}) the choice of the counting number for respectively cycles, edges and vertices are
$\kappa_c=1$, $\kappa_\ell=1-d^\star_\ell$ and $\kappa_v=1-\sum_{c\ni v}\kappa_c -\sum_{\ell\ni v}\kappa_\ell $. 
$d^\star_\ell$ is the number of cycles in $\Cy$ containing edge $\ell$.
This choice is in accordance to general CVM prescriptions, as being dictated
by the constraint that each degree of freedom is counted exactly once in the Kikuchi free energy.
As already said, thanks to these rules the global unit sum rule for counting numbers is automatically satisfied:
\[
\sum_{c\in\Cy}\kappa_c+\sum_{\ell\in\E} \kappa_\ell+\sum_{v\in\V}\kappa_v = |\Cy|-|\E|+|\V|=1.
\]
A dual bipartite graph $\G^\star = (\V_c^\star,\V_t^\star,\E^\star)$ can be defined, where $\V^\star$ 
indexes the cycle basis, and elements of $\V_t^\star$ represent connected intersection between cycles, i.e. either 
single nodes, links or sub-trees corresponding to bridges connecting distant cycles. Elements of $\E^\star$
connect intersecting elements of $\V_c^\star$ and $\V_t^\star$.
\begin{figure}[ht]
\centerline{\resizebox*{0.7\textwidth}{!}{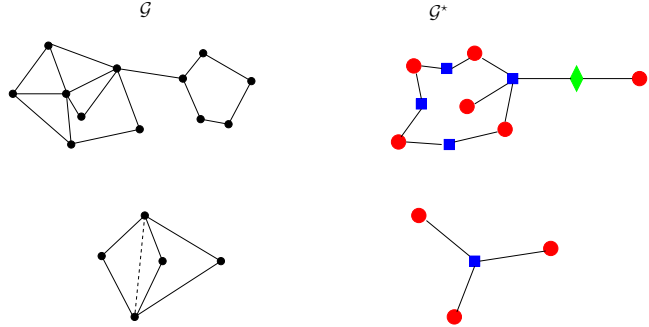}}
\caption{\label{fig:dual_graph} Dual graph construction. Dashed link correspond to one virtual added link.}
\end{figure}
Under this assumption we have the following important property, illustrated on Figure~\ref{fig:polygontree} 
which justifies the approximation (\ref{eq:gbp},\ref{eq:pcycle}).
\begin{prop}\label{prop:dtree}
If $\G^\star$ is acyclic, the factorization~\ref{eq:gbp} is exact.
\end{prop}
\begin{proof}
See Appendix~\ref{app:prop1}.
\end{proof}
\begin{figure}[ht]
\centerline{
\includegraphics[width=\textwidth]{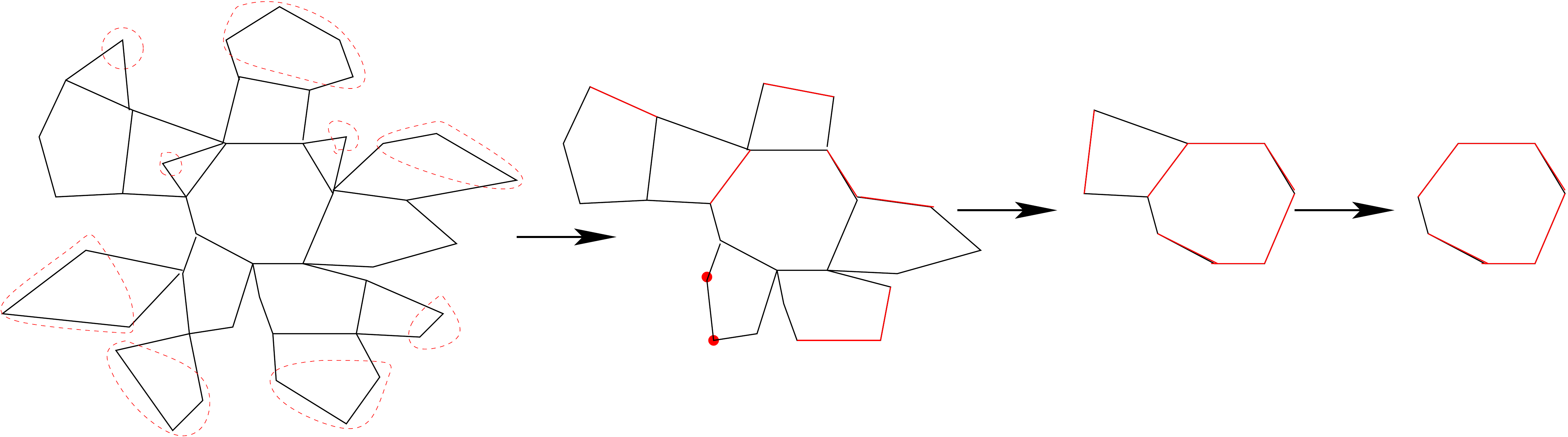}}
\caption{\label{fig:polygontree} Successive graphical models obtained by deconditioning variables (circled in red) 
from the leaves, starting from a polygon tree. Factors corresponding to links or vertices in red 
are modified during the process.}
\end{figure}
The variational problem that GBP aims at solving, is to find the closest distribution of the form~(\ref{eq:gbp})
to the reference distribution~(\ref{eq:refP}). 
For later convenience we define
\[
\psi_\ell(\x_\ell) \egaldef \psi_\ell^0(\x_\ell)\prod_{v\in\ell}\phi_v(x_v),
\]
and also introduce for any $c\in\Cy$:
\begin{equation}\label{eq:cbp}
\Psi_c(\x_c)\egaldef \prod_{\ell\in c}\psi_\ell(\x_\ell)\prod_{v\in c}\phi_v(x_v).
\end{equation}
For a candidate measure $p$, the variational free energy functional reads
\begin{align}
&\F(P_\text{\tiny GBP}||P) = \sum_{c\in\Cy,\x_c} p_c(\x_c)\log\frac{p_c(\x_c)}{\Psi_c(\x_c)}
 +\sum_{\ell\in\E,\atop\x_\ell}\kappa_\ell\ p_\ell(\x_\ell)\log\frac{p_\ell(\x_\ell)}{\psi_\ell(\x_\ell)}\nonumber\\[0.1cm]
&+\sum_{v\in\V,\atop x_v}\kappa_v\ p_v(x_v)\log\frac{p_v(x_v)}{\phi_v(x_v)}
+\sum_{\ell,c\ni\ell,\x_\ell}\lambda_{c\ell}(\x_\ell)\bigl(p_\ell(\x_\ell)-\sum_{\x_c\backslash\x_\ell}p_c(\x_c)\bigr)\nonumber\\[0.1cm]
&+\sum_{v,\ell\ni v,x_v}\lambda_{\ell v}(x_v)\bigl(p_v(x_v)-\sum_{\x_\ell\backslash\x_v}p_\ell(\x_\ell)\bigr)
+\sum_{v,c\ni v,x_v}\lambda_{c v}(x_v)\bigl(p_v(x_v)-\sum_{\x_c\backslash\x_v}p_c(\x_c)\bigr)\label{eq:glagrange}
\end{align}
after introducing three sets of Lagrange multipliers, $\lambda_{c\ell}(\x_\ell)$, $\lambda_{\ell v}(x_v)$ and 
$\lambda_{c v}(x_v)$ to enforce respectively cycle-edge, edge-variable and cycle-variable marginals compatibility.  
The minimum of the free energy is then obtained as:
\[
\begin{cases}
\DD p_c(\x_c) \propto \Psi_c(\x_c)\exp\bigl[\sum_{\ell\in c}\lambda_{c\ell}(\x_\ell)+\sum_{v\in c}\lambda_{cv}(x_v)\bigr]\\[0.5cm]
\DD p_\ell(\x_\ell) \propto \psi_\ell(\x_\ell)\exp\bigl[\frac{1}{\kappa_\ell}\bigl(\sum_{v\in\ell}\lambda_{\ell v}(x_v)-
\sum_{c\ni\ell}\lambda_{c\ell}(\x_\ell)\bigr)\bigr]\\[0.5cm]
\DD p_v(x_v) \propto \phi_v(x_v)\exp\bigl[-\frac{1}{\kappa_v}
\bigl(\sum_{c\ni v}\lambda_{c v}(x_v)+\sum_{\ell\ni v}\lambda_{\ell v}(x_v)\bigr)\bigr]
\end{cases}
\]
As direct consequence of these expressions we have 
\begin{coro}
$p_c$ has the form~\ref{eq:pcycle}.
\end{coro} 
\subsection{Single variable counting numbers and dual loops}\label{sec:ldc}
The counting number $\kappa_v$ contains some information about the local structure
of the dual graph. In order to unravel it we define the local dual graph $\G_v^\star\subset \G^\star$ attached to 
$v$ as $\G_v^\star = (\V_{v;c}^\star,\V_{v;t}^\star,\E_v^\star)$, where $\V_{v;c}^\star$ are dual vertices corresponding to 
cycles containing $v$; $\V_{v;t}^\star$ are dual vertices corresponding to all edges containing $v$ with non-zero counting number; 
\begin{figure}[ht]
\centerline{\resizebox*{\textwidth}{!}{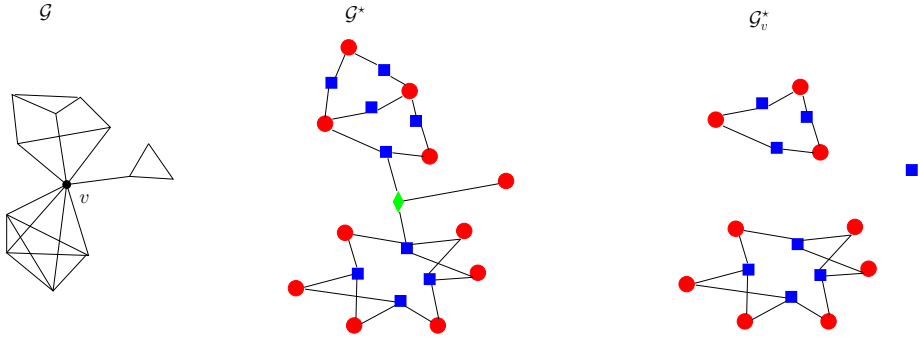}}
\caption{\label{fig:locdual_graph} Local dual graph construction. In this case the choice of cycle basis leads to 
$\kappa_v = 2$ with $d_v^\star = 3$ and $\Cy_v^\star = 4$.}
\end{figure}
$\E_v^\star$ is the set of dual 
edges connecting $\ell$-nodes in $\V_{v;t}^\star$ to their corresponding $c$-nodes in $\V_{v;c}^\star$ they belong to in the primal graph.
\begin{prop}\label{prop:ldc}
Let $d_v^\star$ be the number of components of $\G_v^\star$ and $\Cy_v^\star$ its cyclomatic number. We have
\begin{equation}\label{eq:ldc}
\kappa_v = 1-d_v^\star +\Cy_v^\star. 
\end{equation}
\end{prop}
\begin{proof}
By definition, we have
\begin{align*}
\Cy_v^\star &= \vert \E_v^\star \vert - \vert \V_{v;c}^\star\vert -\vert\V_{v;t}^\star\vert + d_v^\star\\[0.2cm]
&= \sum_{\ell\ni v}d_{\ell}^\star - \sum_{c\ni v} 1 - \sum_{\ell\ni v} 1 +  d_v^\star\\[0.2cm]
&= \kappa_v + d_v^\star -1.
\end{align*}
where between the first and second line it is remarked that for any $\ell$ parent of $v$, any $c$ parent of $\ell$
necessarily contains $v$.
\end{proof}
Qualitatively $\Cy_v^\star$ represents the number of dual cycles ``centered'' on $v$. 
This decomposition will prove useful for building our cycle based region graph. 

Let us give a few examples: for nodes in the bulk of a planar graph we have $\Cy_v^\star=1$, on a cubic lattice $\Cy_v^\star=3$ which generalizes to 
$\Cy_v^\star = d(d-1)/2$ on a $d$-dimensional square lattice.
On a $N/2+N/2$ bipartite  graph we have $\Cy_v = 3N/2-1$ while on a complete graph of size $N$, using a cycle 
basis $\{(1ij),1<i<j\le N\}$ rooted on node $1$ , $\Cy_v = (N-2)(N-3)/2$.

\subsection{Parent-to-child algorithm and minimal graphical representation}
At this point, following the region-based algorithm~\cite{YeFrWe} prescriptions, a message
passing algorithm can be set-up which rules are associated to the Hasse diagram of the regions hierarchy.
Regions are associated to all terms with non vanishing counting number in~(\ref{eq:gbp}), 
and directed edges are associated to each Lagrange multiplier added in~(\ref{eq:glagrange}), corresponding to 
direct parent to child relationship, hence discarding the $\lambda_{cv}$.
The message rules which are obtained are then based on the existence of a certain linear transformation 
of the Lagrange multipliers, which allows one to parameterize the beliefs as follows
\begin{align}
p_v(x_v) &= \phi_v(x_v)\prod_{\ell\ni v} m_{\ell\to v}(x_v),\nonumber\\[0.2cm]
p_\ell(\x_\ell) 
&= \psi_\ell(\x_\ell)\prod_{c\ni\ell}m_{c\to\ell}(\x_\ell)\prod_{v\in\ell}n_{v\to\ell}(x_v),\nonumber\\[0.2cm]
p_c(\x_c) &= \Psi_c(\x_c)\prod_{\ell\in c}n_{\ell\to c}(\x_\ell)\prod_{v\in c}n_{v\to c}(x_v),\label{eq:pc1}
\end{align}
with 
\begin{align*}
n_{\ell\to c}(\x_\ell) &\egaldef \prod_{c'\ni\ell\backslash c} m_{c'\to\ell}(\x_\ell)\\[0.2cm]
n_{v\to \ell}(x_v) &\egaldef 
\prod_{\ell'\ni v\backslash \ell}m_{\ell' \to v}(x_v),\\[0.2cm]
n_{v\to c}(x_v) &\egaldef \prod_{\ell'\ni v,\ell'\notin c}m_{\ell' \to v}(x_v),
\end{align*}
From this we get the following message passing rules:
\begin{align}
m_{c\to\ell}(\x_\ell)\prod_{v\in\ell}m_{\ell_{vc\backslash\ell}\to v}(x_v) 
&\longleftarrow \sum_{\x_c\backslash \x_\ell}\frac{\Psi_c(\x_c)}{\psi_\ell(\x_\ell)}
\prod_{\ell'\in c\backslash\ell}n_{\ell'\to c}(\x_{\ell'})
\times\prod_{v\in c\backslash\ell}n_{v\to c}(x_v),\label{eq:pac1}\\[0.2cm]
m_{\ell\to v}(x_v) &\longleftarrow \sum_{\x_\ell\backslash x_v}\frac{\psi_\ell(\x_\ell)}{\phi_v(x_v)}
\prod_{c\ni\ell}m_{c\to\ell}(\x_\ell)\prod_{v'\in \ell\backslash v} n_{v'\to\ell}(x_{v'}),\label{eq:pac2}
\end{align}
where in the first rule the shorthand notation $\ell_{v c\backslash\ell}$ is used to denote the 
link in $c$ containing $v$ different from $\ell$.

As noticed in~\cite{PaAn}, dependence between Lagrange multipliers are present in the parent-to-child algorithm. 
This results in more complex factor graph with more feed-back loops
than necessary which in turn may cause convergence failures of GBP.
In effect we observe experimentally, both on grids and on 
heterogeneous graphs tested in Section~\ref{sec:exp} that the parent-to-child algorithm fails to converge
for systems sizes exceeding a few hundreds of nodes whatever damping coefficient is inserted into the message passing equations.
In~\cite{PaAn} a minimal graphical representation construction is proposed to settle such problems, 
in order to eliminate all redundant Lagrange multipliers. In our setting this leads in particular to 
having any (non-bridge) variable node to be attached to at most one link node and to have therefore at most one ancestor cycle node
in the factor graph. As a consequence we have always $n_{v\to c}(x_v) = m_{\ell_{vc\backslash\ell}\to v}(x_v)= 1$. 
As shown in Appendix~\ref{app:instabil} this 
leads to an essentially unstable algorithm for graph containing at least one single dual loop. So in short  
we have 
\begin{itemize}
\item poor global convergence properties of the parent to child algorithm;
\item local convergence problems for minimal region graph based algorithm caused by dual loops.
\end{itemize}
This problem of redundant Lagrange multipliers 
has actually also been discussed in the context of the $2$-D Edward Anderson (EA) model in~\cite{DoLaMuRiTo}. 
In this context the authors propose a solution based on a specific gauge choice for the message definition in order
to regularize GBP. Our approach to this problem is different. As we shall see in the next Section it 
is solely based on topological properties of the graph of interactions. This yield a generic method 
independent of the graph or the type of interactions.

\subsection{Mixed factor graph and associated message passing rules}\label{sec:mfg}
We introduce here a specification of the region graph which on the one hand eliminates 
all unnecessary feed-back loops present in the parent-to-child algorithm, 
but on the other 
hand prevent instabilities associated to dual loops. 
In this formulation first a minimal set of Lagrange multipliers are taken into account as proposed in~\cite{PaAn};
but additional ``clone variables'' need to be introduced for variables at the center of dual loops, i.e. for which 
$C_v^\star\ne 0$, as defined in Section~\ref{sec:ldc}, to prevent some instability which we have identified (see Appendix~\ref{app:instabil}).
Before explaining it in details let us give the specification of the region graph which we refer to as the mixed factor graph (MFG)
for reasons which will soon be clear:\\
\begin{itemize}[leftmargin=*,noitemsep,topsep=0pt]
\item (i) Each term in~(\ref{eq:gbp}) having a non-zero counting number is associated to a node in the MFG.
There are three families of nodes, $c$-nodes, $\ell$-nodes and $v$-nodes, respectively associated to cycles, links and 
vertices of the original graph. 
$c$-nodes are always factors while $v$-nodes are always variables. Instead, 
$\ell$-nodes associated to links are composite nodes, i.e. can be of both types. 
\item (ii) Edges of the MFG represent Lagrange multipliers and relate variables to factors. A $v$-node can be linked to 
$\ell$-nodes, considered then as factors nodes. $\ell$-nodes considered as variable nodes
can be linked to $c$-nodes.
\item (iii) all links of a given cycle $c$ with non-vanishing counting numbers are linked as variables to this 
$c$-node.
\item (iv) to a variable $v$ we associate in general two types of $v$-nodes depending on  $d_v^\star$ and $\Cy_v^\star$ defined in 
Section~\ref{sec:ldc}:
\begin{itemize}
\item (a) if $d_v^\star> 1$ one $v$-node is associated to $v$, which connects exactly to one single arbitrary $\ell$-node of each components of 
$\G_v^\star$, its degree being therefore $d_v^\star$ and a counting number of $1-d_v^\star$ is attributed to it. 
If necessary an $\ell$-node with zero counting number can be inserted into the MFG
in order to ensure this $v$-node to be properly connected to all components it owes to.
\item (b) if $\Cy_v^\star >0$, to each $\ell$ containing $v$ we associate one 
singly connected to $\ell$ $v^\star$-nodes, 
as long as these $\ell$-node are in a component of $\G_v^\star$ containing at least one dual loops. Each clone is attributed a counting
number $\kappa_{\vs} = \Cy_v^\star/q$ if $q$ is the number of clones.
\end{itemize}
\end{itemize}
This set of rules is illustrated on Figure~\ref{fig:hfg}.
Rule (iii) ensures that all marginal probabilities of cycles are compatibles at links intersections. Rule (iv)(a)
is applied to cut-vertices, i.e. vertices which separate $\G$ in multiple components when removed as shown 
on the example of Figure~\ref{fig:hfg}. Rule (iv)(b) is there to take into account dual loop corrections. 
The prescription (iv)(b) is such as to ensure a better convergence of GCBP by making use of replicas of $v$-nodes, while
preserving the minimal use of Lagrange multipliers. The number of constraints is still  minimal in the sense that 
the number of independent loops of the MFG is equal to the number of independent loops of the dual graph $\G^\star$.
From the Lagrangian formulation $\kappa_\vs$ is constrained by 
$\sum_{\vs\approx v}\kappa_\vs = \Cy_v^\star$ 
where $\approx$ indicates the correspondence between $\vs$-node and variable $v$.
The choice made in rule (iv)b for $\kappa_{vs}$ satisfies this constraint, albeit other ones are possible. 
\begin{figure}[ht]
\centerline{\resizebox*{0.7\textwidth}{!}{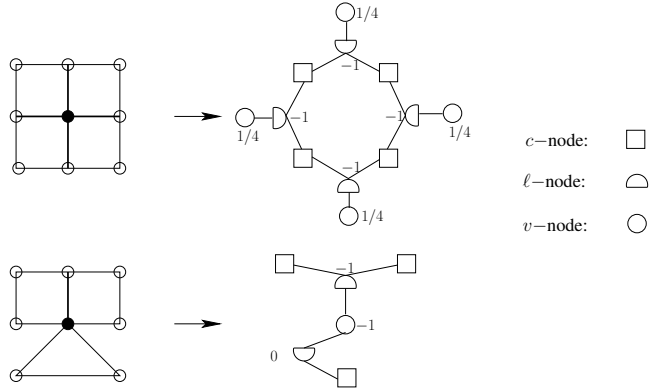}}
\caption{\label{fig:hfg} Pairwise MRF (left). Variables and links with 
non-zero counting number are in bold. Corresponding mixed factor graph(right) with counting numbers.}
\end{figure}

The reason for introducing clone variables becomes clearer when trying to write down message passing equations. 
In fact a direct generalization of the change of variables~(\ref{eq:Lnmap},\ref{eq:mLmap}) used to define ordinary 
BP from the Lagrange multipliers can be obtained as follows:
\begin{align*}
\lambda_{\ell v}(x_v) &= \log \prod_{\ell'\ni v\backslash \ell}m_{\ell' \to v}(x_v) \egaldef \log n_{v\to\ell}(x_v),\\[0.2cm]
\lambda_{\ell\vs}(x_v) &= -\kappa_\vs\log m_{\ell \to \vs}(x_v) \egaldef \log n_{\vs\to\ell}(x_v),\\[0.2cm]
\lambda_{c\ell}(\x_\ell) &= \log n_{\ell\to c}(\x_\ell) +\sum_{v\in\ell}\log n_{v\to\ell}(x_v),
\end{align*}
where $\sum_{v\in\ell}$ is taken over all types of $v$-nodes
and with 
\[
n_{\ell\to c}(\x_\ell) \egaldef \prod_{c'\ni\ell\backslash c} m_{c'\to\ell}(\x_\ell),\\[0.2cm]
\]
Note that $\lambda_{cv}$ have disappeared by definition of the MFG.
From the Lagrangian formulation $\kappa_\vs$ is constrained by 
\[
\sum_{\vs\approx v}\kappa_\vs = \Cy_v^\star,
\] 
where $\approx$ indicates the correspondence between $\vs$-node and variable $v$.
The choice made in rule (iv)b for $\kappa_{vs}$ satisfies this constraint, albeit other ones are possible. 
We get the following expression for the beliefs
\begin{align}
p_v(x_v) &=  \phi_v(x_v)\exp\bigl[-\frac{1}{1-d_v^\star}\sum_{\ell\ni v}\lambda_{\ell v}(x_v)\bigr]
= \phi_v(x_v)\prod_{\ell\ni v} m_{\ell\to v}(x_v),\nonumber\\[0.3cm]
p_\vs(x_v) &=  \phi_v(x_v)\exp\bigl[-\frac{1}{\kappa_\vs}\lambda_{\ell_\vs\vs}(x_v)\bigr]
= \phi_v(x_v) m_{\ell_\vs\to \vs}(x_v),\nonumber\\[0.3cm]
p_\ell(\x_\ell) &=  \psi_\ell(\x_\ell)\exp\bigl[\frac{1}{\kappa_\ell}\bigl(\sum_{v\in\ell}\lambda_{\ell v}(x_v)-
\sum_{c\ni\ell}\lambda_{c\ell}(\x_\ell)\bigr)\bigr]
= \psi_\ell(\x_\ell)\prod_{c\ni\ell}m_{c\to\ell}(\x_\ell)\prod_{v\in\ell}n_{v\to\ell}(x_v),\nonumber\\[0.3cm]
p_c(\x_c) &=  \Psi_c(\x_c)\exp\bigl[\sum_{\ell\in c}\lambda_{c\ell}(\x_\ell)\bigr]
= \Psi_c(\x_c)\prod_{\ell\in c}\bigl[n_{\ell\to c}(\x_\ell)\prod_{v\in\ell}n_{v\to\ell}(x_v)\bigr],\label{eq:pc}
\end{align}
where $\ell_\vs$ denotes the $\ell$-node connected to $\vs$. From this we get the following message passing rules:
\begin{align}
m_{c\to\ell}(\x_\ell) &\longleftarrow \sum_{\x_c\backslash \x_\ell}\frac{\Psi_c(\x_c)}{\psi_\ell(\x_\ell)}
\prod_{\ell'\in c\backslash\ell}\bigl[n_{\ell'\to c}(\x_{\ell'})\prod_{v\in\ell'}n_{v\to\ell'}(x_v)\bigr],\label{eq:mcl}\\[0.2cm]
m_{\ell\to v}(x_v) &\longleftarrow\sum_{\x_\ell\backslash x_v}\frac{\psi_\ell(\x_\ell)}{\phi_v(x_v)}
\times\prod_{c\ni\ell}m_{c\to\ell}(\x_\ell)\prod_{v'\in \ell\backslash v} n_{v'\to\ell}(x_{v'}),\label{eq:mlv}\\[0.2cm]
m_{\ell\to \vs}(x_v) &\longleftarrow\Bigl(\sum_{\x_\ell\backslash x_v}\frac{\psi_\ell(\x_\ell)}{\phi_v(x_v)}
\times\prod_{c\ni\ell}m_{c\to\ell}(\x_\ell)\prod_{v'\in \ell\backslash \vs} n_{v'\to\ell}(x_{v'})\Bigr)^{1/(1+\kappa_\vs)}.\label{eq:mlvs}
\end{align}
\begin{figure}[ht]
\centerline{\resizebox*{\textwidth}{!}{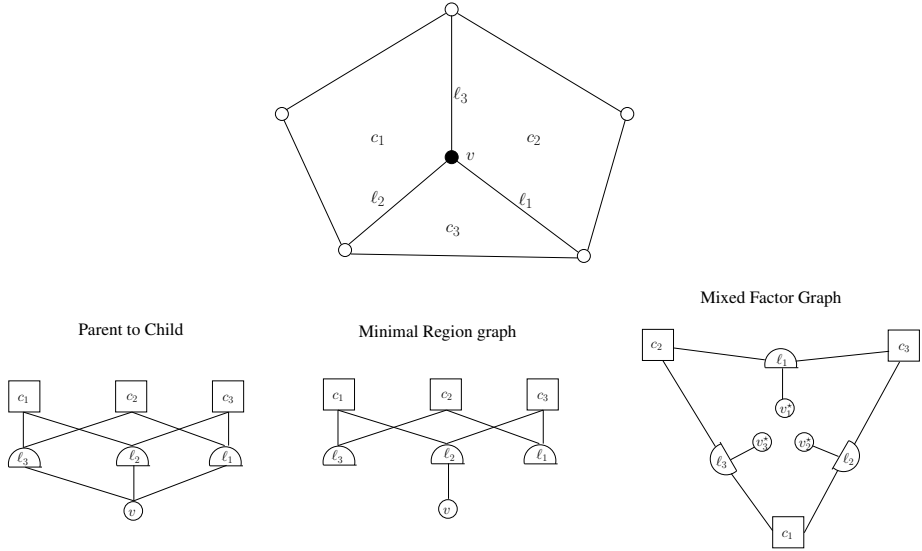}}
\caption{\label{fig:dloop} One dual loop on top ($C_v^\star=1$) with corresponding  factor-graphs.}
\end{figure}
The difference between factor graph of standard parent-to-child algorithm, the minimal one proposed in~\cite{PaAn}   
and the one associated to MFG is illustrated on Figure~\ref{fig:dloop}.
With this formulation GCBP can be 
seen mainly as an ordinary belief propagation defined on the MFG, where~(\ref{eq:mcl},\ref{eq:mlv})
are direct generalization on a MFG of ordinary BP update rules~(\ref{eq:bp1},\ref{eq:bp2}),
with an additional peculiarity of given 
by dual loop corrections carried by clone variables in~(\ref{eq:mlvs}).

\subsection{MAP estimation}\label{sec:map}
The general inference schema proposed in the previous sections can be straightforwardly adapted 
to the optimization context, the same way as the min-sum algorithm also called belief revision~\cite{Pearl}
is derived from BP, by simply replacing ``$\sum$'' by ``$\min$'' (see e.g.~\cite{Ruozzi}) . 
First, adding some specific notations, the messages are parameterized in terms of log probability ratio: 
\[
m_{c\to\ell}(\x_\ell) \propto \exp\bigl(-\mu_{c\to\ell}(x_\ell)\bigr)\qquad\text{and}\qquad
m_{\ell\to v}(x_v) \propto \exp\bigl(-\mu_{\ell\to v} (x_v)\bigr).
\]
The counterparts to ``$n$'' messages are in turn defined as:
\begin{align*}
\nu_{\ell\to c}(\x_\ell) &\egaldef \sum_{c'\ni\ell\backslash c} \mu_{c'\to\ell}(\x_\ell),\\[0.2cm]
\nu_{v\to\ell}(x_v) &\egaldef \sum_{\ell'\ni v\backslash \ell} \mu_{\ell'\to v}(x_v),\\[0.2cm]
\nu_{\ell\to c}(x_{v^\star}) &\egaldef -\kappa_{v^\star}\mu_{\ell \to v^{\star}}(x_{v^\star}),
\end{align*}
where again clone variable are distinguished from ordinary ones using $\star$ notation.
Correspondingly, let 
\[
E_c(\x_c)\egaldef -\log\bigl(\Psi_c(\x_c)\bigr),\qquad
E_\ell(\x_\ell)\egaldef -\log\bigl(\psi_\ell(\x_\ell)\bigr)
\]
and
\[
E_v(x_v) \egaldef  -\log\bigl(\phi_v(x_v)\bigr).
\]
To the generalized belief propagation rules~(\ref{eq:mcl},\ref{eq:mlv},\ref{eq:mlvs})
correspond the following min-sum update rules:
\begin{align}
\mu_{c\to\ell}(\x_\ell) &\longleftarrow \min_{\x_c\backslash \x_\ell}\Bigl(E_c(x_c)-E_\ell(\x_\ell)+
\sum_{\ell'\in c\backslash\ell}\bigl[\nu_{\ell'\to c}(\x_{\ell'})+\sum_{v\in\ell'}\nu_{v\to\ell'}(x_v)\bigr]\Bigr),\label{eq:mcl_opt}\\[0.2cm]
\mu_{\ell\to v}(x_v) &\longleftarrow\min_{\x_\ell\backslash x_v}\Bigl(E_\ell(\x_\ell)-E_v(x_v)
+\sum_{c\ni\ell}\mu_{c\to\ell}(\x_\ell)+\sum_{v'\in \ell\backslash v} \nu_{v'\to\ell}(x_{v'})\Bigr),\label{eq:mlv_opt}\\[0.2cm]
\mu_{\ell\to \vs}(x_v) &\longleftarrow \frac{1}{1+\kappa_\vs}\min_{\x_\ell\backslash x_v}\Bigl(E_\ell(\x_\ell)-E_v(x_v)\nonumber\\[0.2cm]
&\hspace{4cm}+\sum_{c\ni\ell}\mu_{c\to\ell}(\x_\ell)+\sum_{v'\in \ell\backslash \vs} \nu_{v'\to\ell}(x_{v'})\Bigr).\label{eq:mlvs_opt}
\end{align}
As a result the beliefs  associated to the various family of nodes, expressing log marginal probabilities, are given by
\begin{align}
\Yp_v(x_v) &= E_v(x_v)+\sum_{\ell\ni v}\mu_{\ell\to v}(x_v),\nonumber\\[0.2cm]
\Yp_\ell(\x_\ell) &= E_\ell(\x_\ell)+\sum_{c\ni\ell}\mu_{c\to\ell}(\x_\ell)+\sum_{v\in\ell}\nu_{v\to\ell}(x_v),\nonumber\\[0.2cm]
\Yp_c(\x_c) &= E_c(\x_c)+\sum_{\ell\in c}\bigl[\nu_{\ell\to c}(\x_\ell)+\sum_{v\in\ell}\nu_{v\to\ell}(x_v)\bigr].\label{eq:Ec}
\end{align}
When the messages correspond to a fixed point, the usual compatibility between beliefs is expressed as
\begin{align*}
\min_{\x_c\backslash \x_\ell} \Yp_c(\x_c) &= \Yp_\ell(\x_\ell),\qquad\forall \ell\in c,\\[0.2cm]
\min_{\x_\ell\backslash x_v} \Yp_\ell(\x_\ell) &= \Yp_v(x_v),\qquad\forall v\in\ell.
\end{align*}
In addition, if the joint probability measure is given in a Gibbs form, 
\[
P(\x) = e^{-E(\x)},
\]
these beliefs provide us, up to a constant, with the following decomposition of the energy function:
\[
E(\x) = \sum_c \Yp_c(\x_c)+\sum_\ell \kappa_\ell \Yp_\ell(\x_\ell)+\sum_v\kappa_v \Yp_v(x_v),
\]
and the approximate minimizer of $E(\x)$, given by 
\[
x_i^{min} = \argmin_{x_i} \Yp_i(x_i),\qquad\forall i\in \V, 
\]
verifies
\[
E(\x^{min}) = \sum_c \min_{\x_c}[\Yp_c(\x_c)]+\sum_\ell\kappa_\ell \min_{\x_\ell}[\Yp_\ell(\x_\ell)]+\sum_v\kappa_v \min_{x_v}[\Yp_v(x_v)],
\]
by virtue of the belief's compatibility.
Next, as will be also the case for inference, we exploit the ring geometry 
in order to compute efficiently the $c$-node to $\ell$-node messages~\ref{eq:mcl_opt}.
This can be done in $O(nq^3)$ time complexity per message. Indeed, 
the $c$-node to $\ell$-node message update simply reads:
\[
\mu_{c\to\ell}(\x_\ell) = \min_{\x_c\backslash\x_\ell}\bigl[\Yp_c(\x_c)\bigr]-E_\ell(\x_\ell)-\nu_{\ell\to c}(\x_\ell)-\sum_{v\in\ell}\mu_{\ell\to v}(x_v).
\]
Running a min-sum algorithm associated to the energy function $\Yp_c(\x_c)$ given $\x_\ell$ on the loop $c$  
for each $\ell\in c$ yields immediately $\mu_{c\to\ell}$.

\section{Cycle basis determination}\label{sec:cbasis}
\subsection{Various criteria}
At this point, nothing has been said concerning the choice of the cycle basis. 
In~\cite{GeWe} it is argued that a good choice of basis ensures  
the algorithm of being tree-robust (TR), namely that GBP converges to an exact 
fixed point when the underlying graph $\G$ is singly connected. 
They provide a characterization for cycle basis ensuring this property. First it has to be a \emph{weak 
fundamental cycle basis} (WFCB), ensuring in particular the \emph{maxent} property to be satisfied.
By definition a cycle basis is fundamental if each cycle contains an edge that is not included in any other basis cycle.
For a WFCB, this constraint is relaxed, it is a cycle basis for which there is an ordering  
s.t. each cycle contains a link which is absent of all preceding cycles in this ordering.
In addition the WFCB is TR, if it is such that any subset of the cycle basis  contains 
a set of links, each one pertaining to a unique cycle in this subset,  and altogether forming at least one loop. 
The reason behind this can be understood quite simply in the special context of CVM approximation~(\ref{eq:gbp}),
where a simple reduction rule as the ones given in~\cite{Welling2004} is at work.
Suppose the MRF is such that the 
set of non trivial links $\psi_{ij}^{(0)}(x_i,x_j)\ne f(x_i)g(x_j)$ in~(\ref{eq:refP}) forms a tree $\T$.
\begin{prop}\label{prop:TR}
(i) if a trivial link $\ell$ pertains to a single cycle, the factorized joint measure~(\ref{eq:gbp}) 
coincides with the same CVM approximation defined on a reduced graph, 
where link $\ell$ has been removed and $c$ is discarded. \\[0.1cm]
(ii) if the cycle basis is a WFCB based on a series of trivial links, the factorized joint measure~(\ref{eq:gbp})
is reduced to the Bethe joint measure associated to the underlying tree $\T$.
\end{prop}
\begin{proof}
(ii) is the direct consequence of (i) by induction.
See Appendix~\ref{app:prop2}.
\end{proof}
As already stated in Proposition~\ref{prop:dtree}, GCBP is exact when the dual graph $\G^\star$ and henceforth the MFG are acyclic.
It could be tempting to push the logic to the end and try to impose a ``dual-tree robust''  condition for the cycle basis, i.e. that GCBP be exact 
if there exists a cycle basis of $\G$ s.t. $\G^\star$ be singly connected. Clearly this is a dead end, as can already be
seen by considering the simple example of a planar graph:
the natural cycle basis given by the faces of the graph cannot fulfilled such property, when all links at the border of the graph are non-trivial.
Nevertheless, let us simply notice that in the case where the underlying graph of non trivial links noted $\T_2$ 
has an acyclic dual graph $\T_2^\star$, we have the following
\begin{prop}
GCBP will converge to the exact fixed point if\\
(i) the cycle basis has for subset a cycle basis of $\T_2$,\\
(ii) the complementary set of cycles defines a graph for which it is a WFCB based on trivial links.
\end{prop}
\begin{proof}
The argument is the same as before, applying the reduction property (i) of the preceding Proposition to 
the complementary set of cycles, until reaching the core sub-graph $\T_2$, for which GCBP is exact.  
\end{proof}
TR cycle basis are easily identified in special cases like planar or complete graph~\cite{GeWe},
but searching for such a basis in general is difficult, its existence being not always guaranteed.  
Instead there is yet another feature that could be even more desirable, namely 
that the cycle basis be such that the number of independent dual cycles, 
i.e. the cyclomatic number of $\G^\star$ be minimal.
Recall that GCBP is similar to an ordinary  BP on the MFG. Consequently, as for an ordinary BP, we expect these (dual) loops to be a 
source of problem.
As observed in~\cite{Fu2013}, the dual cyclomatic number depends on the 
sum of cycle sizes noted $|c|$:
\[
C(\G^\star) = \sum_{c=1}^{C(\G)}|c|-C(\G)-\vert\E\vert+\PP(\G^\star),
\] 
with $\PP(\G^\star)$ the number of connected components of $\G^\star$.
As a result, a good choice for the cycle basis could be the minimal cycle basis (MCB) for which polynomial time 
algorithms exist~\cite{Horton}. Furthermore if one wants to remain close to the TR prescription, one could even search for 
a minimal WFCB, which is an  APX-hard problem but for which efficient heuristic do exist~\cite{Rizzi09}.

\subsection{Heuristic algorithm}
Exact algorithms for solving the MCB problem have a polynomial time complexity, scaling typically like $O(NL^2)$ up to logarithmic 
corrections~\cite{KaLiMeRo}. Making use of these would completely
spoil the efficiency of GCBP, which main expected virtue is to scale linearly with systems size. We have  therefore to
resort to some approximate procedure. It is guided by the empirical assumption that most important loops to 
be taken care of are the smallest ones. The main steps of the method are the following:
\begin{itemize}
\item (i) build a subset of candidate cycles which contains most important ones. This step can be made linear with system size for sparse graphs with 
bounded degree $d_\text{max}$; typically $O(N d_\text{max}^n)$ for finding cycles with sizes $\le n$.
\item (ii) complete this set in order to have a complete set containing the MCB. This step can be done exactly 
in $O(NL)$ time complexity~\cite{KaLiMeRo}.
\item (iii)  Extract an independent set of shortest sizes. Exact methods use typically Gaussian elimination which is the main source of
time consuming.
\end{itemize}
This strategy is basically the one which is followed by the most efficient exact algorithms. 
In order not to be a limiting speed factor for GCBP steps (ii) and (iii) have to be approximated.
Note that step (ii) is not mandatory. Since the goal is to take into account 
most important loop corrections, then an independent set of short cycles, not necessarily complete can make it.
Concerning step (iii) we replace the Gaussian elimination procedure by an approximate one which additional 
virtue is to respect as much as possible the WFCB criteria explained in the previous section.
Our algorithm goes as follows:
\begin{itemize}
\item (S0) Initialization: weight all the links with the number $n$ of cycles in the candidate set they belong to and extract w.r.t. these weights  
a maximum spanning tree from  $\G$ called $\G_0$. Create a double ordered list $\{c_0(n,s)\}$ of 
candidate cycles indexed by their number $n$ of links not already present in $\G_0$ and their sizes $s$. Create an empty list 
of cycle elements $B_0$.
\item (S1) cycle selection: At step $t$ select in $c_t$ the cycle $c$ with smallest $n$ and then with smallest size $s$ and update 
$B_{t+1} \longleftarrow B_t+\{c\}$. 
\item (S2) update $(c_t,\G_t)\longrightarrow (c_{t+1},\G_{t+1})$:
\begin{itemize}
\item if $n=1$: insert the corresponding link into $G_t$ to obtain $G_{t+1}$ and update $c_t$ in $c_{t+1}$. All cycles with $n=0$ have a linear decomposition in $B_{t+1}$ 
and are eliminated.  
\item if $n>1$: insert one of the $n$ free links of $c$ into $G_t$ to obtain $G_{t+1}$. Update $c_t$ in $c_{t+1}$ as if all the $n$ links where selected. For each 
of the $n-1$ non-selected  links of $c$ create a new cycle by joining this link to the path on $G_{t}$ connecting its two ends point, 
using a Dijkstra algorithm\footnote{This ensures the independence of these new cycles among each others and with $B_{t+1}$}. 
Insert these new cycles into $c_{t+1}$.  
\end{itemize}
if $c_{t+1}\ne\emptyset $ go back to (S1) else exit().
\end{itemize}
Note that if by chance the new added cycle at each step corresponds to $n=1$ we would get a WFCB. The procedure followed 
in the case $n>1$ is there to ensure that we get a complete set at the end. As already said this is not mandatory in 
practice, so if this constraint is relaxed, then the $n$ links can be directly inserted into $G_t$.

\subsection{Cycle basis cleaning}\label{sec:cleaning}
Once a cycle basis has been obtained some adjustments have to be performed to cope with GCBP. 
First the basis can be optimized further by a local random greedy shuffling procedure, which consists in to 
look for combination of pairs of cycles sharing some links, from which smaller cycles can be generated  
(see Figure~\ref{fig:shuffle}).
\begin{figure}[ht]
\centerline{\resizebox*{\textwidth}{!}{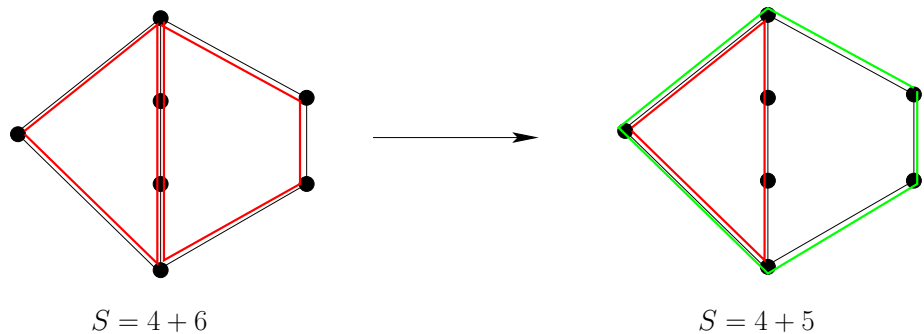}}
\caption{\label{fig:shuffle} Example of cycle combinations leading to smaller cycle basis.}
\end{figure}
Secondly, as already stated in the MFG prescriptions any pair of cycles must have at most one single link in common.
Note in passing that this requirement seems actually difficult if not impossible in general 
to conciliate with the search for TR cycle basis advocated in~\cite{GeWe}. In contrary the smaller the aggregated cycle's size is,
the less cleaning is to be expected. By cleaning we mean the operation shown on Figure~\ref{fig:dual_graph}. This consists in 
to add one link relating the two ends of a path common to two or more cycles and formed by at least two links.  
In this operation a new cycle composed of this path and of the new added link is created which, when combined with the other cycles
containing that path leaves all these cycles intersect on this single link. This cleaning operation is done greedily  
by treating in order the intersection paths with largest sizes until intersections composed of one single link remain. 

Finally in some cases, cycles remain which have non-connected intersection with other cycles. This kind 
of situation occur sometimes but rarely, so in practice the adopted cleaning procedure consists simply to
discard the largest cycle involved in such pathological intersection.

As we observed in practice, these cleaning operations take a small if not negligible part in the overall computation
time needed to determine the cycle basis. The complete workflow is shown on the example of Figure~\ref{fig:bip} leading to 
the MFG starting from a bipartite graph.
\begin{figure}[ht]
\centerline{
\includegraphics[width=\textwidth]{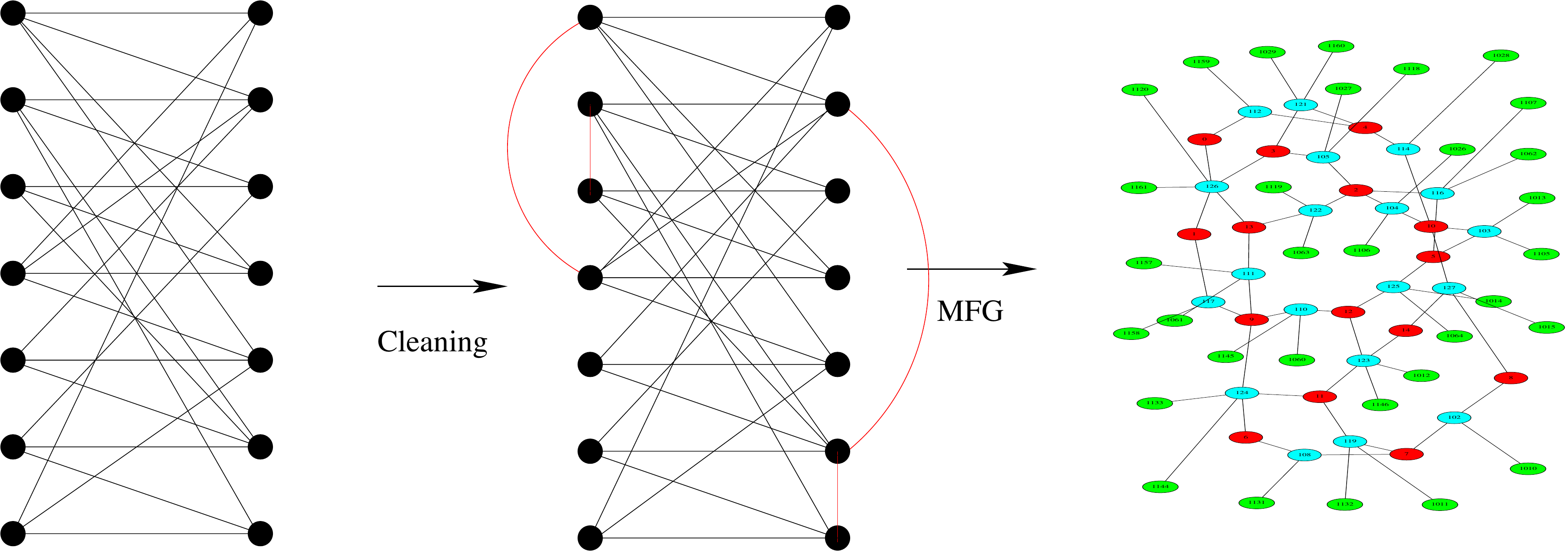}}
\caption{\label{fig:bip} Example of $7+7$ regular bipartite graph of mean connectivity $3.4$,  and corresponding mixed factor graph, 
with $c$-nodes, $\ell$-nodes and $v^\star$-nodes  colored respectively in red, blue and green. $v$-nodes associated to bridges are
absent on this example. $4$ auxiliary links (in red on the middle panel) have been inserted in order to ensure single link intersection between cycles as 
explained in Section~\ref{sec:cleaning}.}
\end{figure}

\section{Loop corrections: $c$-node to $\ell$-node messages}\label{sec:clmsg}
\subsection{General case}
We exploit now the specific structure of the cycle-based region definition to propose an efficient method for computing 
the messages~(\ref{eq:mcl}), with a cost at most linear w.r.t. the size of the cycles per message.
$c$-node to $\ell$-node messages amount to compute,
\begin{equation}
m_{c\to\ell}(\x_\ell) = \frac{p_\ell^c(\x_\ell)}{\psi_\ell(\x_\ell)n_{\ell\to c}(\x_\ell)\prod_{v\in\ell}n_{v\to\ell}(x_v)},\label{eq:mcloop}
\end{equation}
where
\[
p_\ell^c(\x_\ell) \egaldef \sum_{\x_c\backslash\x_\ell} p_c(\x_c).
\]
is the pairwise marginal associated to any link $\ell\in c$, obtained from distribution~(\ref{eq:pc}).
We wish to bypass the summation over $\x_c\backslash \x_\ell$, which has an exponential cost w.r.t. the size of the loop.    
Variables $x\in\{1,\ldots q\}$ are assumed to have $q$ possible states and $p_c$ is a product of pairwise
factors along the cycle
\[
p_c(\x_c) = \prod_{\ell\in c}\psi_\ell^c(\x_\ell).
\]
On the ring 
geometry, the partition function as well as any correlation function can be expressed as
the trace of a product of transition matrices:
\[
Z_\text{ring} = \Tr\bigl(\prod_{\ell=1}^n M^{(\ell)}\bigr),
\]
where $M^{(\ell)}$ is a $q^2$ matrix with elements given by
\[
M_{xy}^{(\ell)} = \psi_\ell^c(x,y)
\]
Upon introducing the following matrices 
\[
U \egaldef \prod_{i=1}^n M^{(i)},\qquad 
U^{(i)} \egaldef \prod_{j=i}^n M^{(j)}\prod_{j=1}^{i-1} M^{(j)},\qquad
V^{(i)} \egaldef \prod_{j=i+1}^n M^{(j)}\prod_{j=1}^{i-1} M^{(j)},
\]
the expression for the exact marginals are given by  
\begin{align*}
p_i^c(x) &= \frac{1}{Z_\text{ring}}\Tr\bigl(\delta_{xx}U^{(i)}\bigr)\\[0.2cm]
p_i^c(x,y) &= \frac{1}{Z_\text{ring}}\Tr\bigl(\delta_{xy}V^{(i)}\bigr).
\end{align*}
In this form the cost for computing each $c$-node to $\ell$-node message is 
$O\bigl(nq^3)$.
As shown in~\cite{Weiss}, running BP on a single cycle always converges and there is a linear
relation between single variable beliefs and the exact marginals given by the largest eigenvalue of some product of matrices
\begin{figure}[ht]
\centerline{\resizebox*{0.5\textwidth}{!}{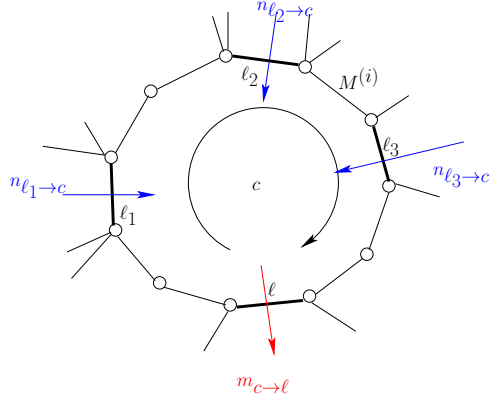}}
\caption{\label{fig:ctol} Message exchange at the cycle level.}
\end{figure}
taken from  the factors along the loop. In fact, somewhat simpler relations can be established, valid also for
pairwise marginals, by applying to a single loop 
the general loop corrections~\cite{chertkov1,SuWaWi} expansion to BP.
First  factorize $p_c(\x_c)$  with help of BP,
\begin{equation}
p_c(\x_c) = \frac{1}{Z_\text{BP}} \prod_{i=1}^n \frac{b_i^c(x_i,x_{i+1})}{b_i^c(x_i)}\label{eq:bp_fact}
\end{equation}
by means of a set of single and pairwise beliefs 
$b_i^c(x_i)$ and $b_i^c(x_i,x_{i+1})$, where $i=1,\ldots n$ indexes the variables along the cycle. 
We define the following $q^2$ matrices
in operator form:
\[
B_{xy}^{(i)} \egaldef \frac{b_i^c(x,y)-b_i^c(x)b_{i+1}^c(y)}{b_i^c(x)},
\]
and associated product of matrices
\begin{align}
U &\egaldef \prod_{i=1}^n B^{i},\qquad 
U^{(i)} \egaldef \prod_{j=i}^n B^{j}\prod_{j=1}^{i-1} B^{j},\label{eq:U}\\[0.2cm]
V^{(i)} &\egaldef \prod_{j=i+1}^n B^{j}\prod_{j=1}^{i-1} B^{j}.\nonumber
\end{align}
\begin{prop}\label{prop:gloop}
The relations between beliefs and exact marginals are then given by 
\begin{align*}
p_i^c(x) &= \frac{b_i^c(x)+U_{xx}^{(i)}}{Z_\text{BP}}\qquad\text{with}\qquad
Z_\text{BP} = 1+ \Tr(U)\\[0.2cm]
p_i^c(x,y) &= \frac{b_i^c(x,y)+V_{yx}^{(i)}b_i(x)+B_{xy}^{(i)}V_{yx}^{(i)}}{Z_\text{BP}}
\end{align*}
\end{prop}
\begin{proof}
See Appendix~\ref{app:gloop} for details.
\end{proof}
$c$-nodes messages~\ref{eq:mcl} can then be computed from these exact marginals.
From these expressions, we see that the cost for computing each message is still $O\bigl(nq^3)$.
The only benefit of using the BP factorization at this point resides in the fact that 
$B^{(j)}$ and therefore $U^{(i)}$ and $V^{(i)}$ have an obvious zero eigenmode:
\[
\sum_y B_{xy}^{(j)}b_i^c(y) =0.
\]
Trying to find the other modes is not advantageous in general except if some symmetries are present 
or when $q$ is small. In particular for the binary case ($q=2$) we end up with a scalar problem for expressing loop 
corrections, as is detailed in the next section.  

\subsection{Binary case}
For binary variables this relationship can be 
made even more explicit as we show now.
Using of the standard Ising spin notation, each node $i\in 0,\ldots n-1$ 
is associated to a binary variables $s_i\in\{-1,1\}$ and  
the joint measure of $\s\egaldef \{s_1,\ldots,s_n\}$ is exponential and given by 
\begin{equation}\label{eq:jointIsing}
P_c(\s) = \frac{1}{Z_c}\exp\bigl(\sum_{i=1}^n h_i^cs_i+\sum_{i=1}^{n-1} J_i^c s_is_{i+1}\bigr),
\end{equation}
where $h_i^c\in{\mathbb R}$ is the local field exerted on variable $i$ and $J_i^c\in{\mathbb R}$ 
denotes the coupling between $s_i$ and $s_{i+1}$. 
Running BP on this measure leads to the following factorization
of the joint measure:
\begin{align}\label{eq:cbpfp}
P(\s) = \frac{1}{Z_\text{BP}}\prod_{i=1}^n\frac{b_i^c(s_i,s_{i+1})}{b_i^c(s_i)b_{i+1}^c(s_{i+1})}\prod_{i=1}^n b_i^c(s_i),
\end{align}
where the $b_i^c(\cdot)$ and $b_i^c(\cdot,\cdot)$ are the single and pairwise approximate marginals delivered by BP.
These can be parameterized as follows 
\begin{align}
b_i^c(s_i) &= \frac{1}{2}(1+\bm_is_i),\label{eq:bi}\\[0.2cm]
b_i^c(s_i,s_{i+1}) &= \frac{1}{4}(1+\bm_is_i+\bm_j s_j +(\bm_i \bm_j+\bchi_i)s_is_j),\label{eq:bij}
\end{align}
where $m_i\egaldef {\mathbb E}(s_i)$  represents the ``magnetization'' of spin $s_i$
and $\chi_i\egaldef {\mathbb E}(s_is_{i+1})-{\mathbb E}(s_i){\mathbb E}(s_{i+1})$ is the covariance, also named ``susceptibility'' 
coefficient, between $s_i$ and $s_{i+1}$. We use the sign $\breve{}$ to denote a BP 
estimate, which is to be distinguished it from the exact value. 
The relation between BP values and  exact ones can be made explicit in the following form.
\begin{prop}
Let 
\begin{equation}\label{eq:Q}
Q \egaldef \prod_{i=1}^n\frac{\bchi_i}{\sqrt{(1-\bm_i^2)(1-\bm_{i+1}^2)}},
\end{equation}
then the BP normalization constant,
the exact magnetization and susceptibility coefficients read:
\begin{align}
Z_\text{BP} &= 1+Q,\qquad\label{eq:loop0}\\[0.2cm]
m_i &=  \frac{1-Q}{1+Q}\ \bm_{i} \label{eq:loop1}\\[0.2cm]
\chi_i &= \frac{\bchi_i}{1+Q}
+\frac{Q}{1+Q}\Bigl(\frac{(1-\bm_i^2)(1-\bm_{i+1}^2)}{\bchi_i}
+4\frac{\bm_i\bm_{i+1}}{1+Q}\Bigr).\label{eq:loop2} 
\end{align}
\end{prop}
\begin{proof}
The proof is based on the following identity
\[
\frac{b_i(s_i,s_{i+1})}{b_i(s_i)b_{i+1}(s_{i+1})} = 1+\bchi_i\frac{(s_i-\bm_i)(s_{i+1}-\bm_{i+1})}{(1-\bm_i^2)(1-\bm_{i+1}^2)},
\]
and follows the same lines as the proof of Proposition~\ref{prop:gloop}.
\end{proof}
Section~\ref{sec:KIC} will be based on these identities.
The corresponding  loop corrected marginals  $p_i$ and $p_{ii+1}$ are expressed from the loop corrected
quantities $(m_i,m_{i+1},\chi_i)$ through the same relations 
(\ref{eq:bi}) and (\ref{eq:bij}) and allow one to obtain all messages~\ref{eq:mcl} send by the $c$-node
at once from the BP beliefs, so the cost per-message in this special case is now $O(1)$ instead of $O(n)$ if there are $n$ 
messages to be sent. 

In addition to this slight but non-crucial reduction in computational cost 
is the scalar characterization in terms of $Q\in]-1,1]$ of the cycle which shows up. 
First from the matrix formulation~\ref{eq:Q}, $Q$ is the non-zero eigenvalue of $U$.
It is the product of ``BP correlations'' along the loop and characterizes
its strength.
\begin{itemize}
\item $Q\simeq 0$ corresponds to weak loop correction, BP is nearly exact.
\item $Q\to 1$ corresponds to a strongly correlated loop.
\item $Q\to -1$ corresponds to a strongly correlated frustrated loop. 
\end{itemize}

\subsection{Loop corrections to the Bethe Free Energy}
The formalism used previously suggests to reconsider the cycle based Kikuchi approximate free energy 
by rewriting it in an appealing form where loop correction are made more explicit. Indeed using the BP factorization 
of each independent cycle marginal~(\ref{eq:bp_fact}) yields the following decomposition of the entropy term for any pairwise 
MRF in terms of single and pairwise marginals $\{p_i,i\in\V\}$ and $\{p_\ell,\ell\in\E\}$ and associated cycle beliefs 
$\{b_i^c,(i,c)\in\V\times\Cy\}$ and $\{b_\ell^c,(\ell,c)\in\E\times\Cy\}$. Starting from the 
cluster expansion we  have:
\[
S_\text{Kikuchi} = \sum_{i\in\V} S_i +  \sum_{\ell\in\E}\Delta S_\ell + \sum_{c\in\Cy}\Delta S_c.
\]  
The first two terms represent the Bethe entropy, 
\[
S_\text{Bethe} = \sum_i S_i +  \Delta S_\ell, 
\]
as a sum of individual variables entropy $S_i$ corrected by mutual information of variables 
\[
-\Delta S_\ell = \sum_{\x_\ell} p_\ell(\x_\ell)\log\frac{p_\ell(\x_\ell)}{p_{\ell_1}(x_{\ell_1})p_{\ell_2}(x_{\ell_2})} \ge 0,
\]
counted for each link $\ell\in\E$. The corrections induced by each cycle $c$ has the following expression:
\begin{align}
\Delta S_c &= S_c - \sum_{i\in c} S_i -\sum_{\ell\in c}\Delta S_\ell\nonumber\\[0.2cm]
&= \log(Z_\text{BP}^c) - \sum_{i\in c} \dkl(p_i\Vert b_i^c ) + \sum_{\ell\in c} \dkl(p_\ell\Vert b_\ell^c),\label{eq:DSc}\\[0.2cm]
&=  \F_\text{Bethe}\bigl[p^c\Vert p^c].\nonumber
\end{align}
where $Z_\text{BP}^c$ is the normalizing factor of the BP factorization~(\ref{eq:bp_fact}) associated to cycle marginal distribution 
$p^c$. The cycle beliefs $b_i^c$ and $b_\ell^c$ are implicitly and uniquely determined from the $p_\ell$'s. 
$\F_\text{Bethe}$ is the Bethe approximation to the free energy functional:
\[
\F\bigl[p\Vert p_0] = D_\text{KL}(p\Vert p_0) + F_0,
\]
$F_0$ beign the free energy associated to $p_0$.
This has the following immediate consequence. Let us consider an auxiliary measure, build from the exact 
marginals:
\[
\tilde p^c(\x_c)  \egaldef \frac{1}{\tilde Z_\text{BP}^c} \frac{\prod_{\ell\in c}p_\ell(\x_\ell)}{\prod_{i\in c}p_i(x_i)}
\]
with normalization constant $\tilde Z_\text{BP}^c$. 
\begin{lem}  
\begin{equation}\label{eq:Fineq}
\log\bigl(Z_\text{BP}^c\bigr) \le \Delta S_c \le \log\bigl(\tilde Z_\text{BP}^c\bigr).
\end{equation}
\end{lem}
\begin{proof}
Recall that on the loop geometry BP has one single stable fixed point which corresponds to a global minimum 
of the approximate Bethe free energy functional~\cite{Heskes4}. Consequently, the minimum is obtained for 
$p=b$ in~(\ref{eq:DSc})  
\[
\F_\text{Bethe}\bigl[p^c\Vert p^c] \ge \log(Z_\text{BP}^c),  
\]
which proves the left hand side inequality. Next consider the following quantity:
\begin{align*}
D_\text{KL}(p^c\Vert\tilde p^c) &= \log\Bigl(\frac{\tilde Z_\text{BP}^c}{Z_\text{BP}^c}\Bigr)+ 
\sum_{i\in c} \dkl(p_i\Vert b_i^c ) - \sum_{\ell\in c} \dkl(p_\ell\Vert b_\ell^c)\\[0.2cm]
&= \log\bigl(\tilde Z_\text{BP}^c\bigr) - \Delta S_c \ge 0,
\end{align*}
since the Kullback-Liebler divergence is non-negative, we get the right hand side inequality of~(\ref{eq:Fineq}).
\end{proof}
As a consequence of~(\ref{eq:Fineq}), if the stochastic operator defined by~(\ref{eq:U}) has a positive trace then the loop correction 
has a counter effect to the Bethe correction $\Delta S_\ell$. 
In particular for binary variables in the ferromagnetic case, $\log(Z_\text{BP}^c) = \log\bigl(1+Q_c\bigr)$ with $Q_c\ge 0$, 
leading therefore to negative loop corrections to the Bethe free energy. 
Since the Kikuchi correction is exact in absence of dual loops, i.e. when $C_i^\star = 0,\ \forall i\in\V$,
we may expect that the correction is overestimated in presence of dual loops, i.e. that we should have a bounding of the free energy 
\begin{equation}\label{eq:Fbounds}
\F_\text{Kikuchi} \le \F \le \F_\text{Bethe},
\end{equation}
for ferromagnetic like systems,
when $\F_\text{Bethe}$ and $\F_\text{Kikuchi}$ are given in terms of the exact single and pairwise beliefs $\{p_i,i\in\V\}$ 
and $\{p_\ell,\ell\in\E\}$. Note that the inequality $F\le \F_\text{Bethe}$ only 
proved in some special ferromagnetic cases~\cite{SuWaWi}, involves the approximate marginals given by BP instead of the exact ones
in our case. The conditions under which the bounding~(\ref{eq:Fbounds}) might be relevant is left aside to future investigations. 

All this also suggests that in presence of dual loops some appropriate correction terms proportional to local dual loop 
counting numbers $C_v^\star$ could be  inserted into the free energy functional in order to compensate for  
the kind of ``overcounting'' of loop corrections which occurs in such cases. This possibility which would potentially lead to a new
family of approximate and hopefully more precise mean field schema is left aside for the moment and will be investigated in the near future.

\section{Kikuchi cycle-based (KIC) inverse inference}\label{sec:KIC}
From the explicit expression of the Kikuchi type approximation~(\ref{eq:gbp}) it should be in principle 
possible to find a set of fields and couplings corresponding to a given input of single and pairwise 
empirical marginals. Assuming first we know the graph structure and have a cycle basis, it remains to 
determine the marginal probabilities $p_c$, $p_\ell$ and $p_v$ associated to each region. We expect the $p_\ell$'s and $p_v$'s
to be given from the data, but the $p_c$'s have to be constructed. This means that the 
global inverse problem get decomposed into $\vert\Cy\vert$ small inverse problems.
In the Ising case, if we denote $h_i^c$ and $J_{\ell}^c$ the local field and coupling associated as in~(\ref{eq:jointIsing})
to the marginal representing cycle $c$, $\hat h_i^\ell$, $\hat J_{\ell}$ associated to $p_\ell$ 
and finally $\hat h_i$ to $p_i$,
then from~(\ref{eq:gbp}) the corresponding Kikuchi cycle based (KIC) approximate inverse Ising solution reads
\begin{align*}
h_i^\text{\tiny (KIC)} &= \kappa_i \hat h_i + \sum_{c\ni i} h_i^c +\sum_{\ell\ni i}(1-d_\ell^\star)\hat h_i^\ell,\\[0.2cm]
J_{\ell}^\text{\tiny (KIC)} &= (1-d_\ell^\star)\hat J_\ell + \sum_{c\ni \ell} J_\ell^c.
\end{align*}
When the graph structure is unknown, one possibility is 
to select a set of candidate links, the one 
carrying the largest amount of mutual empirical information
among all possible edges. Then on the graph defined by those links an algorithm is run 
in order to find the minimal cycle basis, w.r.t. the weights given by minus the 
mutual information. More refined strategies could then be used like the one
based on iterative proportional scaling proposed in~\cite{MaFuHaLa} in the context of 
Gaussian MRF.

In the following we concentrate on how to invert equations~(\ref{eq:loop1},\ref{eq:loop2})
in order to compute $h_i^c$ and $J_\ell^c$ for any cycle $c\in\vert\Cy\vert$.

\subsection{Fixed point method}\label{sec:iter}
Consider a single loop of size $n$.
Assume we are given a set of empirical marginals $\hat p_i(s_i)$ and $\hat p_i(s_i,s_{i+1})$,
for $i=1,\ldots n$ or equivalently a set of magnetization $\hat m_i$ and susceptibilities $\hat \chi_i$.
First note that the change of variable $\{h_i,J_i,i=1,\ldots n\}$ to $\{\bm_i,\bchi_i\}$
is a one to one mapping: on the one hand $h_i$ and $J_i$ can be explicitly written in terms of 
the $\{\bm_i,\bchi_i\}$ (see below); on the other hand, on a loop there is a unique BP fixed point yielding 
factorization~(\ref{eq:cbpfp}), so through relations  
(\ref{eq:bi},\ref{eq:bij}) $\{\bm_i,\bchi_i\}$ are uniquely determined.

Finding a joint-measure of highest likelihood to model the empirical marginals is therefore equivalent
to find a set of parameters $\bm_i$ and $\bchi_i$ defining the joint-measure~(\ref{eq:cbpfp}) 
which satisfy $\chi_i=\hat \chi_i$ and $m_i=\hat m_i$ in equations~(\ref{eq:loop1},\ref{eq:loop2}). 
The problem is therefore to find the unique value of $Q$ for which all the relations are satisfied. 
Note also that these relation could be as well obtained by writing down the gradient of the log likelihood, 
which in the $(h,J)$ variables is a convex function. Hence these equations must anyway have a unique 
valid solution. The reason for not working in these $(h,J)$ variables is that 
the LL is not given explicitly  in these variables but in the
$\bm$ and $\bchi$ variables (see below).
By rewriting equations~(\ref{eq:loop1},\ref{eq:loop2}) in term of the spin-spin correlation
\begin{equation}\label{def:Theta}
\Yi_i \egaldef \frac{\chi_i}{\sqrt{(1-m_i^2)(1-m_{i+1}^2)}},
\end{equation}
letting $Q$ simply read
\begin{equation}\label{eq:QY}
Q = \prod_{i=1}^n \bYi_i,
\end{equation}
we arrive at the following fixed-point equation: 
\begin{prop}\label{prop:schema}
The solution $(\vec\bm,\vec\bchi)$ satisfying  equations~(\ref{eq:loop1},\ref{eq:loop2})
for a given set $\{m_i=\hm_i\egaldef\tanh(\hat h_i),i=1,\ldots n\}$ and 
$\{\chi_i=\hchi_i,i=1,\ldots n\}$ of empirical magnetization and susceptibilities 
is determined by the $n$-dimensional vector $\vec\bYi$ obeying 
\[
\vec\bYi  = \vec f(\vec\bYi),
\]
with
\begin{equation}\label{eq:iterf}
f_i(\vec\bYi) \egaldef A_i(Q)\hat\Yi_i - \frac{Q}{\bYi_i},
\end{equation}
where 
\begin{equation}
A_i(Q) \egaldef  
\frac{(1+Q)(1-Q)^2\hat\Yi_i -4Q(1+Q)\sinh(\hat h_i)\sinh(\hat h_{i+1})}
{\sqrt{(1-2Q\cosh(\hat h_i)+Q^2)(1-2Q\cosh(\hat h_{i+1})+Q^2)}},
\end{equation}
\end{prop}
\begin{proof}
Expressing all the magnetization $\bm_i$ in equation~(\ref{eq:loop2}), in terms of $Q$ and $\tanh(\hat h_i)$ with help of~(\ref{eq:loop1}),
after performing the change of variable $\bchi_i\longrightarrow\bYi_i$ yields the desired result.
\end{proof}
Let us specify the domain $\mathbb D\subset[-1,1]^n$ of validity for this iterations schema.
For arbitrary magnetizations and susceptibility there are some basic constraints.
The first one is that $\bm_i\in[-1,1]$, for all $i\in\{1,\ldots,n\}$ which entails 
\[
Q \le Q_{max} \egaldef \max_{i}\frac{1-\hm_i}{1+\hm_i}.
\]
The second set of constraints is 
that probabilities $b(s_i,s_{i+1})$ are in $[0,1]$:
\begin{equation}
\forall (s_i,s_{i+1})\in\{-1,1\}^2,\ 
0 \le (1+\bm_is_i)(1+\bm_{i+1}s_{i+1})+\bchi_is_is_{i+1} \le 4.
\end{equation}
We may rewrite these constraints in a more convenient form.
We denote by $\bh_i$ the local fields  
corresponding to $\bm_i = \tanh(\bh_i)$. In these notations
the constraints now read:
\[
0 \le e^{\bh_is_i+\bh_{i+1}s_{i+1}} + \bYi_is_is_{i+1} \le 4\cosh(\bh_i)\cosh(\bh_{i+1}).
\]
Considering all possible cases for $(s_i,s_j)$ we end up with the following somewhat simpler constraints:
\begin{equation}\label{ineq:Y}
-e^{-|\bh_i+\bh_{i+1}|} \le \bYi_i \le e^{-|\bh_i-\bh_{i+1}|},
\end{equation}
which combined with $Q\in[-1,Q_{max}]$ entirely defines the domain $\mathbb D$ and 
which prove useful in practice to restrict efficiently the search for a fixed point in a valid domain. 

\paragraph{Stability analysis:}
In order to remain inside ${\mathbb D}$ the iterate schema 
is defined as follows:
\begin{align}
g:&{\mathbb D} \longrightarrow {\mathbb D}\\[0.2cm]
&\vec X \longrightarrow \vec Y = 
\begin{cases} 
\DD \vec f\bigl(\vec X\bigr),\ if\ f\bigl(\vec X\bigr)\in{\mathbb D},\\[0.2cm]
\DD U({\mathbb D}),\ if\ f\bigl(\vec X\bigr)\notin{\mathbb D}.
\end{cases}
\end{align}
where $f$ coincide with~(\ref{eq:iterf}) for any $\bchi$ such the image is in the 
domain $\mathbb D$ and is otherwise replaced by a random 
function $U:{\mathbb D} \longrightarrow {\mathbb D}$. This one 
consists first to draw $Q$ uniformly between $]-1,Q_{max}]$,
and then draw $\bYi_i$ for each $i=1\ldots n$, uniformly between the bounds
given in~(\ref{ineq:Y}). Finally an overall scaling is applied to each $\bYi_i$ 
if the product exceeds $Q_{max}$. Defines as it is $g$ is an iterate on a compact 
domain with no other guaranty than there exists one unique fixed point solution. 
Let us examine the conditions under which this solution corresponds to a stable fixed point.  
The Jacobian of this iterative map, when it coincides with $f$
reads
\[
J_{ij}\egaldef\frac{\partial f_i}{\partial\Yi_j} = \frac{Q}{\Yi_j}\bigl(A_i'(Q)-(1-\delta_{ij})\frac{1}{\Yi_i}\bigr).
\]
Denoting $\Yi_{min}^{1,2}$ the two lowest absolute values of $\Yi_i$ and 
\[
B(Q) \egaldef \max_i \vert A_i'(Q)\Yi_i\vert,
\]
we get the following sufficient condition of local convergence:
\begin{prop}\label{prop:stabil}
The fixed point is stable in general if
\begin{equation}\label{ineq:Qconv}
\vert Q\vert < \frac{\Yi_{min}^{(1)}\Yi_{min}^{(2)}}{n-1+B(Q)},
\end{equation}
and in particular if
\begin{equation}\label{ineq:Qconv0}
\vert Q\vert < \frac{\Yi_{min}^{(1)}\Yi_{min}^{(2)}}{n}.
\end{equation}
in absence of magnetization.
\end{prop}
\begin{proof}
See Appendix~\ref{app:stabil}
\end{proof}
When some of the magnetizations $\hat m_i$ are non zero, the coefficient $B(Q)$ can become arbitrarily 
large when $Q$ approaches $Q_{max}$ so clearly there exists a value of $|Q|$ above which the 
condition~\ref{ineq:Qconv} will be violated. For small $Q$ we have
\begin{equation}
B(0) = \max_i \Big\vert\hat\Yi_i-4\sinh(\hat h_i)\sinh(\hat h_{i+1})
+\cosh(\hat h_i)+\cosh(\hat h_{i+1})\Big\vert, 
\end{equation}
which as well diverges when one of the magnetization $\hm_i$ approaches $\pm 1$, which means that 
convergence problems are likely to occur in this domain. Instead, for small magnetizations 
$B(Q)$ can get smaller to $1$,
\[
\lim_{\max_i\hm_i\to 0} B(Q) = \max_i \hat\Yi_i \le 1.
\] 
The inequality~(\ref{ineq:Qconv0}) becomes relevant in this regime and the iterative schema can 
converge for small $Q$, in particular if the largest correlation $\Yi$ is no greater
than $n^{-1/(n-2)}$ which is close to $1$ for $n\gg 1$.

\subsection{Line search optimization}\label{sec:lso}
The preceding conditions are not always met to 
guaranty the convergence of the fixed
point method. Therefore we develop an alternative method which directly maximizes 
the log likelihood, this latter being an explicit function $LL(\vec\Yi)$ of the  $\Yi_i$'s,
\begin{equation}\label{eq:LLQ}
LL(\vec\Yi) \egaldef -\log\bigl(1+Q(\vec\Yi)\bigr) + \sum_i \Bigl(w_i(\vec\Yi)
+h_i(\vec\Yi)\hat m_i+J_i(\vec\Yi)(\hat\chi_i+\hat m_i\hat m_{i+1})\Bigr)
\end{equation}
By convention we have 
\[
LL(\vec\Yi) = -\infty,\ \forall\ \vec\Yi\notin{\mathbb D}.
\]
The corresponding Ising fields and couplings of the cycle are 
given by 
\begin{align*}
w_i &= \frac{1}{4}\log\frac{b_i(-1,-1)b_i(-1,1)b_i(1,-1)b_i(1,1)}
{b_i^2(-1)b_i^2(1)}\\[0.2cm]
h_i &= \frac{1}{2}\log\frac{b_i(-1)}{b_i(1)}
+\frac{1}{4}\sum_{j\in\{i-1,i\}}\log\frac{b_j(1,1)b_j(s_i=1,s_j=-1)}
{b_j(s_i=-1,s_j=1)b_j(-1,-1)}\\[0.2cm]
J_i &= \frac{1}{4}\log\frac{b_i(-1,-1)b_i(1,1)}{b_i(-1,1)b_i(1,-1)},
\end{align*}
in addition to the weighting exponents $w_i$ which shows up. 
All these parameters are given through~(\ref{eq:bi},\ref{eq:bij}) as function of the magnetizations $\bm_i$
and susceptibilities $\bchi_i$ which in turn are fully determined by the $\bYi_i$'s 
through~(\ref{def:Theta}) and (\ref{eq:loop1},\ref{eq:QY}) given $m_i=\hat m_i$.
Let ${\mathbb D}_Q\subset [-1,Q_{max}]$ the domain of possible values for $Q$.
In order to find the optimal point we show the following 
\begin{prop}
There exists two functions 
\begin{align*}
h&: {\mathbb D}_Q\longrightarrow {\mathbb R}\\[0.2cm]
\vec\Yi&: {\mathbb D}_Q \longrightarrow {\mathbb D}
\end{align*}
s.t.
\[
\argmax_{\vec\Yi\in{\mathbb D}} LL(\vec\Yi) = \vec\Yi(Q^\star)
\]
with
\[
Q^\star = \argmax_{Q\in{\mathbb D}_Q} h(Q).
\]
\end{prop}
\begin{proof}
To prove this we explicitly construct these functions, which in turn will
be used to run a line search algorithm.

First note that taking the gradient of $LL(\vec\Yi)$
w.r.t. the $\bm_i$'s and $\bchi_i$'s in order  to find the 
stationary points leads to equations~(\ref{eq:loop1}) and~(\ref{eq:loop2}).
After doing the change of variables and manipulations given in Proposition~\ref{prop:schema},
the set of equations to be solved reads:
\[
\bYi_i^2 -A(Q)\ \bYi_i+Q = 0,\qquad\text{for}\ i=1,\ldots n,
\]
where $Q$ depends implicitly on the solutions.
A first consequence is that, given $Q$, there is the constraint that the quadratic equation have solutions, 
i.e. that
\[
A_i(Q)^2-4Q\ge 0,\qquad\forall\ i=1,\ldots n,
\]
which depends only on the empirical values $\hat m_i$ and $\hat \chi_i$. This further constraints the 
domain ${\mathbb D}_{Q}\subset[-1,Q_{max}]$ of possible values of $Q$.
If this condition is fulfilled, for each $i=1,\ldots n$, 
there are two solutions,
\[
\bYi_i(Q,\sigma_i) = \frac{A(Q)+\sigma_i\sqrt{A(Q)^2-4Q}}{2},
\]
where $\sigma_i\in\{-1,1\}$ is introduced by convenience. Unfortunately, in general 
both solutions can be valid, as long as they satisfy the constraints~(\ref{ineq:Y}). 
At the fixed point, which is unique, the $\bm_i$'s and $\bYi_i$'s are uniquely given by $Q$,
therefore among the $2^n$ possible choices, the correct one will satisfy~(\ref{eq:QY}) and corresponds 
to the lowest likelihood. The function $h$ can now be defined as follows:
\begin{align*}
h: {\mathbb D}_Q &\longrightarrow {\mathbb R}\\[0.2cm]
 Q &\longrightarrow LL\bigl(\vec\Yi(Q)\bigr)
\end{align*}
where $\vec\Yi(Q)$ in turn is given as
\begin{equation}\label{eq:YQopt}
\vec\Yi(Q) = \argmax_{\sigma} LL\bigl(\vec\Yi'(Q,\sigma)\bigr)
\end{equation}
with 
\[
\Yi_i'(Q,\sigma_i) = \frac{Q}{\prod_{j=1}^n \bYi_j(Q,\sigma_j)}\bYi_i(Q,\sigma_i).
\]
This last normalization is there to ensure that $\vec\Yi(Q)$ effectively corresponds to $Q$.
\end{proof}
\subsection{Combined method and MRF inference}
The two methods can be combined by selecting the solution
with highest LL~(\ref{eq:LLQ}), after running each one with a fixed computational budget. 
\begin{figure}[ht]
\centering
\resizebox{0.6\columnwidth}{!}{\input{LoopPlot.tex}}
\caption{Success rates for the inverse inference on a single cycle with different sizes (color) for the fixed point (FP),
the line search (LS) and the combined methods (LS$+$FP). 
}\label{fig:1loop}
\end{figure}
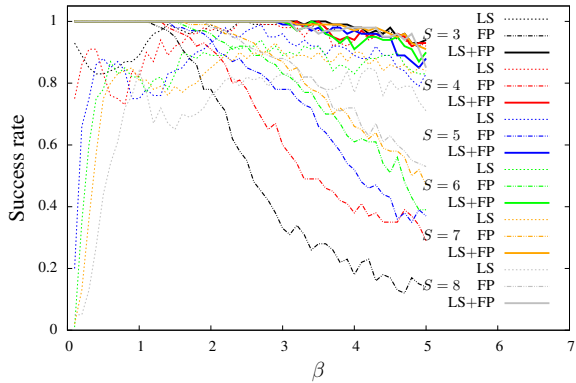
The line search method
has a combinatorial step present in~(\ref{eq:YQopt}), which can be solved by simple enumeration 
for small loops, but may become problematic for large ones, $n\gg 1$. However, for larger cycles, already 
typically for $n>5$, $Q$ is usually very small  and the iterative schema of Section~\ref{sec:iter} is converging.
Even though some specific optimization might well be possibly developed to solve~(\ref{eq:YQopt}),
we leave this question aside, as being non critical as confirmed by the experimental results shown 
on Figure~\ref{fig:1loop}. 

To infer an MRF, a set of candidate cycles is 
either given  either pre-processed from the data e.g. using mutual information scores.
As already mentioned, in such case we look for a minimal cycle basis, which in practice, 
can be approximately obtained at low computational cost as in experiments of the next Section, 
by a simple stochastic heuristic of loop mixing. 
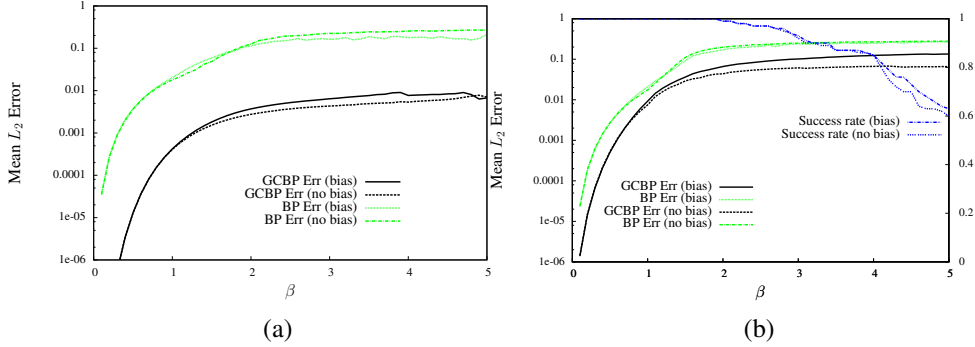
\begin{figure}[ht]
\centering
\resizebox{0.47\columnwidth}{!}{\input{grid_I.tex}}
\hspace{0.5cm}
\resizebox{0.45\textwidth}{!}{\input{bip_dirK4.tex}}\\
(a)\hspace{6cm}(b)
\caption{Mean error for the direct inference of $2$-D random Ising model 
comparing GCBP with BP as a function of $\beta$, on a $5\times 5$ square grid (left) and on random $20+20$ bipartite graphs of mean connectivity $4$ (right)
with or without local fields of amplitude $0.2\beta$,  averaged over $100$ instances.}\label{fig:gridI}
\end{figure}
For general pairwise MRF, with non-binary variables
no specific method is proposed at the cycle level, but at least a gradient descent could 
be used to solve each cycle independently. If necessary, a posterior selection procedure, 
based on the generated solution, could be used to refine the cycle basis, with various possible heuristics,
which are still under investigation. 
Concerning the overall computational cost needed to generate an approximate MRF 
solution, assuming a ``low-cost" method for fixing the cycle basis, 
it is linear in the number of candidate cycles i.e. in the number of potential links. Therefore 
the method can in principle cope with large scale problems when a sparse graph is to be expected.
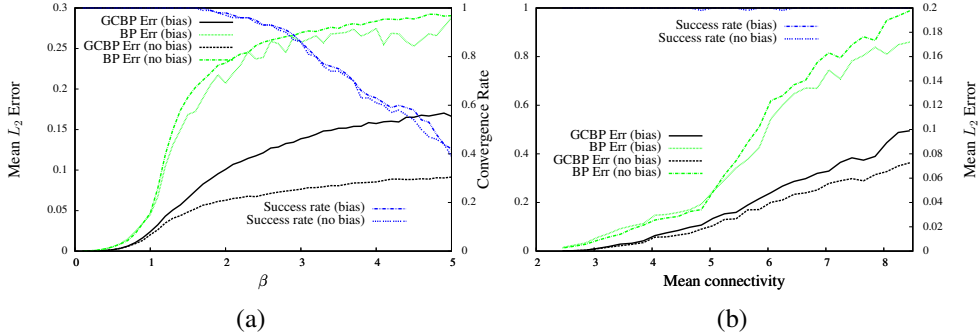
\begin{figure}[ht]
\centering
\resizebox{0.45\textwidth}{!}{\input{bip_dirK5.tex}}
\hspace{0.5cm}
\resizebox{0.45\textwidth}{!}{\input{bip_dirJ1.tex}}\\
(a)\hspace{6cm}(b)
\caption{Success rates and mean error for the direct  inference problem, comparing GCBP with BP 
on $20+20$ random bipartite graphs of mean connectivity 5 with varying $\beta$ (left) or with fixed $\beta=1$ and increasing 
the mean connectivity (right), in presence or not of random local fields of max amplitude $0.2\beta$, averaged over $100$ instances.}\label{fig:bip_dir}
\end{figure}

\section{Experiments}\label{sec:exp}
We have run various experiments to see how this approach to direct and inverse inference works 
in practice. 
\subsection{Direct inference}
Figure~\ref{fig:gridI} deals with direct inference, GCBP is run on  $5\times 5$ grids so that the RMSE  on the
beliefs (single and pairwise) can be computed by exact enumeration. 
Couplings $J_{ij}$ and local fields $h_i$ are  i.i.d sampled uniformly respectively in the range $[-\beta,\beta]$
and $[-0.2\beta,0.2\beta]$ when local fields are present. $\beta$ is varied on the range $[0,5]$, so that weak and strong coupling are tested. $100$
instances are generated for each point. With a damping factor up to $.5$ inserted in the $c$-node to $\ell$-node messages needed at low temperature, 
GCBP always converge on these small grids instances to a fixed point corresponding to a paramagnetic state. 
At larger scale Figure~\ref{fig:ctime_grid}, thanks again to a damping factor up to $.6$, 
the algorithm is also always converging on the considered range of temperature and sizes but 
two dynamical regimes are observed. At high temperature, for $\beta \le 1.5$ the computational time grows like 
$N^\alpha$ with a slight departure from linear complexity as $\beta$ increases, $\alpha=1.05$ for $\beta=0.5$  
and $\alpha=1.15$ at $\beta=1.5$. In that case all the fixed points correspond to 
paramagnetic  states. Instead at $\beta > 1.5$ and no external fields,  
the occurrence of non-paramagnetic  states is observed at sufficiently large scale, $N\ge 10^5$ for $\beta > 1.5$ and $N\ge 10^4$ for $\beta=2$, 
as observed also in the $\pm J$ $2$-D EA model\footnote{Thresholds are comparable after dividing our $\beta$ by $\sqrt{3}$
to have random models with identical variance of the couplings.} in~\cite{DoLaMuRiTo}.  
\begin{figure}[ht]
\centering
\resizebox{0.7\columnwidth}{!}{\input{ctime_grid_new.tex}}
\caption{Convergence behaviour of GCBP and BP regarding computational time on $2$-D EA models of large sizes. 
Cases corresponding to $\beta=0.5,1$ have local random fields
in $[-0.1\beta,0.1\beta]$ while other cases are without external fields.}\label{fig:ctime_grid}
\end{figure}
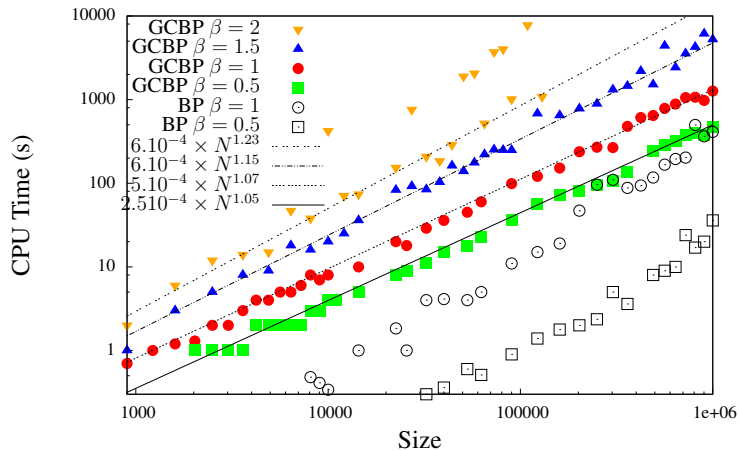
This is an artifact  of the Kikuchi approximation since the $2$-D EA model is thought to be exempt from a
spin-glass phase~\cite{JoLuMaMa}. Convergence is still observed in this regime, but huge fluctuations in computational time occur, 
depending on whether GCBP converges towards a paramagnetic or to a spin-glass fixed point. On the example shown, outliers points 
w.r.t. the fitted scaling actually correspond to spin-glass fixed points, while all other points are paramagnetic.
This is clearly related to the fact that a long range order 
has to be found by a GCBP fixed point when converging to a spin-glass state which is not the case for a paramagnetic one. Indeed 
in the paramagnetic situation, fixed point messages depend from each others within distances on the grid of the order of the spatial
characteristic scale for the correlations which increases with $\beta$. 
When compared to BP, the computational time for GCBP is larger by a factor of $5$ to $25$, but in addition to be less precise, 
BP is by far less robust and actually stops converging around $\beta\gtrsim 1$.
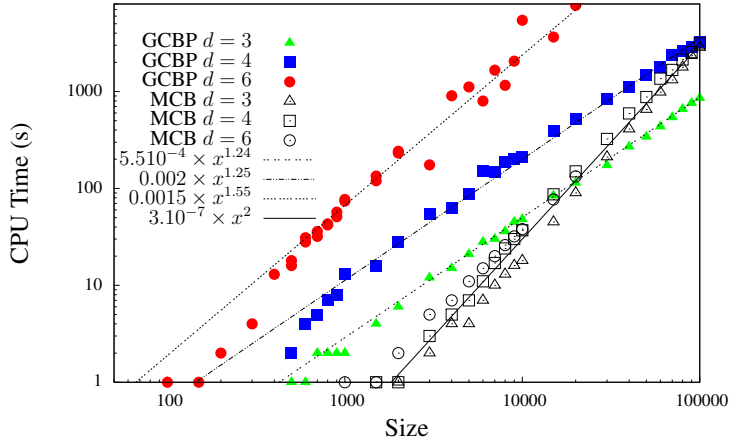
\begin{figure}[ht]
\centering
\resizebox{0.7\columnwidth}{!}{\input{ctime_bip_new.tex}}
\caption{Computational times of GCBP and the MCB search algorithm on random bipartite graphs at $\beta=1$ for different mean connectivity $d$.}\label{fig:ctime_bip}
\end{figure}
The same experiments are performed first on small random sparse $20+20$ regular bipartite graphs, for which exact beliefs can as well be 
computed by complete enumeration. In these cases the cycle basis are not given in advance and have to be determined. 
On Figures~\ref{fig:gridI}.b and \ref{fig:bip_dir}.a   we again vary the temperature
for a fixed mean connectivity $d=4$ and $d=5$, while on Figure~\ref{fig:bip_dir}.b the inverse temperature is kept fixed at $\beta=1$ and the 
mean connectivity  is varied up to $d=9$. 
As seen on Figure~\ref{fig:bip_dir}.b. convergence problems are absent below 
$\beta\lesssim 2$ but occur at small temperatures with increasing frequency above this threshold signaling the presence of a 
spin glass phase.
In addition, up to $d=9$ we observe a significant gain factor in the error made by GCBP w.r.t to ordinary BP.
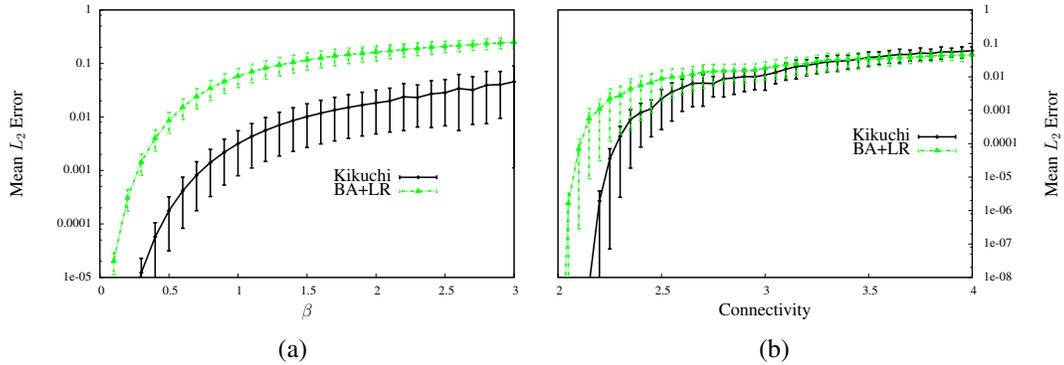
\begin{figure}[ht]
	\centering
        \resizebox{0.495\columnwidth}{!}{\input{grid_IIa.tex}}
        \resizebox{0.495\columnwidth}{!}{\input{bip_II.tex}}\\
        (a)\hspace{6cm}(b)
\caption{Comparison of KIC with BA+LR at infinite sampling 
on a $5\times 5$ square grid when $\beta$ is varied (left),
on random bipartite graphs 
at $\beta=1$ with biases of amplitude $0.2\beta$ varying the mean connectivity (right).}\label{fig:inf}
\end{figure}

On Figure~\ref{fig:ctime_bip} are shown results of tests that were performed on random sparse bipartite graphs
of size up to $N=10^5$ and mean connectivity up to $d=6$. We obtain as well good convergence properties, with no 
convergence failures, thanks again to a damping factor of $.7$ for $d=3$ to $.9$ for $d=6$. Concerning
computational time we observe a scaling in $N^\alpha$ which deviates from the linear one as expected as the graph becomes denser for a fixed temperature,
$\alpha$ ranging from $1.2$ at $d=3$ to $1.55$ at $d=6$.
Heterogeneous graphs with larger mean connectivity have a tendency to contains more highly connected nodes 
for which $\Cy_v^\star\gg 1$. We suspect these nodes to be mainly responsible for a slowing down of convergence. 
On the same figure we also show the computational time needed by our approximate pre-processing 
cycle basis stochastic optimization. The scaling is quadratic when the heuristic detailed in Section~\ref{sec:cbasis}
is used in its complete version, but the very small multiplicative constant allows us to go for relatively large size,
before becoming a limiting factor for GCBP around $N\simeq 10^4$ for $d=3$ and $N\simeq 10^5$ for $d=4$. Since collecting 
most important small loops has a linear complexity, the way to overcome this issue  at large scale is 
then to limit ourselves to an incomplete set of independent cycles.

\subsection{Inverse inference}
For the inverse Ising problem, we first test the single loop algorithm explained in Section.~\ref{sec:iter}
and Section.~\ref{sec:lso} and the results are shown on Figure~\ref{fig:1loop}. For this we generate loops of increasing 
sizes $S\in\{3,\ldots 8\}$. 
Couplings and biases are sampled as before, with an inverse temperature parameter $\beta$ varied again in 
the range $[0,5]$. The inference is considered successful for a precision threshold,
arbitrarily  chosen to $10^{-5}\beta$, on the max error
of the couplings and biases. A comparable computational budget of a maximum of $100$ iterations for FP
or estimations for LS is given to both methods. Note however that generally when it converges FP does it 
within $10$ or $20$ iterations.
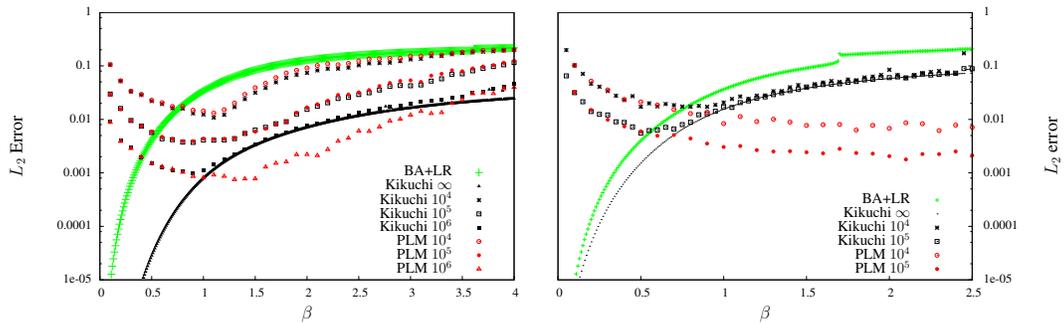
\begin{figure}[ht]
	\centering
        \resizebox{0.495\columnwidth}{!}{\input{grid_113.tex}}
        \resizebox{0.495\columnwidth}{!}{\input{bip_9.tex}}
	\caption{Comparison of KIC with  BA+LR at infinite and with PLM at finite sampling 
on a $5\times 5$ square grid (left) and on a bipartite model with connectivity $3$(right) when $\beta$ is varied.}
 	\label{fig:gridIIb}
\end{figure}
The Fixed point method is always successful for all sizes when $\beta \le 1.2$, but this rate degrades
when $\beta$ is increased albeit less severely with larger loops. In contrary the line search 
method is not sufficiently precise at small $\beta$ but sees its success rate increase with $\beta$ especially for small loops.
Therefore the two methods are very much complementary, and combining them leads to nearly maximal success rates,
at least for $\beta\le 3$.

Our KIC method is then tested and compared with the linear response of 
the Bethe-Peierls  approximation~\cite{NgBe} (BA+LR)
at infinite sampling and with the pseudo-likelihood method (PLM)~\cite{RaWaLa,DeRi} at finite sampling, again on 
small square grid and on small sparse random bipartite models. Couplings and biases are sampled 
as before. Comparison with BA+LR indicates a gain in precision between $1$ to $2$ orders of magnitude for $5\times 5$
grids as seen on Figure~\ref{fig:inf} (left). For bipartite models,  
Figure~\ref{fig:inf} (right) shows a decreasing gain  with increasing
mean connectivity, BA+LR and KIC returning the same error around  $d=3.4$.   
On Figure~\ref{fig:gridIIb} one representative grid and bipartite instances are shown. 
As expected the error increases with $\beta$ but stays reasonably close to the order of a few percents 
in the strong coupling region $\beta>1$, in contrary to BA+LR which is useless in this region.
At finite sampling, by comparing with PLM, we see that the precision is either limited by the 
sampling itself (small $\beta$ or small sampling $Ns \le 10^5$) either by the Kikuchi approximation itself
for $\beta>1$ and $Ns=10^6$ on the grid instance and at $Ns=10^4$ and $\beta>1$ on the bipartite instance.

\section{Conclusion}
Our investigations on GBP has led us to propose a systematic way of dealing with cycle regions and 
a new mean field approach to inverse problems. Our contribution is two-fold: for the direct problem, 
we propose (i) an original specification of the region graph (MFG)
ensuring simple and robust convergence properties 
(ii) the loop message computation using ordinary BP ensuring fast message exchange between regions. 
(i)+(ii) characterize GCBP  as a new region based algorithm generic to pairwise MRF, which we have made specific in the binary case. 
For the inverse Ising problem, we propose a new mean-field approach (KIC) general for pairwise MRF models,  
which is simple and efficient at least for binary models and sparse graphs
but non necessarily of finite tree-width like 2d grids.
In particular the  modular aspect of the method, which consists in a decomposition of the 
problem into small independent inverse problems corresponding to each independent cycle is valid in general, not only for binary MRF.
For incomplete data, since it takes as input single and pairwise marginals, 
it could be a good alternative to PLM which requires complete data.

Still, the scalability of GCBP and KIC relies on the scalability of the cycle basis search algorithm for irregular graphs.
In~\cite{GeWe} it is argued that a good choice of basis ensures  
the algorithm of being tree-robust (TR), namely that GBP converges to an exact 
fixed point when the underlying graph $\G$ is singly connected after eliminating fake links. 
In our experiments we did not follow this prescription, but instead proposed a simpler one, namely 
based on the search of a minimal cycle basis, for which a specific heuristic has been developed with reasonable scalability. 

Concerning possible applications of this work, it is planned to use 
both the direct and inverse approach in combination, in order to test some traffic prediction schema based on the 
Ising model that has been developed in some preceding related work~\cite{MaLaFu}. In addition, 
the systematic treatment of the loops that we propose could presumably be extended in a specific way to the Potts model 
which has been applied in many different contexts like image processing~\cite{Tanaka} for instance.
Yet another perspective of this framework is to be found in the combinatorial optimization context which could help improve
approximate heuristics.

\bibliographystyle{acm}
\bibliography{gcbp}

\newpage
\appendix

\section{Proof of Proposition \ref{prop:dtree}}\label{app:prop1}

If $\G^\star$ is acyclic, we can build a junction tree using each cycle as a clique, 
so the form 3.1 is correct except maybe 
for the specific form chosen for $p_c$. 
The leave nodes of $\G^\star$ correspond 
either to dandling trees either to cycle regions of the primal graph $\G$. From the hypothesis on $\G$ 
these components are connected to the rest of the primal graph $\G$  either via a single node 
either via a link. So summing over all variables contained in each of these region except the contact 
node or link results in a subgraph of $\G$ which dual is still acyclic,  with a modified factor 
corresponding to the contact link or vertex. By induction, 
$\G$ can be reduced until one single arbitrary loop region remains, which still 
corresponds to a sub-graph of $\G$. This results therefore in a marginal probability 
$p_c$ having pairwise form with factor graph corresponding to cycle $c$.

\section{Dual loop-based instabilities}\label{app:instabil}
Let us consider an Ising model on the single dual loop graph of Figure~\ref{fig:dloop}
with uniform external field $h$ and coupling $J$. We give the label $0$ to the 
central node with counting number $\kappa_0 = 1$ and labels $\{1,2,3\}$ to the peripheral ones, these having $\kappa_v=0$.
Links with non-vanishing counting 
numbers ($\kappa_\ell = -1$)  are for $\ell\in\{01,02,03\}$ and 
cycles are labelled $\{012,023,031\}$. 
Using the corresponding minimal factor graph, we attach arbitrarily the only $v$-node indexed by $0$ to $\ell=01$. 
The following  exponential parameterization of the messages is adopted:
\begin{align*}
m_{c\to\ell}(\s_\ell) &= e^{w_{c\to\ell}+h_{c\to\ell}^1s_{\ell_1}+h_{c\to\ell}^2s_{\ell_2}+J_{c\to\ell}s_{\ell_1}s_{\ell_2}}\\[0.2cm]
m_{\ell\to 0}(s_0) &= e^{w_{\ell\to 0}+h_{\ell\to 0}s_0}.
\end{align*}
From the update rules~(\ref{eq:pac1},\ref{eq:pac2})
we get in particular for $(i,j)\in\{(1,2),(2,3),(3,1)\}$
\[
m_{0ij\to 0i}(s_0) \longleftarrow \sum_{s_j}\exp\Bigl(h_{0kj\to 0j}^0s_0+\bigl(h_j+h_{0kj\to 0j}^j\bigr)s_j+\bigl(J_{0j}+J_{0kj\to 0j}\bigr)s_0s_j\Bigr),
\] 
and more specifically
\[
h_{0ij\to 0j}^0 \longleftarrow h_{0kj\to 0j}^0 + \frac{1}{4}\log\frac{A_{++}A_{-+}}{A_{+-}A_{--}}
\]
with 
\[
A_{\sigma_1\sigma_2} \egaldef 
h_{0kj\to 0j}^0+\sigma_1(h_j+h_{0kj\to 0j}^j)+\sigma_2(J_{0j}+J_{0kj\to 0j}).
\]
From this we see that these iterative equations are at least marginally unstable, by the presence  
of an eigenmode of the Jacobian of eigenvalue $1$ corresponding to $h_{0kj\to 0j}^0=cte,\ \forall kj$. One additional dual loop 
centered on $v$-node $0$ would actually render this mode unstable.

\section{Proof of Proposition \ref{prop:TR}}\label{app:prop2}
By definition of the Lagrange multipliers, when a fixed point is obtained, the corresponding set 
of beliefs $\{b_i,b_\ell,b_c\}$ allows one to factorize the joint measure as~(\ref{eq:refP}), 
where for all cycles of the basis, $b_c(\x_c)$ is itself in Bethe form
\[
b_c(\x_c) = \frac{1}{Z_c}\prod_{i=1}^b\frac{b_{ii+1}^c(x_i,x_{i+1})}{b_i^c(x_i)}
\]
where the  $b_i^c$ and $b_{ii+1}^c$ are obtained from $b_c$ by running BP on the cycle and 
are in general different from the $b_i$ and $b_\ell$ computed globally. The relation between the two 
corresponds to the loop correction.
Let us call trivial, an edge $(ij)$ which factor is trivial 
$\psi_{(ij)}(x_i,x_j) = f(x_i)f(x_j)$. Similarly we say that a cycle has a trivial belief if it is related 
to variable and pairwise belief as
\[
b_c(\x_c) = \prod_{i=1}^b\frac{b_{ii+1}(x_i,x_{i+1})}{b_i(x_i)},
\]
i.e. the $b_i$ and $b_i^c$ coincide.
First we remark that a cycle $c$ containing one such trivial edge, not contained in any other cycle, 
has necessarily a trivial belief, because from the factorization~(\ref{eq:refP}) for any edge $\ell$
we have in that case
\begin{align*}
\psi_\ell^{(0)}(\x_\ell) &= f(x_i)g(x_j)b_\ell(x_\ell)\prod_{c\ni\ell} \frac{b_\ell^c(x_\ell)}{b_\ell(x_\ell)},\\[0.2cm]
&= f(x_i)g(x_j) b_\ell^c(x_\ell),
\end{align*}
so the pairwise cycle belief has to be of the form $b_\ell^c(x_\ell) = b_i^c(x_i)b_j^c(x_j)$.
As a result the factorized joint measure actually coincides with the same  CVM approximation form~(\ref{eq:gbp}) on a reduced graph, 
where link $\ell$ has been removed and $c$ is now discarded.
From hypothesis (ii) the set of trivial links contained
in one single cycle is non empty. As a results all these link can be removed and all corresponding
cycles discarded. On the reduced graph, again since all cycles have a trivial belief, there is  a non-empty subset 
of trivial link, that can be removed and so on. The procedure 
stop after eliminating all trivial links until only the underlying dual tree remains. The definition of the counting 
numbers ensures that we then end up with the Bethe form of the joint measure associated to this dual tree.

\section{Proof of Proposition \ref{prop:gloop}}\label{app:gloop}
The proof is based on the following factorization of the joint measure on a cycle with help 
of a belief propagation fixed point:
\[
P(\x) = \frac{1}{Z_\text{BP}}\prod_{i=1}^n\frac{b_{i}(x_i,x_j)}{b_i(x_i)b_{i+1}(x_{i+1})}\prod_{i\in\V}b_i(x_i)
\]
with 
\begin{align*}
\frac{b_{i}(x_i,x_j)}{b_i(x_i)b_{i+1}(x_{i+1})} &= 1 +\frac{b_i(x_i,x_{i+1})-b_i(x_i)b_{i+1}(x_{i+1})}{b_i(x_i)b_{i+1}(x_{i+1})}\\[0.2cm]
&\egaldef 1+\frac{B_{x_ix_{i+1}}^{(i)}}{b_{i+1}(x_{i+1})},
\end{align*}
and then by expanding the factors when taking averages. Let us call bond $ii+1$ the contribution corresponding 
to the factor $\frac{B_{x_ix_{i+1}}^{(i)}}{b_{i+1}(x_{i+1})}$ instead of $1$. The point is that one extremity of a bond 
cannot be left alone in this expansion, if the corresponding variable is summed over, because of  
the following identities:
\[
\sum_{x_i} b_i(x_i)\frac{B_{x_ix_{i+1}}^{(i)}}{b_{i+1}(x_{i+1})} = 
\sum_{x_i} B_{x_{i-1}x_i}^{(i-1)} = 0. 
\]
For example, the partition either all or none of the bound have to be selected, yielding only the two contributions:
\begin{align*}
Z_\text{BP} &= \sum_\x \Bigl(\prod_{i=1}^n b_i(x_i)+\prod_{i=1}^n B_{x_ix_{i+1}}^{(i)}\Bigr),\\
&= 1 + \Tr(U).
\end{align*}
For the single variable marginal, say $p_i(x_i)$, again either none or either all of the bonds have
to be selected, giving 
\begin{align*}
p_i(x_i) &= \frac{1}{Z_\text{BP}} \sum_{\x\backslash x_i} \Bigl(\prod_{j=1}^n b_j(x_j)+\prod_{j=1}^n B_{x_jx_{j+1}}^{(i)}\Bigr)\\[0.2cm]
&= \frac{b_i(x_i)+U_{x_ix_i}^{(i)}}{Z_\text{BP}}.
\end{align*}
For the pairwise marginals $p_i(x_i,x_{i+1})$ two additional contributions emerge corresponding to
selecting only the bond $ii+1$ or to selecting  all the bonds  except this one, yielding the announced expression.

\section{Proof of Proposition \ref{prop:stabil}}\label{app:stabil}
The problem is to bound in absolute value the largest eigenvalue  of the Jacobian. 
Let $\lambda$ an eigenvalue and ${\bf v}$ an eigenvector of $J$
Let 
\[
v = \max_j v_j
\]
and $i$ the corresponding index, s.t. $v_i = v$.
We have 
\begin{align*}
\vert\lambda\vert &= \vert\sum_j J_{ij}\frac{v_j}{v}\vert \\[0.2cm]
&\le \sum_j \vert J_{ij}\vert\\[0.2cm]
&\le  \vert \frac{Q}{\Theta_j}A'_i(Q)\vert +\sum_{j\ne i} \vert\frac{Q}{\Theta_i\Theta_j}\vert\\[0.2cm]
&\le \frac{\vert Q\vert}{\Theta_{min}^{(1)}\Theta_{min}^{(2)}}\bigl(B(Q)+n-1\bigr),
\end{align*}
with the definition of $B(Q)$ and $\Theta_{min}^{(1,2)}$ given in the text.
Imposing $\vert\lambda\vert \le 1$ leads to the conditions given in the proposition.
In particular when magnetization are absent, i.e. when $h_i=0,\forall i$,
we have
\[
A'(Q) = \hat\Theta_i
\]
so 
\[
B(Q) = \max_i\vert\hat\Theta_i\Theta_i\vert \le 1.
\]

\end{document}

%% file: factor_graph.pdf_t
\begin{picture}(0,0)%
\includegraphics{factor_graph.pdf}%
\end{picture}%
\setlength{\unitlength}{3947sp}%
\begingroup\makeatletter\ifx\SetFigFont\undefined%
\gdef\SetFigFont#1#2#3#4#5{%
  \reset@font\fontsize{#1}{#2pt}%
  \fontfamily{#3}\fontseries{#4}\fontshape{#5}%
  \selectfont}%
\fi\endgroup%
\begin{picture}(7291,4819)(2243,-5398)
\put(3931,-2191){\makebox(0,0)[lb]{\smash{{\SetFigFont{20}{24.0}{\rmdefault}{\mddefault}{\updefault}a}}}}
\put(3031,-5191){\makebox(0,0)[lb]{\smash{{\SetFigFont{20}{24.0}{\rmdefault}{\mddefault}{\updefault}b}}}}
\put(7531,-3406){\makebox(0,0)[lb]{\smash{{\SetFigFont{20}{24.0}{\rmdefault}{\mddefault}{\updefault}c}}}}
\put(2371,-1216){\makebox(0,0)[lb]{\smash{{\SetFigFont{20}{24.0}{\rmdefault}{\mddefault}{\updefault}$1$}}}}
\put(5386,-991){\makebox(0,0)[lb]{\smash{{\SetFigFont{20}{24.0}{\rmdefault}{\mddefault}{\updefault}$2$}}}}
\put(5746,-3016){\makebox(0,0)[lb]{\smash{{\SetFigFont{20}{24.0}{\rmdefault}{\mddefault}{\updefault}$3$}}}}
\put(4321,-4396){\makebox(0,0)[lb]{\smash{{\SetFigFont{20}{24.0}{\rmdefault}{\mddefault}{\updefault}$4$}}}}
\put(9136,-3991){\makebox(0,0)[lb]{\smash{{\SetFigFont{20}{24.0}{\rmdefault}{\mddefault}{\updefault}$6$}}}}
\put(9106,-2416){\makebox(0,0)[lb]{\smash{{\SetFigFont{20}{24.0}{\rmdefault}{\mddefault}{\updefault}$5$}}}}
\put(5776,-4186){\makebox(0,0)[lb]{\smash{{\SetFigFont{25}{30.0}{\rmdefault}{\mddefault}{\updefault}$m_{c\to 4}(x_4)$}}}}
\put(4801,-2461){\makebox(0,0)[lb]{\smash{{\SetFigFont{25}{30.0}{\rmdefault}{\mddefault}{\updefault}$m_{a\to 3}(x_3)$}}}}
\end{picture}%

%% file: cycle_basis.pdf_t
\begin{picture}(0,0)%
\includegraphics{cycle_basis.pdf}%
\end{picture}%
\setlength{\unitlength}{4144sp}%
\begingroup\makeatletter\ifx\SetFigFont\undefined%
\gdef\SetFigFont#1#2#3#4#5{%
  \reset@font\fontsize{#1}{#2pt}%
  \fontfamily{#3}\fontseries{#4}\fontshape{#5}%
  \selectfont}%
\fi\endgroup%
\begin{picture}(19571,8030)(384,-9017)
\put(5626,-4156){\makebox(0,0)[lb]{\smash{{\SetFigFont{25}{30.0}{\rmdefault}{\mddefault}{\updefault}$c_1$}}}}
\put(5806,-5596){\makebox(0,0)[lb]{\smash{{\SetFigFont{25}{30.0}{\rmdefault}{\mddefault}{\updefault}$c_2$}}}}
\put(8371,-7441){\makebox(0,0)[lb]{\smash{{\SetFigFont{25}{30.0}{\rmdefault}{\mddefault}{\updefault}$c_5$}}}}
\put(11026,-2851){\makebox(0,0)[lb]{\smash{{\SetFigFont{25}{30.0}{\rmdefault}{\mddefault}{\updefault}$c_4$}}}}
\put(11071,-3931){\makebox(0,0)[lb]{\smash{{\SetFigFont{25}{30.0}{\rmdefault}{\mddefault}{\updefault}$c_3$}}}}
\put(14491,-3121){\makebox(0,0)[lb]{\smash{{\SetFigFont{25}{30.0}{\rmdefault}{\mddefault}{\updefault}$c_1$}}}}
\put(14086,-5371){\makebox(0,0)[lb]{\smash{{\SetFigFont{25}{30.0}{\rmdefault}{\mddefault}{\updefault}$c_2$}}}}
\put(14716,-7171){\makebox(0,0)[lb]{\smash{{\SetFigFont{25}{30.0}{\rmdefault}{\mddefault}{\updefault}$c_3$}}}}
\put(16471,-5056){\makebox(0,0)[lb]{\smash{{\SetFigFont{25}{30.0}{\rmdefault}{\mddefault}{\updefault}$c_5$}}}}
\put(17416,-3346){\makebox(0,0)[lb]{\smash{{\SetFigFont{25}{30.0}{\rmdefault}{\mddefault}{\updefault}$c_6$}}}}
\put(16921,-6586){\makebox(0,0)[lb]{\smash{{\SetFigFont{25}{30.0}{\rmdefault}{\mddefault}{\updefault}$c_4$}}}}
\put(17686,-1861){\makebox(0,0)[lb]{\smash{{\SetFigFont{25}{30.0}{\rmdefault}{\mddefault}{\updefault}$c_7$}}}}
\put(856,-3661){\makebox(0,0)[lb]{\smash{{\SetFigFont{25}{30.0}{\rmdefault}{\mddefault}{\updefault}$c_1$}}}}
\put(2206,-3661){\makebox(0,0)[lb]{\smash{{\SetFigFont{25}{30.0}{\rmdefault}{\mddefault}{\updefault}$c_2$}}}}
\put(3556,-3661){\makebox(0,0)[lb]{\smash{{\SetFigFont{25}{30.0}{\rmdefault}{\mddefault}{\updefault}$c_3$}}}}
\put(811,-4966){\makebox(0,0)[lb]{\smash{{\SetFigFont{25}{30.0}{\rmdefault}{\mddefault}{\updefault}$c_4$}}}}
\put(2251,-4966){\makebox(0,0)[lb]{\smash{{\SetFigFont{25}{30.0}{\rmdefault}{\mddefault}{\updefault}$c_5$}}}}
\put(3646,-4921){\makebox(0,0)[lb]{\smash{{\SetFigFont{25}{30.0}{\rmdefault}{\mddefault}{\updefault}$c_6$}}}}
\put(901,-6271){\makebox(0,0)[lb]{\smash{{\SetFigFont{25}{30.0}{\rmdefault}{\mddefault}{\updefault}$c_7$}}}}
\put(2296,-6271){\makebox(0,0)[lb]{\smash{{\SetFigFont{25}{30.0}{\rmdefault}{\mddefault}{\updefault}$c_8$}}}}
\put(3691,-6226){\makebox(0,0)[lb]{\smash{{\SetFigFont{25}{30.0}{\rmdefault}{\mddefault}{\updefault}$c_9$}}}}
\put(721,-8116){\makebox(0,0)[lb]{\smash{{\SetFigFont{29}{34.8}{\rmdefault}{\mddefault}{\updefault}$C=24-16+1=9$}}}}
\put(7336,-8521){\makebox(0,0)[lb]{\smash{{\SetFigFont{29}{34.8}{\rmdefault}{\mddefault}{\updefault}$C=12-8+1=5$}}}}
\put(14941,-8881){\makebox(0,0)[lb]{\smash{{\SetFigFont{29}{34.8}{\rmdefault}{\mddefault}{\updefault}$C=26-20+1=7$}}}}
\end{picture}%

%% file: dual_graph.pdf_t
\begin{picture}(0,0)%
\includegraphics{dual_graph.pdf}%
\end{picture}%
\setlength{\unitlength}{4144sp}%
\begingroup\makeatletter\ifx\SetFigFont\undefined%
\gdef\SetFigFont#1#2#3#4#5{%
  \reset@font\fontsize{#1}{#2pt}%
  \fontfamily{#3}\fontseries{#4}\fontshape{#5}%
  \selectfont}%
\fi\endgroup%
\begin{picture}(11322,5721)(1054,-5326)
\put(8506,119){\makebox(0,0)[lb]{\smash{{\SetFigFont{20}{24.0}{\rmdefault}{\mddefault}{\updefault}$\G^\star$}}}}
\put(3376,164){\makebox(0,0)[lb]{\smash{{\SetFigFont{20}{24.0}{\rmdefault}{\mddefault}{\updefault}$\G$}}}}
\end{picture}%

%% file: locdual_graph.pdf_t
\begin{picture}(0,0)%
\includegraphics{locdual_graph.pdf}%
\end{picture}%
\setlength{\unitlength}{4144sp}%
\begingroup\makeatletter\ifx\SetFigFont\undefined%
\gdef\SetFigFont#1#2#3#4#5{%
  \reset@font\fontsize{#1}{#2pt}%
  \fontfamily{#3}\fontseries{#4}\fontshape{#5}%
  \selectfont}%
\fi\endgroup%
\begin{picture}(15039,5436)(1474,-4981)
\put(2686,-2821){\makebox(0,0)[lb]{\smash{{\SetFigFont{20}{24.0}{\rmdefault}{\mddefault}{\updefault}$v$}}}}
\put(2026,224){\makebox(0,0)[lb]{\smash{{\SetFigFont{20}{24.0}{\rmdefault}{\mddefault}{\updefault}$\G$}}}}
\put(7066,194){\makebox(0,0)[lb]{\smash{{\SetFigFont{20}{24.0}{\rmdefault}{\mddefault}{\updefault}$\G^\star$}}}}
\put(13726,119){\makebox(0,0)[lb]{\smash{{\SetFigFont{20}{24.0}{\rmdefault}{\mddefault}{\updefault}$\G_v^\star$}}}}
\end{picture}%

%% file: hfg_ex.pdf_t
\begin{picture}(0,0)%
\includegraphics{hfg_ex.pdf}%
\end{picture}%
\setlength{\unitlength}{4144sp}%
\begingroup\makeatletter\ifx\SetFigFont\undefined%
\gdef\SetFigFont#1#2#3#4#5{%
  \reset@font\fontsize{#1}{#2pt}%
  \fontfamily{#3}\fontseries{#4}\fontshape{#5}%
  \selectfont}%
\fi\endgroup%
\begin{picture}(11072,6684)(2592,-6568)
\put(7444,-1964){\makebox(0,0)[lb]{\smash{{\SetFigFont{17}{20.4}{\rmdefault}{\mddefault}{\updefault}$-1$}}}}
\put(8329,-1049){\makebox(0,0)[lb]{\smash{{\SetFigFont{17}{20.4}{\rmdefault}{\mddefault}{\updefault}$-1$}}}}
\put(9094,-1934){\makebox(0,0)[lb]{\smash{{\SetFigFont{17}{20.4}{\rmdefault}{\mddefault}{\updefault}$-1$}}}}
\put(8344,-2819){\makebox(0,0)[lb]{\smash{{\SetFigFont{17}{20.4}{\rmdefault}{\mddefault}{\updefault}$-1$}}}}
\put(8689,-3674){\makebox(0,0)[lb]{\smash{{\SetFigFont{17}{20.4}{\rmdefault}{\mddefault}{\updefault}$1/4$}}}}
\put(10204,-2219){\makebox(0,0)[lb]{\smash{{\SetFigFont{17}{20.4}{\rmdefault}{\mddefault}{\updefault}$1/4$}}}}
\put(8599,-149){\makebox(0,0)[lb]{\smash{{\SetFigFont{17}{20.4}{\rmdefault}{\mddefault}{\updefault}$1/4$}}}}
\put(6529,-2279){\makebox(0,0)[lb]{\smash{{\SetFigFont{17}{20.4}{\rmdefault}{\mddefault}{\updefault}$1/4$}}}}
\put(8236,-4576){\makebox(0,0)[lb]{\smash{{\SetFigFont{17}{20.4}{\rmdefault}{\mddefault}{\updefault}$-1$}}}}
\put(8596,-5566){\makebox(0,0)[lb]{\smash{{\SetFigFont{17}{20.4}{\rmdefault}{\mddefault}{\updefault}$-1$}}}}
\put(7111,-6106){\makebox(0,0)[lb]{\smash{{\SetFigFont{17}{20.4}{\rmdefault}{\mddefault}{\updefault}$0$}}}}
\put(11491,-3091){\makebox(0,0)[lb]{\smash{{\SetFigFont{20}{24.0}{\rmdefault}{\mddefault}{\updefault}$\ell-$node:}}}}
\put(11536,-2371){\makebox(0,0)[lb]{\smash{{\SetFigFont{20}{24.0}{\rmdefault}{\mddefault}{\updefault}$c-$node:}}}}
\put(11491,-3856){\makebox(0,0)[lb]{\smash{{\SetFigFont{20}{24.0}{\rmdefault}{\mddefault}{\updefault}$v-$node:}}}}
\end{picture}%

%% file: dloop.pdf_t
\begin{picture}(0,0)%
\includegraphics{dloop.pdf}%
\end{picture}%
\setlength{\unitlength}{4144sp}%
\begingroup\makeatletter\ifx\SetFigFont\undefined%
\gdef\SetFigFont#1#2#3#4#5{%
  \reset@font\fontsize{#1}{#2pt}%
  \fontfamily{#3}\fontseries{#4}\fontshape{#5}%
  \selectfont}%
\fi\endgroup%
\begin{picture}(20792,12452)(641,-9208)
\put(852,-7246){\makebox(0,0)[lb]{\smash{{\SetFigFont{20}{24.0}{\rmdefault}{\mddefault}{\updefault}$\ell_3$}}}}
\put(3387,-8460){\makebox(0,0)[lb]{\smash{{\SetFigFont{20}{24.0}{\rmdefault}{\mddefault}{\updefault}$v$}}}}
\put(3402,-7201){\makebox(0,0)[lb]{\smash{{\SetFigFont{20}{24.0}{\rmdefault}{\mddefault}{\updefault}$\ell_2$}}}}
\put(5517,-7215){\makebox(0,0)[lb]{\smash{{\SetFigFont{20}{24.0}{\rmdefault}{\mddefault}{\updefault}$\ell_1$}}}}
\put(5532,-5835){\makebox(0,0)[lb]{\smash{{\SetFigFont{20}{24.0}{\rmdefault}{\mddefault}{\updefault}$c_3$}}}}
\put(867,-5850){\makebox(0,0)[lb]{\smash{{\SetFigFont{20}{24.0}{\rmdefault}{\mddefault}{\updefault}$c_1$}}}}
\put(3342,-5820){\makebox(0,0)[lb]{\smash{{\SetFigFont{20}{24.0}{\rmdefault}{\mddefault}{\updefault}$c_2$}}}}
\put(8026,-7201){\makebox(0,0)[lb]{\smash{{\SetFigFont{20}{24.0}{\rmdefault}{\mddefault}{\updefault}$\ell_3$}}}}
\put(18284,-8948){\makebox(0,0)[lb]{\smash{{\SetFigFont{20}{24.0}{\rmdefault}{\mddefault}{\updefault}$c_1$}}}}
\put(15299,-4643){\makebox(0,0)[lb]{\smash{{\SetFigFont{20}{24.0}{\rmdefault}{\mddefault}{\updefault}$c_2$}}}}
\put(20939,-4673){\makebox(0,0)[lb]{\smash{{\SetFigFont{20}{24.0}{\rmdefault}{\mddefault}{\updefault}$c_3$}}}}
\put(19679,-7253){\makebox(0,0)[lb]{\smash{{\SetFigFont{20}{24.0}{\rmdefault}{\mddefault}{\updefault}$\ell_2$}}}}
\put(16784,-7328){\makebox(0,0)[lb]{\smash{{\SetFigFont{20}{24.0}{\rmdefault}{\mddefault}{\updefault}$\ell_3$}}}}
\put(17697,-6878){\makebox(0,0)[lb]{\smash{{\SetFigFont{20}{24.0}{\rmdefault}{\mddefault}{\updefault}$v_3^\star$}}}}
\put(18164,-5003){\makebox(0,0)[lb]{\smash{{\SetFigFont{20}{24.0}{\rmdefault}{\mddefault}{\updefault}$\ell_1$}}}}
\put(18237,-6128){\makebox(0,0)[lb]{\smash{{\SetFigFont{20}{24.0}{\rmdefault}{\mddefault}{\updefault}$v_1^\star$}}}}
\put(18672,-6893){\makebox(0,0)[lb]{\smash{{\SetFigFont{20}{24.0}{\rmdefault}{\mddefault}{\updefault}$v_2^\star$}}}}
\put(2251,-4336){\makebox(0,0)[lb]{\smash{{\SetFigFont{25}{30.0}{\rmdefault}{\mddefault}{\updefault}Parent to Child}}}}
\put(10561,-8415){\makebox(0,0)[lb]{\smash{{\SetFigFont{20}{24.0}{\rmdefault}{\mddefault}{\updefault}$v$}}}}
\put(10576,-7156){\makebox(0,0)[lb]{\smash{{\SetFigFont{20}{24.0}{\rmdefault}{\mddefault}{\updefault}$\ell_2$}}}}
\put(12691,-7170){\makebox(0,0)[lb]{\smash{{\SetFigFont{20}{24.0}{\rmdefault}{\mddefault}{\updefault}$\ell_1$}}}}
\put(12706,-5790){\makebox(0,0)[lb]{\smash{{\SetFigFont{20}{24.0}{\rmdefault}{\mddefault}{\updefault}$c_3$}}}}
\put(8041,-5805){\makebox(0,0)[lb]{\smash{{\SetFigFont{20}{24.0}{\rmdefault}{\mddefault}{\updefault}$c_1$}}}}
\put(10516,-5775){\makebox(0,0)[lb]{\smash{{\SetFigFont{20}{24.0}{\rmdefault}{\mddefault}{\updefault}$c_2$}}}}
\put(8641,-4381){\makebox(0,0)[lb]{\smash{{\SetFigFont{25}{30.0}{\rmdefault}{\mddefault}{\updefault}Minimal Region graph}}}}
\put(16426,-3616){\makebox(0,0)[lb]{\smash{{\SetFigFont{25}{30.0}{\rmdefault}{\mddefault}{\updefault}Mixed Factor Graph}}}}
\put(10801,1041){\makebox(0,0)[lb]{\smash{{\SetFigFont{29}{34.8}{\rmdefault}{\mddefault}{\updefault}$\ell_3$}}}}
\put(11088,-348){\makebox(0,0)[lb]{\smash{{\SetFigFont{29}{34.8}{\rmdefault}{\mddefault}{\updefault}$v$}}}}
\put(8911, 74){\makebox(0,0)[lb]{\smash{{\SetFigFont{29}{34.8}{\rmdefault}{\mddefault}{\updefault}$c_1$}}}}
\put(12466, 74){\makebox(0,0)[lb]{\smash{{\SetFigFont{29}{34.8}{\rmdefault}{\mddefault}{\updefault}$c_2$}}}}
\put(8911,-1321){\makebox(0,0)[lb]{\smash{{\SetFigFont{29}{34.8}{\rmdefault}{\mddefault}{\updefault}$\ell_2$}}}}
\put(12286,-1366){\makebox(0,0)[lb]{\smash{{\SetFigFont{29}{34.8}{\rmdefault}{\mddefault}{\updefault}$\ell_1$}}}}
\put(10621,-1996){\makebox(0,0)[lb]{\smash{{\SetFigFont{29}{34.8}{\rmdefault}{\mddefault}{\updefault}$c_3$}}}}
\end{picture}%

%% file: cycle_comb.pdf_t
\begin{picture}(0,0)%
\includegraphics{cycle_comb.pdf}%
\end{picture}%
\setlength{\unitlength}{4144sp}%
\begingroup\makeatletter\ifx\SetFigFont\undefined%
\gdef\SetFigFont#1#2#3#4#5{%
  \reset@font\fontsize{#1}{#2pt}%
  \fontfamily{#3}\fontseries{#4}\fontshape{#5}%
  \selectfont}%
\fi\endgroup%
\begin{picture}(12368,4534)(2547,-5327)
\put(11881,-5191){\makebox(0,0)[lb]{\smash{{\SetFigFont{29}{34.8}{\rmdefault}{\mddefault}{\updefault}$S = 4+5$ }}}}
\put(3646,-5191){\makebox(0,0)[lb]{\smash{{\SetFigFont{29}{34.8}{\rmdefault}{\mddefault}{\updefault}$S = 4+6$ }}}}
\end{picture}%

%% file: ctol.pdf_t
\begin{picture}(0,0)%
\includegraphics{ctol.pdf}%
\end{picture}%
\setlength{\unitlength}{4144sp}%
\begingroup\makeatletter\ifx\SetFigFont\undefined%
\gdef\SetFigFont#1#2#3#4#5{%
  \reset@font\fontsize{#1}{#2pt}%
  \fontfamily{#3}\fontseries{#4}\fontshape{#5}%
  \selectfont}%
\fi\endgroup%
\begin{picture}(9307,8644)(3181,-9170)
\put(5446,-5416){\makebox(0,0)[lb]{\smash{{\SetFigFont{34}{40.8}{\rmdefault}{\mddefault}{\updefault}$\ell_1$}}}}
\put(8236,-1141){\makebox(0,0)[lb]{\smash{{\SetFigFont{50}{60.0}{\rmdefault}{\mddefault}{\updefault}{\blue $n_{\ell_2\to c}$}}}}}
\put(8146,-4831){\makebox(0,0)[lb]{\smash{{\SetFigFont{34}{40.8}{\rmdefault}{\mddefault}{\updefault}$c$}}}}
\put(11836,-4516){\makebox(0,0)[lb]{\smash{{\SetFigFont{50}{60.0}{\rmdefault}{\mddefault}{\updefault}{\blue $n_{\ell_3\to c}$}}}}}
\put(3196,-4741){\makebox(0,0)[lb]{\smash{{\SetFigFont{50}{60.0}{\rmdefault}{\mddefault}{\updefault}{\blue $n_{\ell_1\to c}$}}}}}
\put(7831,-8836){\makebox(0,0)[lb]{\smash{{\SetFigFont{70}{84.0}{\rmdefault}{\mddefault}{\updefault}{\red $m_{c\to\ell}$}}}}}
\put(9901,-2761){\makebox(0,0)[lb]{\smash{{\SetFigFont{34}{40.8}{\rmdefault}{\mddefault}{\updefault}$M^{(i)}$}}}}
\put(7921,-2581){\makebox(0,0)[lb]{\smash{{\SetFigFont{34}{40.8}{\rmdefault}{\mddefault}{\updefault}$\ell_2$}}}}
\put(8461,-6991){\makebox(0,0)[lb]{\smash{{\SetFigFont{70}{84.0}{\rmdefault}{\mddefault}{\updefault}$\ell$}}}}
\put(10801,-4021){\makebox(0,0)[lb]{\smash{{\SetFigFont{34}{40.8}{\rmdefault}{\mddefault}{\updefault}$\ell_3$}}}}
\end{picture}%

%% file: LoopPlot.tex
\begingroup
  \makeatletter
  \providecommand\color[2][]{%
    \GenericError{(gnuplot) \space\space\space\@spaces}{%
      Package color not loaded in conjunction with
      terminal option `colourtext'%
    }{See the gnuplot documentation for explanation.%
    }{Either use 'blacktext' in gnuplot or load the package
      color.sty in LaTeX.}%
    \renewcommand\color[2][]{}%
  }%
  \providecommand\includegraphics[2][]{%
    \GenericError{(gnuplot) \space\space\space\@spaces}{%
      Package graphicx or graphics not loaded%
    }{See the gnuplot documentation for explanation.%
    }{The gnuplot epslatex terminal needs graphicx.sty or graphics.sty.}%
    \renewcommand\includegraphics[2][]{}%
  }%
  \providecommand\rotatebox[2]{#2}%
  \@ifundefined{ifGPcolor}{%
    \newif\ifGPcolor
    \GPcolortrue
  }{}%
  \@ifundefined{ifGPblacktext}{%
    \newif\ifGPblacktext
    \GPblacktexttrue
  }{}%
  \let\gplgaddtomacro\g@addto@macro
  \gdef\gplbacktext{}%
  \gdef\gplfronttext{}%
  \makeatother
  \ifGPblacktext
    \def\colorrgb#1{}%
    \def\colorgray#1{}%
  \else
    \ifGPcolor
      \def\colorrgb#1{\color[rgb]{#1}}%
      \def\colorgray#1{\color[gray]{#1}}%
      \expandafter\def\csname LTw\endcsname{\color{white}}%
      \expandafter\def\csname LTb\endcsname{\color{black}}%
      \expandafter\def\csname LTa\endcsname{\color{black}}%
      \expandafter\def\csname LT0\endcsname{\color[rgb]{1,0,0}}%
      \expandafter\def\csname LT1\endcsname{\color[rgb]{0,1,0}}%
      \expandafter\def\csname LT2\endcsname{\color[rgb]{0,0,1}}%
      \expandafter\def\csname LT3\endcsname{\color[rgb]{1,0,1}}%
      \expandafter\def\csname LT4\endcsname{\color[rgb]{0,1,1}}%
      \expandafter\def\csname LT5\endcsname{\color[rgb]{1,1,0}}%
      \expandafter\def\csname LT6\endcsname{\color[rgb]{0,0,0}}%
      \expandafter\def\csname LT7\endcsname{\color[rgb]{1,0.3,0}}%
      \expandafter\def\csname LT8\endcsname{\color[rgb]{0.5,0.5,0.5}}%
    \else
      \def\colorrgb#1{\color{black}}%
      \def\colorgray#1{\color[gray]{#1}}%
      \expandafter\def\csname LTw\endcsname{\color{white}}%
      \expandafter\def\csname LTb\endcsname{\color{black}}%
      \expandafter\def\csname LTa\endcsname{\color{black}}%
      \expandafter\def\csname LT0\endcsname{\color{black}}%
      \expandafter\def\csname LT1\endcsname{\color{black}}%
      \expandafter\def\csname LT2\endcsname{\color{black}}%
      \expandafter\def\csname LT3\endcsname{\color{black}}%
      \expandafter\def\csname LT4\endcsname{\color{black}}%
      \expandafter\def\csname LT5\endcsname{\color{black}}%
      \expandafter\def\csname LT6\endcsname{\color{black}}%
      \expandafter\def\csname LT7\endcsname{\color{black}}%
      \expandafter\def\csname LT8\endcsname{\color{black}}%
    \fi
  \fi
  \setlength{\unitlength}{0.0500bp}%
  \begin{picture}(7200.00,5040.00)%
    \gplgaddtomacro\gplbacktext{%
      \csname LTb\endcsname%
      \put(198,770){\makebox(0,0)[r]{\strut{} 0}}%
      \put(198,1583){\makebox(0,0)[r]{\strut{} 0.2}}%
      \put(198,2396){\makebox(0,0)[r]{\strut{} 0.4}}%
      \put(198,3209){\makebox(0,0)[r]{\strut{} 0.6}}%
      \put(198,4023){\makebox(0,0)[r]{\strut{} 0.8}}%
      \put(198,4836){\makebox(0,0)[r]{\strut{} 1}}%
      \put(330,550){\makebox(0,0){\strut{} 0}}%
      \put(1274,550){\makebox(0,0){\strut{} 1}}%
      \put(2217,550){\makebox(0,0){\strut{} 2}}%
      \put(3161,550){\makebox(0,0){\strut{} 3}}%
      \put(4104,550){\makebox(0,0){\strut{} 4}}%
      \put(5048,550){\makebox(0,0){\strut{} 5}}%
      \put(5991,550){\makebox(0,0){\strut{} 6}}%
      \put(6935,550){\makebox(0,0){\strut{} 7}}%
      \put(-308,2904){\rotatebox{-270}{\makebox(0,0){\strut{}\Large Success rate}}}%
      \put(3632,220){\makebox(0,0){\strut{}\Large $\beta$}}%
    }%
    \gplgaddtomacro\gplfronttext{%
      \csname LTb\endcsname%
      \put(5948,4866){\makebox(0,0)[r]{\strut{}LS}}%
      \csname LTb\endcsname%
      \put(5948,4646){\makebox(0,0)[r]{\strut{}$S=3$\hspace{0.4cm}FP}}%
      \csname LTb\endcsname%
      \put(5948,4426){\makebox(0,0)[r]{\strut{}LS$+$FP}}%
      \csname LTb\endcsname%
      \put(5948,4206){\makebox(0,0)[r]{\strut{}LS}}%
      \csname LTb\endcsname%
      \put(5948,3986){\makebox(0,0)[r]{\strut{}$S=4$\hspace{0.4cm}FP}}%
      \csname LTb\endcsname%
      \put(5948,3766){\makebox(0,0)[r]{\strut{}LS$+$FP}}%
      \csname LTb\endcsname%
      \put(5948,3546){\makebox(0,0)[r]{\strut{}LS}}%
      \csname LTb\endcsname%
      \put(5948,3326){\makebox(0,0)[r]{\strut{}$S=5$\hspace{0.4cm}FP}}%
      \csname LTb\endcsname%
      \put(5948,3106){\makebox(0,0)[r]{\strut{}LS$+$FP}}%
      \csname LTb\endcsname%
      \put(5948,2886){\makebox(0,0)[r]{\strut{}LS}}%
      \csname LTb\endcsname%
      \put(5948,2666){\makebox(0,0)[r]{\strut{}$S=6$\hspace{0.4cm}FP}}%
      \csname LTb\endcsname%
      \put(5948,2446){\makebox(0,0)[r]{\strut{}LS$+$FP}}%
      \csname LTb\endcsname%
      \put(5948,2226){\makebox(0,0)[r]{\strut{}LS}}%
      \csname LTb\endcsname%
      \put(5948,2006){\makebox(0,0)[r]{\strut{}$S=7$\hspace{0.4cm}FP}}%
      \csname LTb\endcsname%
      \put(5948,1786){\makebox(0,0)[r]{\strut{}LS$+$FP}}%
      \csname LTb\endcsname%
      \put(5948,1566){\makebox(0,0)[r]{\strut{}LS}}%
      \csname LTb\endcsname%
      \put(5948,1346){\makebox(0,0)[r]{\strut{}$S=8$\hspace{0.4cm}FP}}%
      \csname LTb\endcsname%
      \put(5948,1126){\makebox(0,0)[r]{\strut{}LS$+$FP}}%
    }%
    \gplbacktext
    \put(0,0){\includegraphics{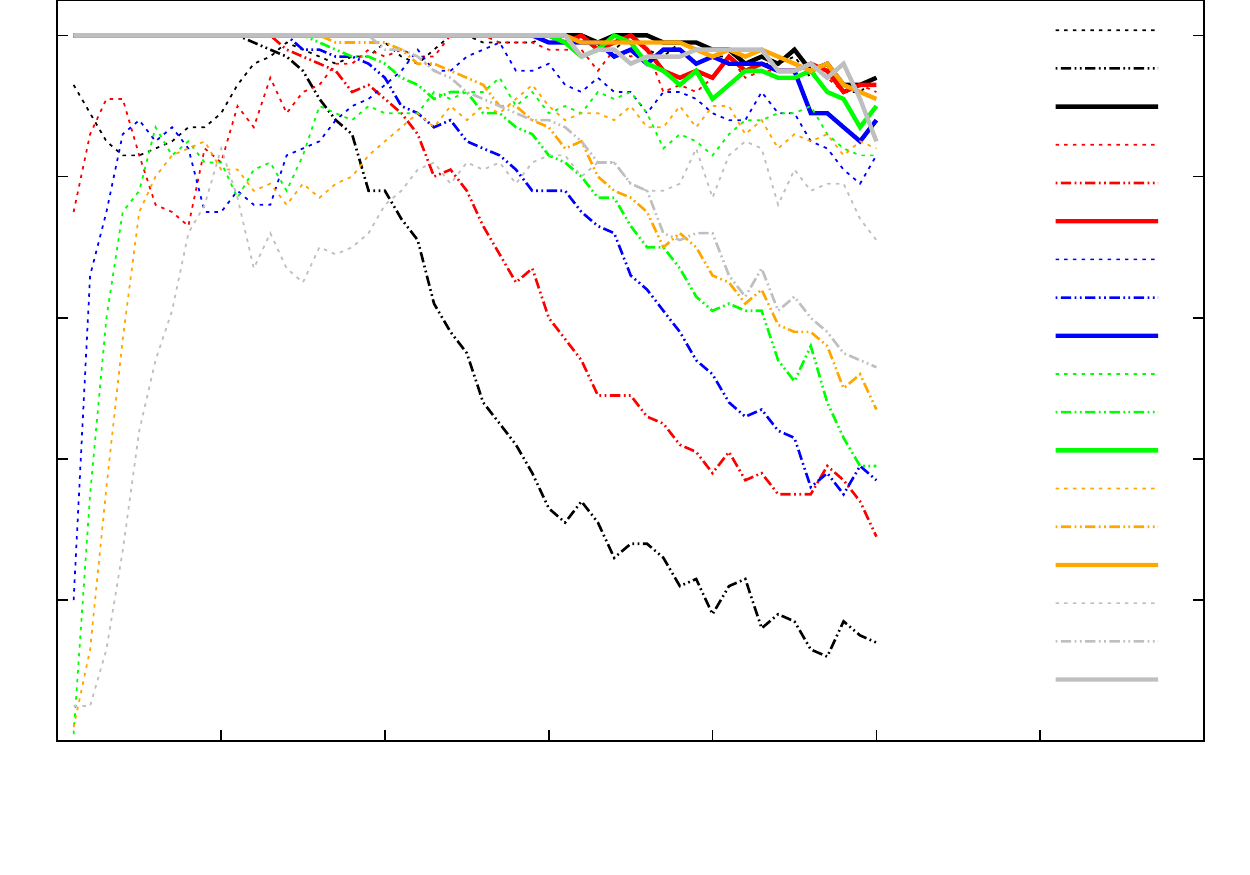}}%
    \gplfronttext
  \end{picture}%
\endgroup

%% file: grid_I.tex
\begingroup
  \makeatletter
  \providecommand\color[2][]{%
    \GenericError{(gnuplot) \space\space\space\@spaces}{%
      Package color not loaded in conjunction with
      terminal option `colourtext'%
    }{See the gnuplot documentation for explanation.%
    }{Either use 'blacktext' in gnuplot or load the package
      color.sty in LaTeX.}%
    \renewcommand\color[2][]{}%
  }%
  \providecommand\includegraphics[2][]{%
    \GenericError{(gnuplot) \space\space\space\@spaces}{%
      Package graphicx or graphics not loaded%
    }{See the gnuplot documentation for explanation.%
    }{The gnuplot epslatex terminal needs graphicx.sty or graphics.sty.}%
    \renewcommand\includegraphics[2][]{}%
  }%
  \providecommand\rotatebox[2]{#2}%
  \@ifundefined{ifGPcolor}{%
    \newif\ifGPcolor
    \GPcolortrue
  }{}%
  \@ifundefined{ifGPblacktext}{%
    \newif\ifGPblacktext
    \GPblacktexttrue
  }{}%
  \let\gplgaddtomacro\g@addto@macro
  \gdef\gplbacktext{}%
  \gdef\gplfronttext{}%
  \makeatother
  \ifGPblacktext
    \def\colorrgb#1{}%
    \def\colorgray#1{}%
  \else
    \ifGPcolor
      \def\colorrgb#1{\color[rgb]{#1}}%
      \def\colorgray#1{\color[gray]{#1}}%
      \expandafter\def\csname LTw\endcsname{\color{white}}%
      \expandafter\def\csname LTb\endcsname{\color{black}}%
      \expandafter\def\csname LTa\endcsname{\color{black}}%
      \expandafter\def\csname LT0\endcsname{\color[rgb]{1,0,0}}%
      \expandafter\def\csname LT1\endcsname{\color[rgb]{0,1,0}}%
      \expandafter\def\csname LT2\endcsname{\color[rgb]{0,0,1}}%
      \expandafter\def\csname LT3\endcsname{\color[rgb]{1,0,1}}%
      \expandafter\def\csname LT4\endcsname{\color[rgb]{0,1,1}}%
      \expandafter\def\csname LT5\endcsname{\color[rgb]{1,1,0}}%
      \expandafter\def\csname LT6\endcsname{\color[rgb]{0,0,0}}%
      \expandafter\def\csname LT7\endcsname{\color[rgb]{1,0.3,0}}%
      \expandafter\def\csname LT8\endcsname{\color[rgb]{0.5,0.5,0.5}}%
    \else
      \def\colorrgb#1{\color{black}}%
      \def\colorgray#1{\color[gray]{#1}}%
      \expandafter\def\csname LTw\endcsname{\color{white}}%
      \expandafter\def\csname LTb\endcsname{\color{black}}%
      \expandafter\def\csname LTa\endcsname{\color{black}}%
      \expandafter\def\csname LT0\endcsname{\color{black}}%
      \expandafter\def\csname LT1\endcsname{\color{black}}%
      \expandafter\def\csname LT2\endcsname{\color{black}}%
      \expandafter\def\csname LT3\endcsname{\color{black}}%
      \expandafter\def\csname LT4\endcsname{\color{black}}%
      \expandafter\def\csname LT5\endcsname{\color{black}}%
      \expandafter\def\csname LT6\endcsname{\color{black}}%
      \expandafter\def\csname LT7\endcsname{\color{black}}%
      \expandafter\def\csname LT8\endcsname{\color{black}}%
    \fi
  \fi
  \setlength{\unitlength}{0.0500bp}%
  \begin{picture}(7200.00,5040.00)%
    \gplgaddtomacro\gplbacktext{%
      \csname LTb\endcsname%
      \put(198,770){\makebox(0,0)[r]{\strut{} 1e-06}}%
      \put(198,1481){\makebox(0,0)[r]{\strut{} 1e-05}}%
      \put(198,2193){\makebox(0,0)[r]{\strut{} 0.0001}}%
      \put(198,2904){\makebox(0,0)[r]{\strut{} 0.001}}%
      \put(198,3616){\makebox(0,0)[r]{\strut{} 0.01}}%
      \put(198,4328){\makebox(0,0)[r]{\strut{} 0.1}}%
      \put(198,5039){\makebox(0,0)[r]{\strut{} 1}}%
      \put(330,550){\makebox(0,0){\strut{} 0}}%
      \put(1651,550){\makebox(0,0){\strut{} 1}}%
      \put(2972,550){\makebox(0,0){\strut{} 2}}%
      \put(4293,550){\makebox(0,0){\strut{} 3}}%
      \put(5614,550){\makebox(0,0){\strut{} 4}}%
      \put(6935,550){\makebox(0,0){\strut{} 5}}%
      \put(-968,2904){\rotatebox{-270}{\makebox(0,0){\strut{}\Large Mean $L_2$ Error}}}%
      \put(3632,220){\makebox(0,0){\strut{}\Large $\beta$}}%
    }%
    \gplgaddtomacro\gplfronttext{%
      \csname LTb\endcsname%
      \put(4759,2083){\makebox(0,0)[r]{\strut{}\large GCBP Err  (bias)}}%
      \csname LTb\endcsname%
      \put(4759,1863){\makebox(0,0)[r]{\strut{}\large GCBP Err (no bias)}}%
      \csname LTb\endcsname%
      \put(4759,1643){\makebox(0,0)[r]{\strut{}\large BP Err  (bias)}}%
      \csname LTb\endcsname%
      \put(4759,1423){\makebox(0,0)[r]{\strut{}\large BP Err (no bias)}}%
    }%
    \gplbacktext
    \put(0,0){\includegraphics{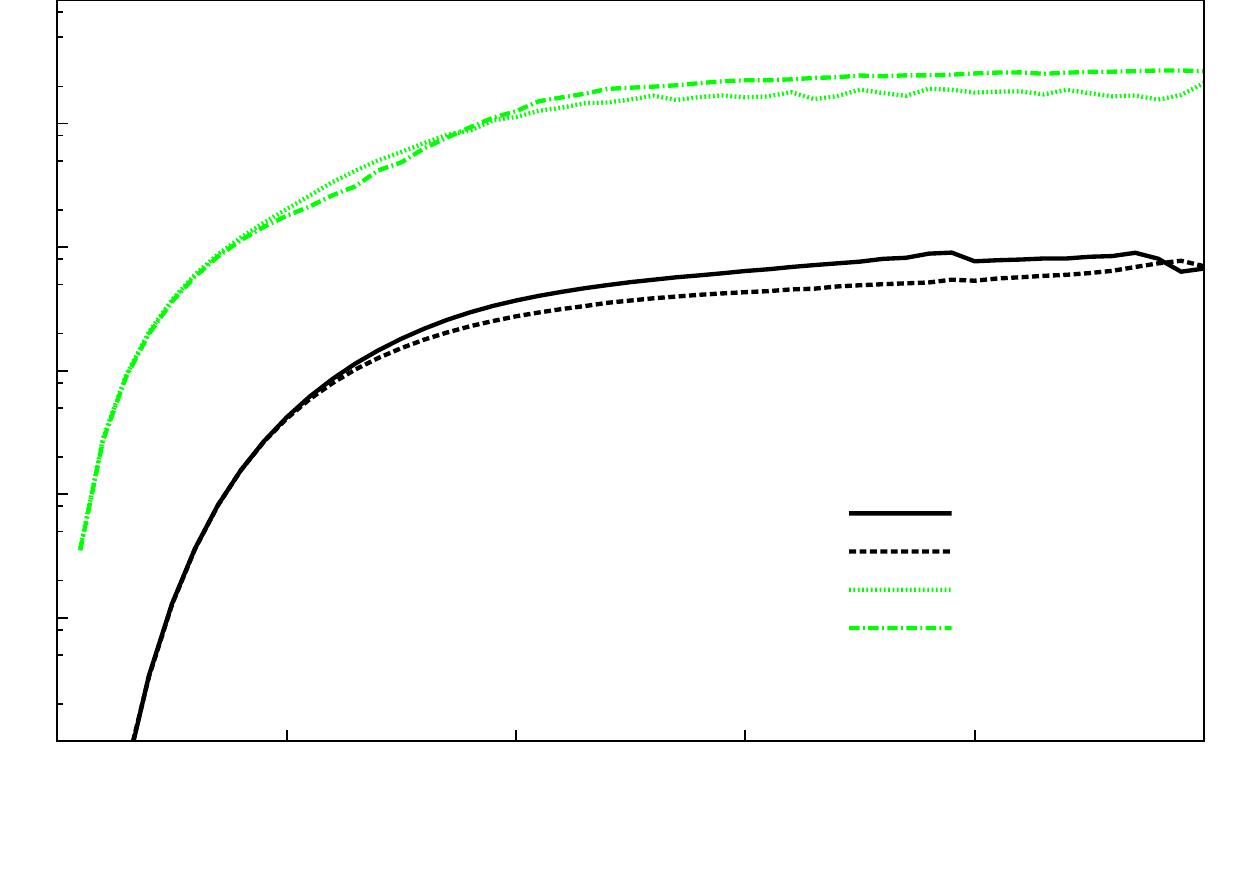}}%
    \gplfronttext
  \end{picture}%
\endgroup

%% file: bip_dirK4.tex
\begingroup
  \makeatletter
  \providecommand\color[2][]{%
    \GenericError{(gnuplot) \space\space\space\@spaces}{%
      Package color not loaded in conjunction with
      terminal option `colourtext'%
    }{See the gnuplot documentation for explanation.%
    }{Either use 'blacktext' in gnuplot or load the package
      color.sty in LaTeX.}%
    \renewcommand\color[2][]{}%
  }%
  \providecommand\includegraphics[2][]{%
    \GenericError{(gnuplot) \space\space\space\@spaces}{%
      Package graphicx or graphics not loaded%
    }{See the gnuplot documentation for explanation.%
    }{The gnuplot epslatex terminal needs graphicx.sty or graphics.sty.}%
    \renewcommand\includegraphics[2][]{}%
  }%
  \providecommand\rotatebox[2]{#2}%
  \@ifundefined{ifGPcolor}{%
    \newif\ifGPcolor
    \GPcolortrue
  }{}%
  \@ifundefined{ifGPblacktext}{%
    \newif\ifGPblacktext
    \GPblacktexttrue
  }{}%
  \let\gplgaddtomacro\g@addto@macro
  \gdef\gplbacktext{}%
  \gdef\gplfronttext{}%
  \makeatother
  \ifGPblacktext
    \def\colorrgb#1{}%
    \def\colorgray#1{}%
  \else
    \ifGPcolor
      \def\colorrgb#1{\color[rgb]{#1}}%
      \def\colorgray#1{\color[gray]{#1}}%
      \expandafter\def\csname LTw\endcsname{\color{white}}%
      \expandafter\def\csname LTb\endcsname{\color{black}}%
      \expandafter\def\csname LTa\endcsname{\color{black}}%
      \expandafter\def\csname LT0\endcsname{\color[rgb]{1,0,0}}%
      \expandafter\def\csname LT1\endcsname{\color[rgb]{0,1,0}}%
      \expandafter\def\csname LT2\endcsname{\color[rgb]{0,0,1}}%
      \expandafter\def\csname LT3\endcsname{\color[rgb]{1,0,1}}%
      \expandafter\def\csname LT4\endcsname{\color[rgb]{0,1,1}}%
      \expandafter\def\csname LT5\endcsname{\color[rgb]{1,1,0}}%
      \expandafter\def\csname LT6\endcsname{\color[rgb]{0,0,0}}%
      \expandafter\def\csname LT7\endcsname{\color[rgb]{1,0.3,0}}%
      \expandafter\def\csname LT8\endcsname{\color[rgb]{0.5,0.5,0.5}}%
    \else
      \def\colorrgb#1{\color{black}}%
      \def\colorgray#1{\color[gray]{#1}}%
      \expandafter\def\csname LTw\endcsname{\color{white}}%
      \expandafter\def\csname LTb\endcsname{\color{black}}%
      \expandafter\def\csname LTa\endcsname{\color{black}}%
      \expandafter\def\csname LT0\endcsname{\color{black}}%
      \expandafter\def\csname LT1\endcsname{\color{black}}%
      \expandafter\def\csname LT2\endcsname{\color{black}}%
      \expandafter\def\csname LT3\endcsname{\color{black}}%
      \expandafter\def\csname LT4\endcsname{\color{black}}%
      \expandafter\def\csname LT5\endcsname{\color{black}}%
      \expandafter\def\csname LT6\endcsname{\color{black}}%
      \expandafter\def\csname LT7\endcsname{\color{black}}%
      \expandafter\def\csname LT8\endcsname{\color{black}}%
    \fi
  \fi
  \setlength{\unitlength}{0.0500bp}%
  \begin{picture}(7200.00,5040.00)%
    \gplgaddtomacro\gplbacktext{%
      \csname LTb\endcsname%
      \put(198,770){\makebox(0,0)[r]{\strut{} 1e-06}}%
      \put(198,1481){\makebox(0,0)[r]{\strut{} 1e-05}}%
      \put(198,2193){\makebox(0,0)[r]{\strut{} 0.0001}}%
      \put(198,2904){\makebox(0,0)[r]{\strut{} 0.001}}%
      \put(198,3616){\makebox(0,0)[r]{\strut{} 0.01}}%
      \put(198,4328){\makebox(0,0)[r]{\strut{} 0.1}}%
      \put(198,5039){\makebox(0,0)[r]{\strut{} 1}}%
      \put(330,550){\makebox(0,0){\strut{} 0}}%
      \put(1651,550){\makebox(0,0){\strut{} 1}}%
      \put(2972,550){\makebox(0,0){\strut{} 2}}%
      \put(4293,550){\makebox(0,0){\strut{} 3}}%
      \put(5614,550){\makebox(0,0){\strut{} 4}}%
      \put(6935,550){\makebox(0,0){\strut{} 5}}%
      \put(-968,2904){\rotatebox{-270}{\makebox(0,0){\strut{}\Large Mean $L_2$ Error}}}%
      \put(3632,220){\makebox(0,0){\strut{}\Large $\beta$}}%
    }%
    \gplgaddtomacro\gplfronttext{%
      \csname LTb\endcsname%
      \put(2778,2083){\makebox(0,0)[r]{\strut{}\large GCBP Err  (bias)}}%
      \csname LTb\endcsname%
      \put(2778,1863){\makebox(0,0)[r]{\strut{}\large BP Err  (bias)}}%
      \csname LTb\endcsname%
      \put(2778,1643){\makebox(0,0)[r]{\strut{}\large GCBP Err  (no bias)}}%
      \csname LTb\endcsname%
      \put(2778,1423){\makebox(0,0)[r]{\strut{}\large BP Err  (no bias)}}%
    }%
    \gplgaddtomacro\gplbacktext{%
      \csname LTb\endcsname%
      \put(330,550){\makebox(0,0){\strut{} 0}}%
      \put(1651,550){\makebox(0,0){\strut{} 1}}%
      \put(2972,550){\makebox(0,0){\strut{} 2}}%
      \put(4293,550){\makebox(0,0){\strut{} 3}}%
      \put(5614,550){\makebox(0,0){\strut{} 4}}%
      \put(6935,550){\makebox(0,0){\strut{} 5}}%
      \put(7067,770){\makebox(0,0)[l]{\strut{} 0}}%
      \put(7067,1624){\makebox(0,0)[l]{\strut{} 0.2}}%
      \put(7067,2478){\makebox(0,0)[l]{\strut{} 0.4}}%
      \put(7067,3331){\makebox(0,0)[l]{\strut{} 0.6}}%
      \put(7067,4185){\makebox(0,0)[l]{\strut{} 0.8}}%
      \put(7067,5039){\makebox(0,0)[l]{\strut{} 1}}%
      \put(3632,220){\makebox(0,0){\strut{}\Large $\beta$}}%
    }%
    \gplgaddtomacro\gplfronttext{%
      \csname LTb\endcsname%
      \put(6080,3221){\makebox(0,0)[r]{\strut{}\large Success rate (bias)}}%
      \csname LTb\endcsname%
      \put(6080,3001){\makebox(0,0)[r]{\strut{}\large Success rate (no bias)}}%
    }%
    \gplbacktext
    \put(0,0){\includegraphics{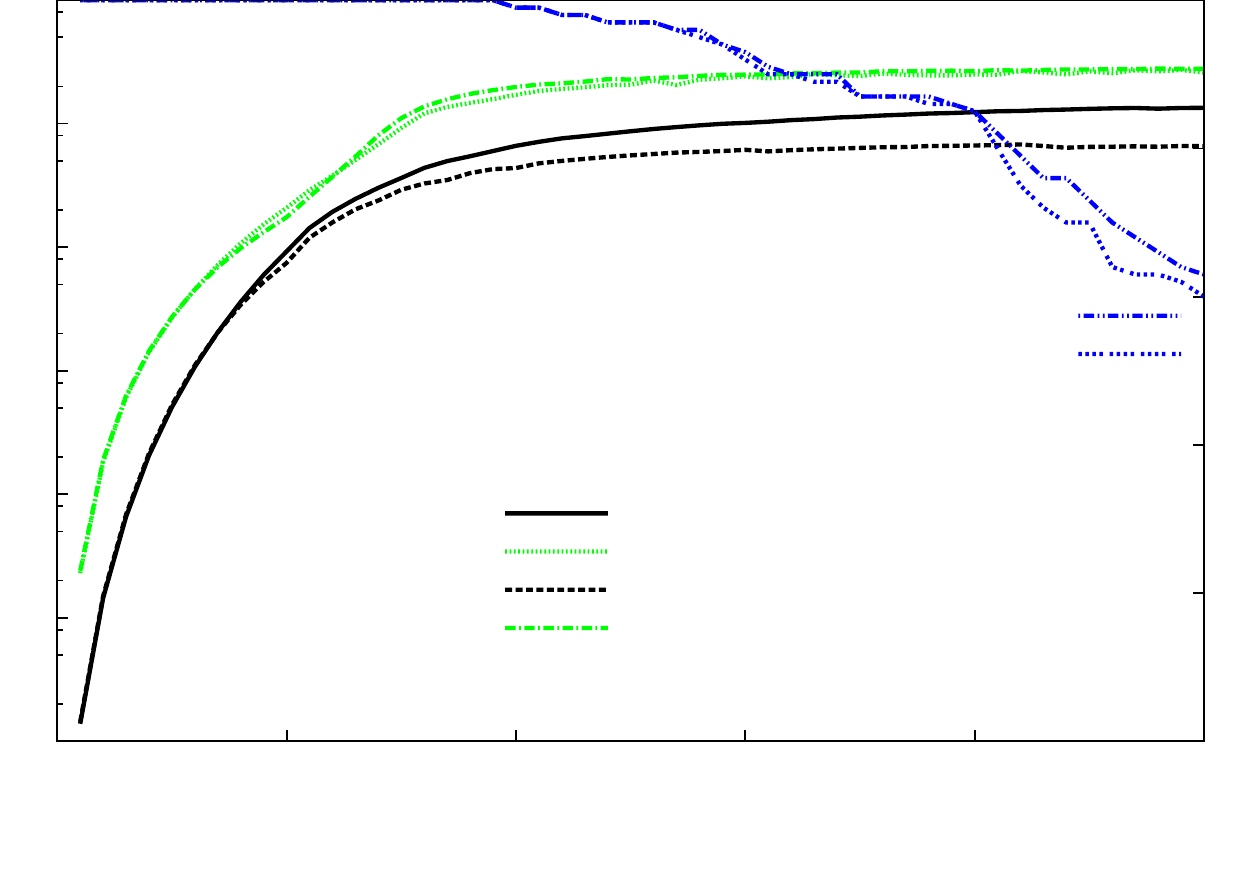}}%
    \gplfronttext
  \end{picture}%
\endgroup

%% file: bip_dirK5.tex
\begingroup
  \makeatletter
  \providecommand\color[2][]{%
    \GenericError{(gnuplot) \space\space\space\@spaces}{%
      Package color not loaded in conjunction with
      terminal option `colourtext'%
    }{See the gnuplot documentation for explanation.%
    }{Either use 'blacktext' in gnuplot or load the package
      color.sty in LaTeX.}%
    \renewcommand\color[2][]{}%
  }%
  \providecommand\includegraphics[2][]{%
    \GenericError{(gnuplot) \space\space\space\@spaces}{%
      Package graphicx or graphics not loaded%
    }{See the gnuplot documentation for explanation.%
    }{The gnuplot epslatex terminal needs graphicx.sty or graphics.sty.}%
    \renewcommand\includegraphics[2][]{}%
  }%
  \providecommand\rotatebox[2]{#2}%
  \@ifundefined{ifGPcolor}{%
    \newif\ifGPcolor
    \GPcolortrue
  }{}%
  \@ifundefined{ifGPblacktext}{%
    \newif\ifGPblacktext
    \GPblacktexttrue
  }{}%
  \let\gplgaddtomacro\g@addto@macro
  \gdef\gplbacktext{}%
  \gdef\gplfronttext{}%
  \makeatother
  \ifGPblacktext
    \def\colorrgb#1{}%
    \def\colorgray#1{}%
  \else
    \ifGPcolor
      \def\colorrgb#1{\color[rgb]{#1}}%
      \def\colorgray#1{\color[gray]{#1}}%
      \expandafter\def\csname LTw\endcsname{\color{white}}%
      \expandafter\def\csname LTb\endcsname{\color{black}}%
      \expandafter\def\csname LTa\endcsname{\color{black}}%
      \expandafter\def\csname LT0\endcsname{\color[rgb]{1,0,0}}%
      \expandafter\def\csname LT1\endcsname{\color[rgb]{0,1,0}}%
      \expandafter\def\csname LT2\endcsname{\color[rgb]{0,0,1}}%
      \expandafter\def\csname LT3\endcsname{\color[rgb]{1,0,1}}%
      \expandafter\def\csname LT4\endcsname{\color[rgb]{0,1,1}}%
      \expandafter\def\csname LT5\endcsname{\color[rgb]{1,1,0}}%
      \expandafter\def\csname LT6\endcsname{\color[rgb]{0,0,0}}%
      \expandafter\def\csname LT7\endcsname{\color[rgb]{1,0.3,0}}%
      \expandafter\def\csname LT8\endcsname{\color[rgb]{0.5,0.5,0.5}}%
    \else
      \def\colorrgb#1{\color{black}}%
      \def\colorgray#1{\color[gray]{#1}}%
      \expandafter\def\csname LTw\endcsname{\color{white}}%
      \expandafter\def\csname LTb\endcsname{\color{black}}%
      \expandafter\def\csname LTa\endcsname{\color{black}}%
      \expandafter\def\csname LT0\endcsname{\color{black}}%
      \expandafter\def\csname LT1\endcsname{\color{black}}%
      \expandafter\def\csname LT2\endcsname{\color{black}}%
      \expandafter\def\csname LT3\endcsname{\color{black}}%
      \expandafter\def\csname LT4\endcsname{\color{black}}%
      \expandafter\def\csname LT5\endcsname{\color{black}}%
      \expandafter\def\csname LT6\endcsname{\color{black}}%
      \expandafter\def\csname LT7\endcsname{\color{black}}%
      \expandafter\def\csname LT8\endcsname{\color{black}}%
    \fi
  \fi
  \setlength{\unitlength}{0.0500bp}%
  \begin{picture}(7200.00,5040.00)%
    \gplgaddtomacro\gplbacktext{%
      \csname LTb\endcsname%
      \put(198,770){\makebox(0,0)[r]{\strut{} 0}}%
      \put(198,1481){\makebox(0,0)[r]{\strut{} 0.05}}%
      \put(198,2193){\makebox(0,0)[r]{\strut{} 0.1}}%
      \put(198,2904){\makebox(0,0)[r]{\strut{} 0.15}}%
      \put(198,3616){\makebox(0,0)[r]{\strut{} 0.2}}%
      \put(198,4327){\makebox(0,0)[r]{\strut{} 0.25}}%
      \put(198,5039){\makebox(0,0)[r]{\strut{} 0.3}}%
      \put(330,550){\makebox(0,0){\strut{} 0}}%
      \put(1651,550){\makebox(0,0){\strut{} 1}}%
      \put(2972,550){\makebox(0,0){\strut{} 2}}%
      \put(4293,550){\makebox(0,0){\strut{} 3}}%
      \put(5614,550){\makebox(0,0){\strut{} 4}}%
      \put(6935,550){\makebox(0,0){\strut{} 5}}%
      \put(-704,2904){\rotatebox{-270}{\makebox(0,0){\strut{}\Large Mean $L_2$ Error}}}%
      \put(3632,220){\makebox(0,0){\strut{}\Large $\beta$}}%
    }%
    \gplgaddtomacro\gplfronttext{%
      \csname LTb\endcsname%
      \put(2381,4787){\makebox(0,0)[r]{\strut{}\large GCBP Err  (bias)}}%
      \csname LTb\endcsname%
      \put(2381,4567){\makebox(0,0)[r]{\strut{}\large BP Err  (bias)}}%
      \csname LTb\endcsname%
      \put(2381,4347){\makebox(0,0)[r]{\strut{}\large GCBP Err  (no bias)}}%
      \csname LTb\endcsname%
      \put(2381,4127){\makebox(0,0)[r]{\strut{}\large BP Err  (no bias)}}%
    }%
    \gplgaddtomacro\gplbacktext{%
      \csname LTb\endcsname%
      \put(330,550){\makebox(0,0){\strut{} 0}}%
      \put(1651,550){\makebox(0,0){\strut{} 1}}%
      \put(2972,550){\makebox(0,0){\strut{} 2}}%
      \put(4293,550){\makebox(0,0){\strut{} 3}}%
      \put(5614,550){\makebox(0,0){\strut{} 4}}%
      \put(6935,550){\makebox(0,0){\strut{} 5}}%
      \put(7067,770){\makebox(0,0)[l]{\strut{} 0}}%
      \put(7067,1624){\makebox(0,0)[l]{\strut{} 0.2}}%
      \put(7067,2478){\makebox(0,0)[l]{\strut{} 0.4}}%
      \put(7067,3331){\makebox(0,0)[l]{\strut{} 0.6}}%
      \put(7067,4185){\makebox(0,0)[l]{\strut{} 0.8}}%
      \put(7067,5039){\makebox(0,0)[l]{\strut{} 1}}%
      \put(3632,220){\makebox(0,0){\strut{}\Large $\beta$}}%
    }%
    \gplgaddtomacro\gplfronttext{%
      \csname LTb\endcsname%
      \put(5420,1514){\makebox(0,0)[r]{\strut{}\large Success rate (bias)}}%
      \csname LTb\endcsname%
      \put(5420,1294){\makebox(0,0)[r]{\strut{}\large Success rate (no bias)}}%
    }%
    \gplbacktext
    \put(0,0){\includegraphics{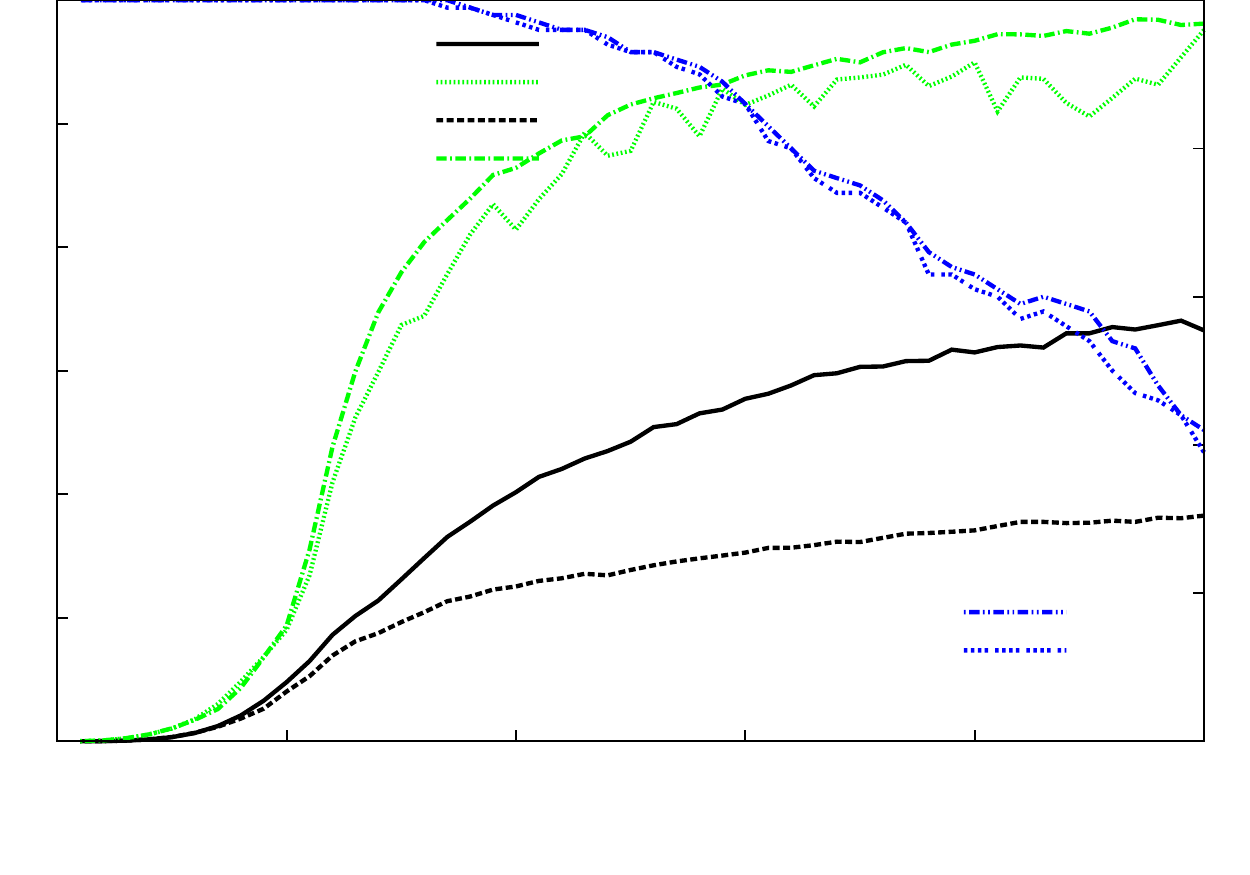}}%
    \gplfronttext
  \end{picture}%
\endgroup

%% file: bip_dirJ1.tex
\begingroup
  \makeatletter
  \providecommand\color[2][]{%
    \GenericError{(gnuplot) \space\space\space\@spaces}{%
      Package color not loaded in conjunction with
      terminal option `colourtext'%
    }{See the gnuplot documentation for explanation.%
    }{Either use 'blacktext' in gnuplot or load the package
      color.sty in LaTeX.}%
    \renewcommand\color[2][]{}%
  }%
  \providecommand\includegraphics[2][]{%
    \GenericError{(gnuplot) \space\space\space\@spaces}{%
      Package graphicx or graphics not loaded%
    }{See the gnuplot documentation for explanation.%
    }{The gnuplot epslatex terminal needs graphicx.sty or graphics.sty.}%
    \renewcommand\includegraphics[2][]{}%
  }%
  \providecommand\rotatebox[2]{#2}%
  \@ifundefined{ifGPcolor}{%
    \newif\ifGPcolor
    \GPcolortrue
  }{}%
  \@ifundefined{ifGPblacktext}{%
    \newif\ifGPblacktext
    \GPblacktexttrue
  }{}%
  \let\gplgaddtomacro\g@addto@macro
  \gdef\gplbacktext{}%
  \gdef\gplfronttext{}%
  \makeatother
  \ifGPblacktext
    \def\colorrgb#1{}%
    \def\colorgray#1{}%
  \else
    \ifGPcolor
      \def\colorrgb#1{\color[rgb]{#1}}%
      \def\colorgray#1{\color[gray]{#1}}%
      \expandafter\def\csname LTw\endcsname{\color{white}}%
      \expandafter\def\csname LTb\endcsname{\color{black}}%
      \expandafter\def\csname LTa\endcsname{\color{black}}%
      \expandafter\def\csname LT0\endcsname{\color[rgb]{1,0,0}}%
      \expandafter\def\csname LT1\endcsname{\color[rgb]{0,1,0}}%
      \expandafter\def\csname LT2\endcsname{\color[rgb]{0,0,1}}%
      \expandafter\def\csname LT3\endcsname{\color[rgb]{1,0,1}}%
      \expandafter\def\csname LT4\endcsname{\color[rgb]{0,1,1}}%
      \expandafter\def\csname LT5\endcsname{\color[rgb]{1,1,0}}%
      \expandafter\def\csname LT6\endcsname{\color[rgb]{0,0,0}}%
      \expandafter\def\csname LT7\endcsname{\color[rgb]{1,0.3,0}}%
      \expandafter\def\csname LT8\endcsname{\color[rgb]{0.5,0.5,0.5}}%
    \else
      \def\colorrgb#1{\color{black}}%
      \def\colorgray#1{\color[gray]{#1}}%
      \expandafter\def\csname LTw\endcsname{\color{white}}%
      \expandafter\def\csname LTb\endcsname{\color{black}}%
      \expandafter\def\csname LTa\endcsname{\color{black}}%
      \expandafter\def\csname LT0\endcsname{\color{black}}%
      \expandafter\def\csname LT1\endcsname{\color{black}}%
      \expandafter\def\csname LT2\endcsname{\color{black}}%
      \expandafter\def\csname LT3\endcsname{\color{black}}%
      \expandafter\def\csname LT4\endcsname{\color{black}}%
      \expandafter\def\csname LT5\endcsname{\color{black}}%
      \expandafter\def\csname LT6\endcsname{\color{black}}%
      \expandafter\def\csname LT7\endcsname{\color{black}}%
      \expandafter\def\csname LT8\endcsname{\color{black}}%
    \fi
  \fi
  \setlength{\unitlength}{0.0500bp}%
  \begin{picture}(7200.00,5040.00)%
    \gplgaddtomacro\gplbacktext{%
      \csname LTb\endcsname%
      \put(198,770){\makebox(0,0)[r]{\strut{} 0}}%
      \put(198,1624){\makebox(0,0)[r]{\strut{} 0.2}}%
      \put(198,2478){\makebox(0,0)[r]{\strut{} 0.4}}%
      \put(198,3331){\makebox(0,0)[r]{\strut{} 0.6}}%
      \put(198,4185){\makebox(0,0)[r]{\strut{} 0.8}}%
      \put(198,5039){\makebox(0,0)[r]{\strut{} 1}}%
      \put(330,550){\makebox(0,0){\strut{} 2}}%
      \put(1346,550){\makebox(0,0){\strut{} 3}}%
      \put(2362,550){\makebox(0,0){\strut{} 4}}%
      \put(3378,550){\makebox(0,0){\strut{} 5}}%
      \put(4395,550){\makebox(0,0){\strut{} 6}}%
      \put(5411,550){\makebox(0,0){\strut{} 7}}%
      \put(6427,550){\makebox(0,0){\strut{} 8}}%
      \put(-572,2904){\rotatebox{-270}{\makebox(0,0){\strut{}\Large Convergence Rate}}}%
      \put(3632,220){\makebox(0,0){\strut{}\Large Mean connectivity}}%
    }%
    \gplgaddtomacro\gplfronttext{%
      \csname LTb\endcsname%
      \put(4556,4716){\makebox(0,0)[r]{\strut{}\large Success rate (bias)}}%
      \csname LTb\endcsname%
      \put(4556,4496){\makebox(0,0)[r]{\strut{}\large Success rate (no bias)}}%
    }%
    \gplgaddtomacro\gplbacktext{%
      \csname LTb\endcsname%
      \put(330,550){\makebox(0,0){\strut{} 2}}%
      \put(1346,550){\makebox(0,0){\strut{} 3}}%
      \put(2362,550){\makebox(0,0){\strut{} 4}}%
      \put(3378,550){\makebox(0,0){\strut{} 5}}%
      \put(4395,550){\makebox(0,0){\strut{} 6}}%
      \put(5411,550){\makebox(0,0){\strut{} 7}}%
      \put(6427,550){\makebox(0,0){\strut{} 8}}%
      \put(7067,770){\makebox(0,0)[l]{\strut{} 0}}%
      \put(7067,1197){\makebox(0,0)[l]{\strut{} 0.02}}%
      \put(7067,1624){\makebox(0,0)[l]{\strut{} 0.04}}%
      \put(7067,2051){\makebox(0,0)[l]{\strut{} 0.06}}%
      \put(7067,2478){\makebox(0,0)[l]{\strut{} 0.08}}%
      \put(7067,2905){\makebox(0,0)[l]{\strut{} 0.1}}%
      \put(7067,3331){\makebox(0,0)[l]{\strut{} 0.12}}%
      \put(7067,3758){\makebox(0,0)[l]{\strut{} 0.14}}%
      \put(7067,4185){\makebox(0,0)[l]{\strut{} 0.16}}%
      \put(7067,4612){\makebox(0,0)[l]{\strut{} 0.18}}%
      \put(7067,5039){\makebox(0,0)[l]{\strut{} 0.2}}%
      \put(7968,2904){\rotatebox{-270}{\makebox(0,0){\strut{}\Large Mean $L_2$ Error}}}%
      \put(3632,220){\makebox(0,0){\strut{}\Large Mean connectivity}}%
    }%
    \gplgaddtomacro\gplfronttext{%
      \csname LTb\endcsname%
      \put(2523,2794){\makebox(0,0)[r]{\strut{}\large GCBP Err  (bias)}}%
      \csname LTb\endcsname%
      \put(2523,2574){\makebox(0,0)[r]{\strut{}\large BP Err  (bias)}}%
      \csname LTb\endcsname%
      \put(2523,2354){\makebox(0,0)[r]{\strut{}\large GCBP Err  (no bias)}}%
      \csname LTb\endcsname%
      \put(2523,2134){\makebox(0,0)[r]{\strut{}\large BP Err  (no bias)}}%
    }%
    \gplbacktext
    \put(0,0){\includegraphics{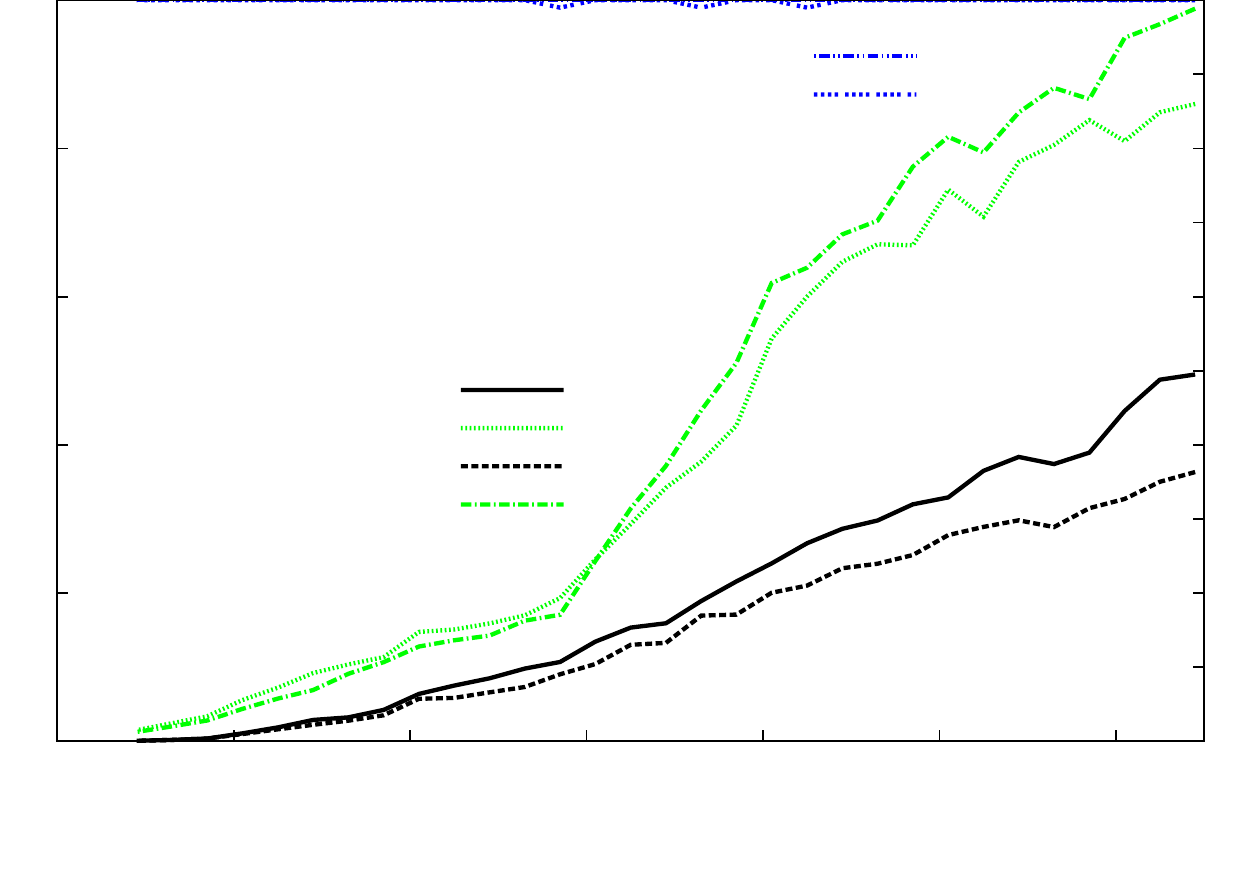}}%
    \gplfronttext
  \end{picture}%
\endgroup

%% file: ctime_grid_new.tex
\begingroup
  \makeatletter
  \providecommand\color[2][]{%
    \GenericError{(gnuplot) \space\space\space\@spaces}{%
      Package color not loaded in conjunction with
      terminal option `colourtext'%
    }{See the gnuplot documentation for explanation.%
    }{Either use 'blacktext' in gnuplot or load the package
      color.sty in LaTeX.}%
    \renewcommand\color[2][]{}%
  }%
  \providecommand\includegraphics[2][]{%
    \GenericError{(gnuplot) \space\space\space\@spaces}{%
      Package graphicx or graphics not loaded%
    }{See the gnuplot documentation for explanation.%
    }{The gnuplot epslatex terminal needs graphicx.sty or graphics.sty.}%
    \renewcommand\includegraphics[2][]{}%
  }%
  \providecommand\rotatebox[2]{#2}%
  \@ifundefined{ifGPcolor}{%
    \newif\ifGPcolor
    \GPcolortrue
  }{}%
  \@ifundefined{ifGPblacktext}{%
    \newif\ifGPblacktext
    \GPblacktexttrue
  }{}%
  \let\gplgaddtomacro\g@addto@macro
  \gdef\gplbacktext{}%
  \gdef\gplfronttext{}%
  \makeatother
  \ifGPblacktext
    \def\colorrgb#1{}%
    \def\colorgray#1{}%
  \else
    \ifGPcolor
      \def\colorrgb#1{\color[rgb]{#1}}%
      \def\colorgray#1{\color[gray]{#1}}%
      \expandafter\def\csname LTw\endcsname{\color{white}}%
      \expandafter\def\csname LTb\endcsname{\color{black}}%
      \expandafter\def\csname LTa\endcsname{\color{black}}%
      \expandafter\def\csname LT0\endcsname{\color[rgb]{1,0,0}}%
      \expandafter\def\csname LT1\endcsname{\color[rgb]{0,1,0}}%
      \expandafter\def\csname LT2\endcsname{\color[rgb]{0,0,1}}%
      \expandafter\def\csname LT3\endcsname{\color[rgb]{1,0,1}}%
      \expandafter\def\csname LT4\endcsname{\color[rgb]{0,1,1}}%
      \expandafter\def\csname LT5\endcsname{\color[rgb]{1,1,0}}%
      \expandafter\def\csname LT6\endcsname{\color[rgb]{0,0,0}}%
      \expandafter\def\csname LT7\endcsname{\color[rgb]{1,0.3,0}}%
      \expandafter\def\csname LT8\endcsname{\color[rgb]{0.5,0.5,0.5}}%
    \else
      \def\colorrgb#1{\color{black}}%
      \def\colorgray#1{\color[gray]{#1}}%
      \expandafter\def\csname LTw\endcsname{\color{white}}%
      \expandafter\def\csname LTb\endcsname{\color{black}}%
      \expandafter\def\csname LTa\endcsname{\color{black}}%
      \expandafter\def\csname LT0\endcsname{\color{black}}%
      \expandafter\def\csname LT1\endcsname{\color{black}}%
      \expandafter\def\csname LT2\endcsname{\color{black}}%
      \expandafter\def\csname LT3\endcsname{\color{black}}%
      \expandafter\def\csname LT4\endcsname{\color{black}}%
      \expandafter\def\csname LT5\endcsname{\color{black}}%
      \expandafter\def\csname LT6\endcsname{\color{black}}%
      \expandafter\def\csname LT7\endcsname{\color{black}}%
      \expandafter\def\csname LT8\endcsname{\color{black}}%
    \fi
  \fi
  \setlength{\unitlength}{0.0500bp}%
  \begin{picture}(7200.00,5040.00)%
    \gplgaddtomacro\gplbacktext{%
      \csname LTb\endcsname%
      \put(198,1264){\makebox(0,0)[r]{\strut{} 1}}%
      \put(198,2207){\makebox(0,0)[r]{\strut{} 10}}%
      \put(198,3151){\makebox(0,0)[r]{\strut{} 100}}%
      \put(198,4095){\makebox(0,0)[r]{\strut{} 1000}}%
      \put(198,5039){\makebox(0,0)[r]{\strut{} 10000}}%
      \put(429,550){\makebox(0,0){\strut{} 1000}}%
      \put(2598,550){\makebox(0,0){\strut{} 10000}}%
      \put(4766,550){\makebox(0,0){\strut{} 100000}}%
      \put(6935,550){\makebox(0,0){\strut{} 1e+06}}%
      \put(-836,2904){\rotatebox{-270}{\makebox(0,0){\strut{}\Large CPU Time (s)}}}%
      \put(3632,220){\makebox(0,0){\strut{}\Large Size}}%
    }%
    \gplgaddtomacro\gplfronttext{%
      \csname LTb\endcsname%
      \put(1833,4886){\makebox(0,0)[r]{\strut{}\large GCBP $\beta=2$}}%
      \csname LTb\endcsname%
      \put(1833,4666){\makebox(0,0)[r]{\strut{}\large GCBP $\beta=1.5$}}%
      \csname LTb\endcsname%
      \put(1833,4446){\makebox(0,0)[r]{\strut{}\large GCBP $\beta=1$}}%
      \csname LTb\endcsname%
      \put(1833,4226){\makebox(0,0)[r]{\strut{}\large GCBP $\beta=0.5$}}%
      \csname LTb\endcsname%
      \put(1833,4006){\makebox(0,0)[r]{\strut{}\large BP $\beta=1$}}%
      \csname LTb\endcsname%
      \put(1833,3786){\makebox(0,0)[r]{\strut{}\large BP $\beta=0.5$}}%
      \csname LTb\endcsname%
      \put(1833,3566){\makebox(0,0)[r]{\strut{}\large $6.10^{-4}\times N^{1.23}$}}%
      \csname LTb\endcsname%
      \put(1833,3346){\makebox(0,0)[r]{\strut{}\large $6.10^{-4}\times N^{1.15}$}}%
      \csname LTb\endcsname%
      \put(1833,3126){\makebox(0,0)[r]{\strut{}\large $5.10^{-4}\times N^{1.07}$}}%
      \csname LTb\endcsname%
      \put(1833,2906){\makebox(0,0)[r]{\strut{}\large $2.5 10^{-4}\times N^{1.05}$}}%
    }%
    \gplbacktext
    \put(0,0){\includegraphics{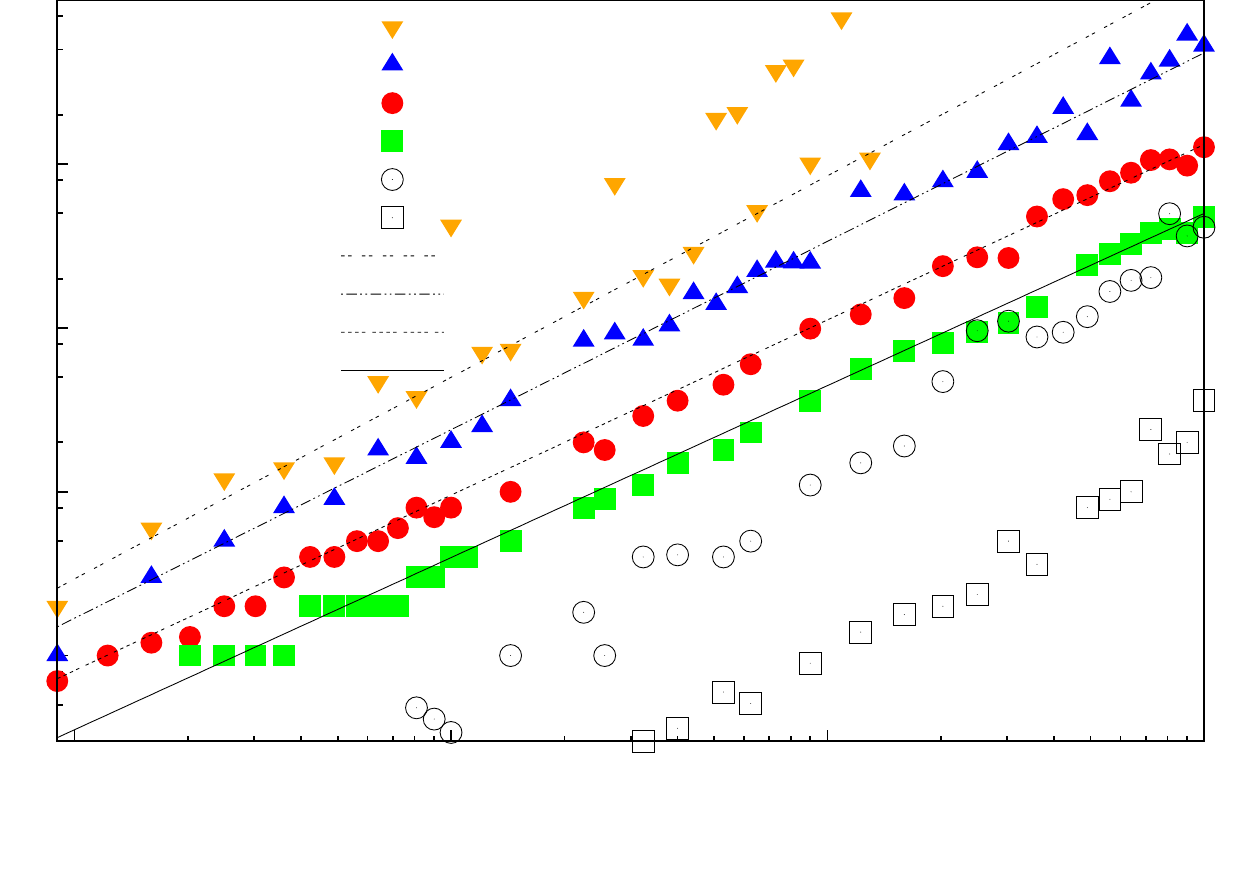}}%
    \gplfronttext
  \end{picture}%
\endgroup

%% file: ctime_bip_new.tex
\begingroup
  \makeatletter
  \providecommand\color[2][]{%
    \GenericError{(gnuplot) \space\space\space\@spaces}{%
      Package color not loaded in conjunction with
      terminal option `colourtext'%
    }{See the gnuplot documentation for explanation.%
    }{Either use 'blacktext' in gnuplot or load the package
      color.sty in LaTeX.}%
    \renewcommand\color[2][]{}%
  }%
  \providecommand\includegraphics[2][]{%
    \GenericError{(gnuplot) \space\space\space\@spaces}{%
      Package graphicx or graphics not loaded%
    }{See the gnuplot documentation for explanation.%
    }{The gnuplot epslatex terminal needs graphicx.sty or graphics.sty.}%
    \renewcommand\includegraphics[2][]{}%
  }%
  \providecommand\rotatebox[2]{#2}%
  \@ifundefined{ifGPcolor}{%
    \newif\ifGPcolor
    \GPcolortrue
  }{}%
  \@ifundefined{ifGPblacktext}{%
    \newif\ifGPblacktext
    \GPblacktexttrue
  }{}%
  \let\gplgaddtomacro\g@addto@macro
  \gdef\gplbacktext{}%
  \gdef\gplfronttext{}%
  \makeatother
  \ifGPblacktext
    \def\colorrgb#1{}%
    \def\colorgray#1{}%
  \else
    \ifGPcolor
      \def\colorrgb#1{\color[rgb]{#1}}%
      \def\colorgray#1{\color[gray]{#1}}%
      \expandafter\def\csname LTw\endcsname{\color{white}}%
      \expandafter\def\csname LTb\endcsname{\color{black}}%
      \expandafter\def\csname LTa\endcsname{\color{black}}%
      \expandafter\def\csname LT0\endcsname{\color[rgb]{1,0,0}}%
      \expandafter\def\csname LT1\endcsname{\color[rgb]{0,1,0}}%
      \expandafter\def\csname LT2\endcsname{\color[rgb]{0,0,1}}%
      \expandafter\def\csname LT3\endcsname{\color[rgb]{1,0,1}}%
      \expandafter\def\csname LT4\endcsname{\color[rgb]{0,1,1}}%
      \expandafter\def\csname LT5\endcsname{\color[rgb]{1,1,0}}%
      \expandafter\def\csname LT6\endcsname{\color[rgb]{0,0,0}}%
      \expandafter\def\csname LT7\endcsname{\color[rgb]{1,0.3,0}}%
      \expandafter\def\csname LT8\endcsname{\color[rgb]{0.5,0.5,0.5}}%
    \else
      \def\colorrgb#1{\color{black}}%
      \def\colorgray#1{\color[gray]{#1}}%
      \expandafter\def\csname LTw\endcsname{\color{white}}%
      \expandafter\def\csname LTb\endcsname{\color{black}}%
      \expandafter\def\csname LTa\endcsname{\color{black}}%
      \expandafter\def\csname LT0\endcsname{\color{black}}%
      \expandafter\def\csname LT1\endcsname{\color{black}}%
      \expandafter\def\csname LT2\endcsname{\color{black}}%
      \expandafter\def\csname LT3\endcsname{\color{black}}%
      \expandafter\def\csname LT4\endcsname{\color{black}}%
      \expandafter\def\csname LT5\endcsname{\color{black}}%
      \expandafter\def\csname LT6\endcsname{\color{black}}%
      \expandafter\def\csname LT7\endcsname{\color{black}}%
      \expandafter\def\csname LT8\endcsname{\color{black}}%
    \fi
  \fi
  \setlength{\unitlength}{0.0500bp}%
  \begin{picture}(7200.00,5040.00)%
    \gplgaddtomacro\gplbacktext{%
      \csname LTb\endcsname%
      \put(198,770){\makebox(0,0)[r]{\strut{} 1}}%
      \put(198,1864){\makebox(0,0)[r]{\strut{} 10}}%
      \put(198,2957){\makebox(0,0)[r]{\strut{} 100}}%
      \put(198,4051){\makebox(0,0)[r]{\strut{} 1000}}%
      \put(932,550){\makebox(0,0){\strut{} 100}}%
      \put(2933,550){\makebox(0,0){\strut{} 1000}}%
      \put(4934,550){\makebox(0,0){\strut{} 10000}}%
      \put(6935,550){\makebox(0,0){\strut{} 100000}}%
      \put(-704,2904){\rotatebox{-270}{\makebox(0,0){\strut{}\Large CPU Time (s)}}}%
      \put(3632,220){\makebox(0,0){\strut{}\Large Size}}%
    }%
    \gplgaddtomacro\gplfronttext{%
      \csname LTb\endcsname%
      \put(1884,4600){\makebox(0,0)[r]{\strut{}\large GCBP $d=3$}}%
      \csname LTb\endcsname%
      \put(1884,4380){\makebox(0,0)[r]{\strut{}\large GCBP $d=4$}}%
      \csname LTb\endcsname%
      \put(1884,4160){\makebox(0,0)[r]{\strut{}\large GCBP $d=6$}}%
      \csname LTb\endcsname%
      \put(1884,3940){\makebox(0,0)[r]{\strut{}\large MCB $d=3$}}%
      \csname LTb\endcsname%
      \put(1884,3720){\makebox(0,0)[r]{\strut{}\large MCB $d=4$}}%
      \csname LTb\endcsname%
      \put(1884,3500){\makebox(0,0)[r]{\strut{}\large MCB $d=6$}}%
      \csname LTb\endcsname%
      \put(1884,3280){\makebox(0,0)[r]{\strut{}\large $5.5 10^{-4}\times x^{1.24}$}}%
      \csname LTb\endcsname%
      \put(1884,3060){\makebox(0,0)[r]{\strut{}\large $0.002 \times x^{1.25}$}}%
      \csname LTb\endcsname%
      \put(1884,2840){\makebox(0,0)[r]{\strut{}\large $0.0015 \times x^{1.55}$}}%
      \csname LTb\endcsname%
      \put(1884,2620){\makebox(0,0)[r]{\strut{}\large $3.10^{-7}\times x^{2}$}}%
    }%
    \gplbacktext
    \put(0,0){\includegraphics{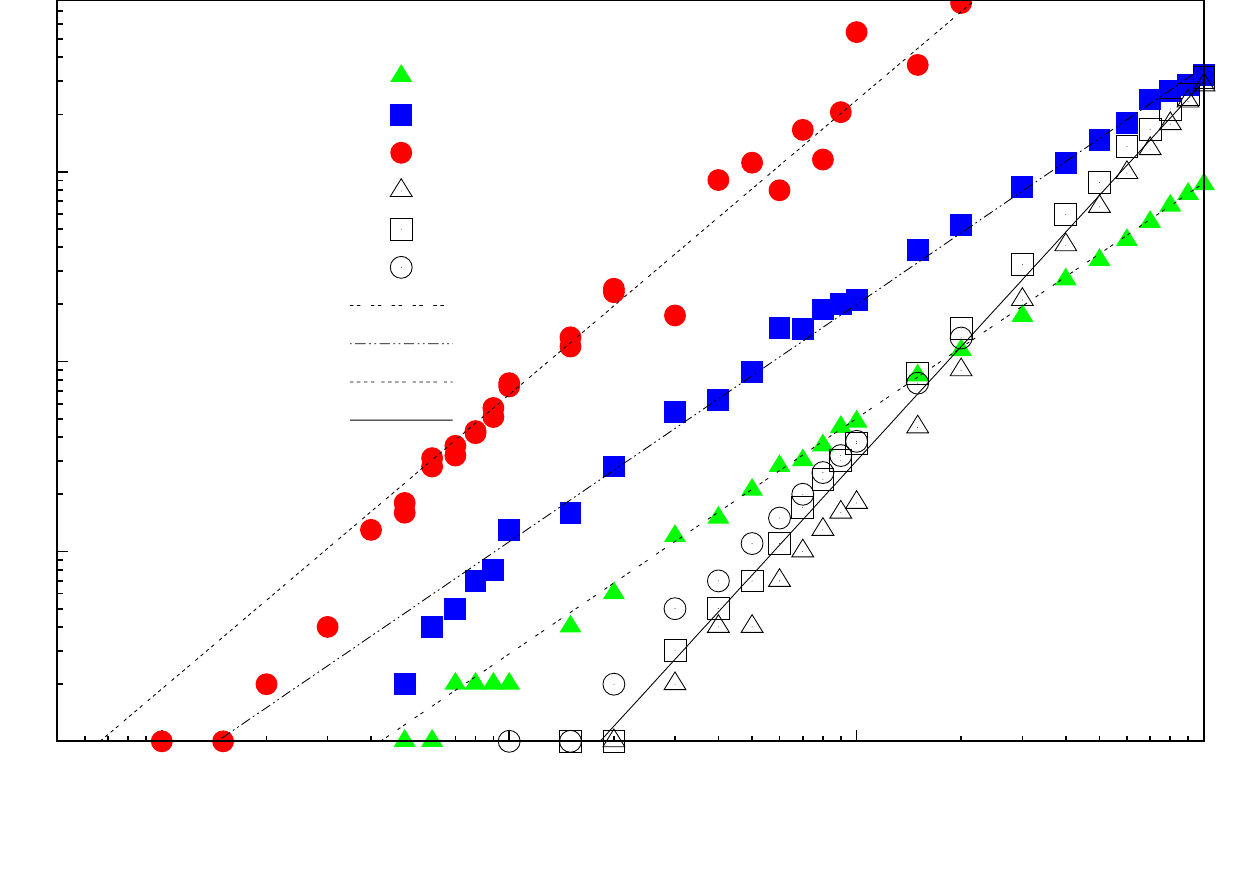}}%
    \gplfronttext
  \end{picture}%
\endgroup

%% file: grid_IIa.tex
\begingroup
  \makeatletter
  \providecommand\color[2][]{%
    \GenericError{(gnuplot) \space\space\space\@spaces}{%
      Package color not loaded in conjunction with
      terminal option `colourtext'%
    }{See the gnuplot documentation for explanation.%
    }{Either use 'blacktext' in gnuplot or load the package
      color.sty in LaTeX.}%
    \renewcommand\color[2][]{}%
  }%
  \providecommand\includegraphics[2][]{%
    \GenericError{(gnuplot) \space\space\space\@spaces}{%
      Package graphicx or graphics not loaded%
    }{See the gnuplot documentation for explanation.%
    }{The gnuplot epslatex terminal needs graphicx.sty or graphics.sty.}%
    \renewcommand\includegraphics[2][]{}%
  }%
  \providecommand\rotatebox[2]{#2}%
  \@ifundefined{ifGPcolor}{%
    \newif\ifGPcolor
    \GPcolortrue
  }{}%
  \@ifundefined{ifGPblacktext}{%
    \newif\ifGPblacktext
    \GPblacktexttrue
  }{}%
  \let\gplgaddtomacro\g@addto@macro
  \gdef\gplbacktext{}%
  \gdef\gplfronttext{}%
  \makeatother
  \ifGPblacktext
    \def\colorrgb#1{}%
    \def\colorgray#1{}%
  \else
    \ifGPcolor
      \def\colorrgb#1{\color[rgb]{#1}}%
      \def\colorgray#1{\color[gray]{#1}}%
      \expandafter\def\csname LTw\endcsname{\color{white}}%
      \expandafter\def\csname LTb\endcsname{\color{black}}%
      \expandafter\def\csname LTa\endcsname{\color{black}}%
      \expandafter\def\csname LT0\endcsname{\color[rgb]{1,0,0}}%
      \expandafter\def\csname LT1\endcsname{\color[rgb]{0,1,0}}%
      \expandafter\def\csname LT2\endcsname{\color[rgb]{0,0,1}}%
      \expandafter\def\csname LT3\endcsname{\color[rgb]{1,0,1}}%
      \expandafter\def\csname LT4\endcsname{\color[rgb]{0,1,1}}%
      \expandafter\def\csname LT5\endcsname{\color[rgb]{1,1,0}}%
      \expandafter\def\csname LT6\endcsname{\color[rgb]{0,0,0}}%
      \expandafter\def\csname LT7\endcsname{\color[rgb]{1,0.3,0}}%
      \expandafter\def\csname LT8\endcsname{\color[rgb]{0.5,0.5,0.5}}%
    \else
      \def\colorrgb#1{\color{black}}%
      \def\colorgray#1{\color[gray]{#1}}%
      \expandafter\def\csname LTw\endcsname{\color{white}}%
      \expandafter\def\csname LTb\endcsname{\color{black}}%
      \expandafter\def\csname LTa\endcsname{\color{black}}%
      \expandafter\def\csname LT0\endcsname{\color{black}}%
      \expandafter\def\csname LT1\endcsname{\color{black}}%
      \expandafter\def\csname LT2\endcsname{\color{black}}%
      \expandafter\def\csname LT3\endcsname{\color{black}}%
      \expandafter\def\csname LT4\endcsname{\color{black}}%
      \expandafter\def\csname LT5\endcsname{\color{black}}%
      \expandafter\def\csname LT6\endcsname{\color{black}}%
      \expandafter\def\csname LT7\endcsname{\color{black}}%
      \expandafter\def\csname LT8\endcsname{\color{black}}%
    \fi
  \fi
  \setlength{\unitlength}{0.0500bp}%
  \begin{picture}(7200.00,5040.00)%
    \gplgaddtomacro\gplbacktext{%
      \csname LTb\endcsname%
      \put(198,770){\makebox(0,0)[r]{\strut{} 1e-05}}%
      \put(198,1624){\makebox(0,0)[r]{\strut{} 0.0001}}%
      \put(198,2478){\makebox(0,0)[r]{\strut{} 0.001}}%
      \put(198,3331){\makebox(0,0)[r]{\strut{} 0.01}}%
      \put(198,4185){\makebox(0,0)[r]{\strut{} 0.1}}%
      \put(198,5039){\makebox(0,0)[r]{\strut{} 1}}%
      \put(330,550){\makebox(0,0){\strut{} 0}}%
      \put(1431,550){\makebox(0,0){\strut{} 0.5}}%
      \put(2532,550){\makebox(0,0){\strut{} 1}}%
      \put(3633,550){\makebox(0,0){\strut{} 1.5}}%
      \put(4733,550){\makebox(0,0){\strut{} 2}}%
      \put(5834,550){\makebox(0,0){\strut{} 2.5}}%
      \put(6935,550){\makebox(0,0){\strut{} 3}}%
      \put(-968,2904){\rotatebox{-270}{\makebox(0,0){\strut{}\Large Mean $L_2$ Error}}}%
      \put(3632,220){\makebox(0,0){\strut{}\Large $\beta$}}%
    }%
    \gplgaddtomacro\gplfronttext{%
      \csname LTb\endcsname%
      \put(4979,2368){\makebox(0,0)[r]{\strut{}\Large Kikuchi}}%
      \csname LTb\endcsname%
      \put(4979,2148){\makebox(0,0)[r]{\strut{}\Large BA+LR}}%
    }%
    \gplbacktext
    \put(0,0){\includegraphics{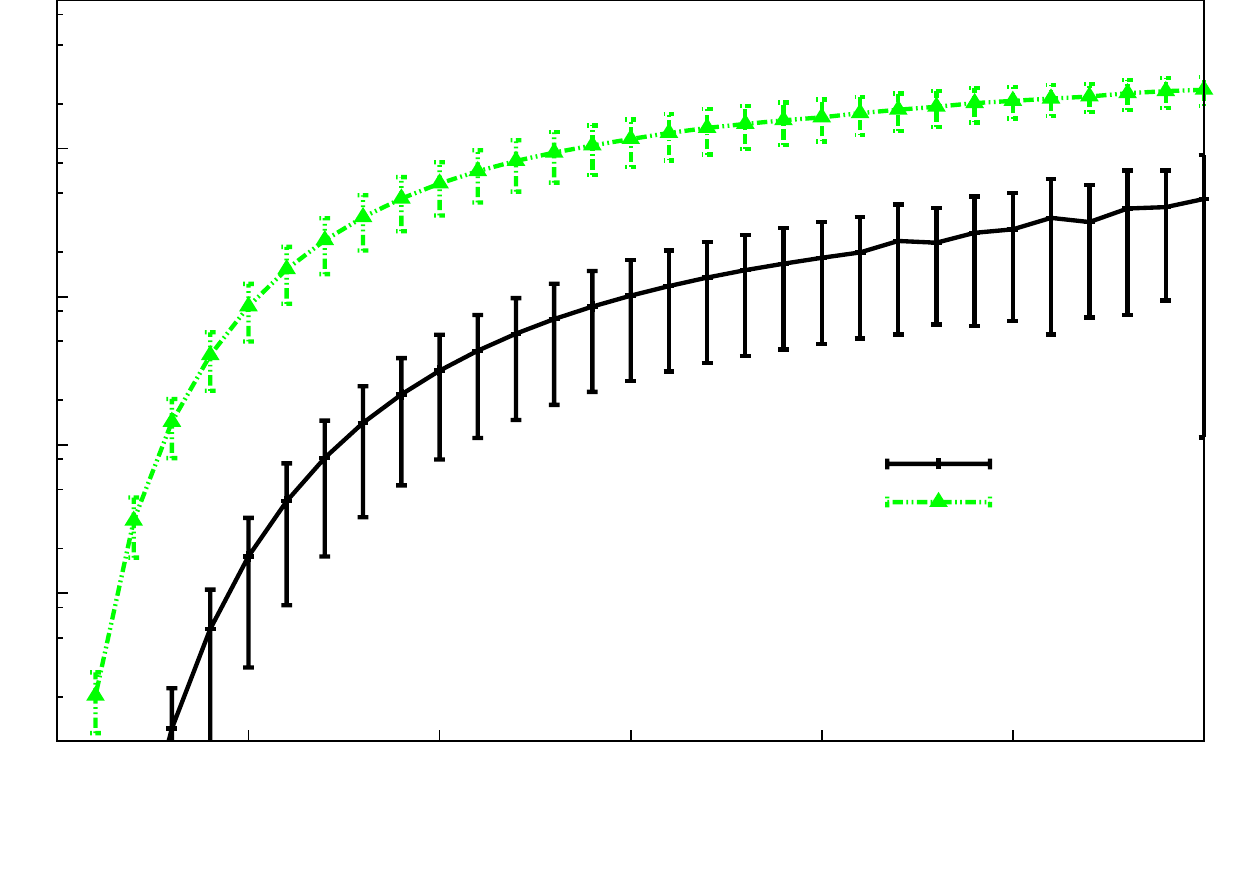}}%
    \gplfronttext
  \end{picture}%
\endgroup

%% file: bip_II.tex
\begingroup
  \makeatletter
  \providecommand\color[2][]{%
    \GenericError{(gnuplot) \space\space\space\@spaces}{%
      Package color not loaded in conjunction with
      terminal option `colourtext'%
    }{See the gnuplot documentation for explanation.%
    }{Either use 'blacktext' in gnuplot or load the package
      color.sty in LaTeX.}%
    \renewcommand\color[2][]{}%
  }%
  \providecommand\includegraphics[2][]{%
    \GenericError{(gnuplot) \space\space\space\@spaces}{%
      Package graphicx or graphics not loaded%
    }{See the gnuplot documentation for explanation.%
    }{The gnuplot epslatex terminal needs graphicx.sty or graphics.sty.}%
    \renewcommand\includegraphics[2][]{}%
  }%
  \providecommand\rotatebox[2]{#2}%
  \@ifundefined{ifGPcolor}{%
    \newif\ifGPcolor
    \GPcolortrue
  }{}%
  \@ifundefined{ifGPblacktext}{%
    \newif\ifGPblacktext
    \GPblacktexttrue
  }{}%
  \let\gplgaddtomacro\g@addto@macro
  \gdef\gplbacktext{}%
  \gdef\gplfronttext{}%
  \makeatother
  \ifGPblacktext
    \def\colorrgb#1{}%
    \def\colorgray#1{}%
  \else
    \ifGPcolor
      \def\colorrgb#1{\color[rgb]{#1}}%
      \def\colorgray#1{\color[gray]{#1}}%
      \expandafter\def\csname LTw\endcsname{\color{white}}%
      \expandafter\def\csname LTb\endcsname{\color{black}}%
      \expandafter\def\csname LTa\endcsname{\color{black}}%
      \expandafter\def\csname LT0\endcsname{\color[rgb]{1,0,0}}%
      \expandafter\def\csname LT1\endcsname{\color[rgb]{0,1,0}}%
      \expandafter\def\csname LT2\endcsname{\color[rgb]{0,0,1}}%
      \expandafter\def\csname LT3\endcsname{\color[rgb]{1,0,1}}%
      \expandafter\def\csname LT4\endcsname{\color[rgb]{0,1,1}}%
      \expandafter\def\csname LT5\endcsname{\color[rgb]{1,1,0}}%
      \expandafter\def\csname LT6\endcsname{\color[rgb]{0,0,0}}%
      \expandafter\def\csname LT7\endcsname{\color[rgb]{1,0.3,0}}%
      \expandafter\def\csname LT8\endcsname{\color[rgb]{0.5,0.5,0.5}}%
    \else
      \def\colorrgb#1{\color{black}}%
      \def\colorgray#1{\color[gray]{#1}}%
      \expandafter\def\csname LTw\endcsname{\color{white}}%
      \expandafter\def\csname LTb\endcsname{\color{black}}%
      \expandafter\def\csname LTa\endcsname{\color{black}}%
      \expandafter\def\csname LT0\endcsname{\color{black}}%
      \expandafter\def\csname LT1\endcsname{\color{black}}%
      \expandafter\def\csname LT2\endcsname{\color{black}}%
      \expandafter\def\csname LT3\endcsname{\color{black}}%
      \expandafter\def\csname LT4\endcsname{\color{black}}%
      \expandafter\def\csname LT5\endcsname{\color{black}}%
      \expandafter\def\csname LT6\endcsname{\color{black}}%
      \expandafter\def\csname LT7\endcsname{\color{black}}%
      \expandafter\def\csname LT8\endcsname{\color{black}}%
    \fi
  \fi
  \setlength{\unitlength}{0.0500bp}%
  \begin{picture}(7200.00,5040.00)%
    \gplgaddtomacro\gplbacktext{%
      \csname LTb\endcsname%
      \put(330,550){\makebox(0,0){\strut{} 2}}%
      \put(1981,550){\makebox(0,0){\strut{} 2.5}}%
      \put(3633,550){\makebox(0,0){\strut{} 3}}%
      \put(5284,550){\makebox(0,0){\strut{} 3.5}}%
      \put(6935,550){\makebox(0,0){\strut{} 4}}%
      \put(7067,770){\makebox(0,0)[l]{\strut{} 1e-08}}%
      \put(7067,1304){\makebox(0,0)[l]{\strut{} 1e-07}}%
      \put(7067,1837){\makebox(0,0)[l]{\strut{} 1e-06}}%
      \put(7067,2371){\makebox(0,0)[l]{\strut{} 1e-05}}%
      \put(7067,2905){\makebox(0,0)[l]{\strut{} 0.0001}}%
      \put(7067,3438){\makebox(0,0)[l]{\strut{} 0.001}}%
      \put(7067,3972){\makebox(0,0)[l]{\strut{} 0.01}}%
      \put(7067,4505){\makebox(0,0)[l]{\strut{} 0.1}}%
      \put(7067,5039){\makebox(0,0)[l]{\strut{} 1}}%
      \put(8232,2904){\rotatebox{-270}{\makebox(0,0){\strut{}\Large Mean $L_2$ Error}}}%
      \put(3632,220){\makebox(0,0){\strut{}\Large Connectivity}}%
    }%
    \gplgaddtomacro\gplfronttext{%
      \csname LTb\endcsname%
      \put(5948,3014){\makebox(0,0)[r]{\strut{}\Large Kikuchi}}%
      \csname LTb\endcsname%
      \put(5948,2794){\makebox(0,0)[r]{\strut{}\Large BA+LR}}%
    }%
    \gplbacktext
    \put(0,0){\includegraphics{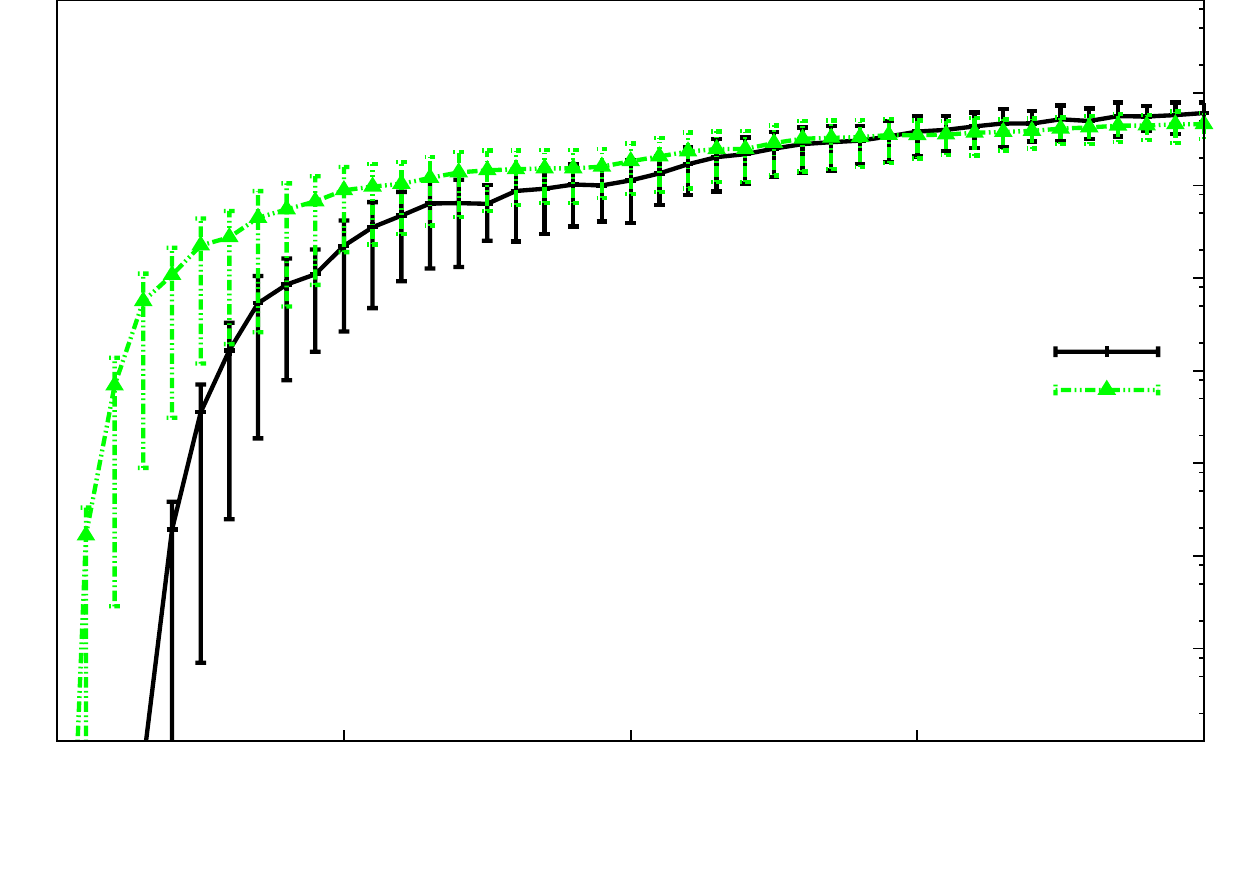}}%
    \gplfronttext
  \end{picture}%
\endgroup

%% file: grid_113.tex
\begingroup
  \makeatletter
  \providecommand\color[2][]{%
    \GenericError{(gnuplot) \space\space\space\@spaces}{%
      Package color not loaded in conjunction with
      terminal option `colourtext'%
    }{See the gnuplot documentation for explanation.%
    }{Either use 'blacktext' in gnuplot or load the package
      color.sty in LaTeX.}%
    \renewcommand\color[2][]{}%
  }%
  \providecommand\includegraphics[2][]{%
    \GenericError{(gnuplot) \space\space\space\@spaces}{%
      Package graphicx or graphics not loaded%
    }{See the gnuplot documentation for explanation.%
    }{The gnuplot epslatex terminal needs graphicx.sty or graphics.sty.}%
    \renewcommand\includegraphics[2][]{}%
  }%
  \providecommand\rotatebox[2]{#2}%
  \@ifundefined{ifGPcolor}{%
    \newif\ifGPcolor
    \GPcolortrue
  }{}%
  \@ifundefined{ifGPblacktext}{%
    \newif\ifGPblacktext
    \GPblacktexttrue
  }{}%
  \let\gplgaddtomacro\g@addto@macro
  \gdef\gplbacktext{}%
  \gdef\gplfronttext{}%
  \makeatother
  \ifGPblacktext
    \def\colorrgb#1{}%
    \def\colorgray#1{}%
  \else
    \ifGPcolor
      \def\colorrgb#1{\color[rgb]{#1}}%
      \def\colorgray#1{\color[gray]{#1}}%
      \expandafter\def\csname LTw\endcsname{\color{white}}%
      \expandafter\def\csname LTb\endcsname{\color{black}}%
      \expandafter\def\csname LTa\endcsname{\color{black}}%
      \expandafter\def\csname LT0\endcsname{\color[rgb]{1,0,0}}%
      \expandafter\def\csname LT1\endcsname{\color[rgb]{0,1,0}}%
      \expandafter\def\csname LT2\endcsname{\color[rgb]{0,0,1}}%
      \expandafter\def\csname LT3\endcsname{\color[rgb]{1,0,1}}%
      \expandafter\def\csname LT4\endcsname{\color[rgb]{0,1,1}}%
      \expandafter\def\csname LT5\endcsname{\color[rgb]{1,1,0}}%
      \expandafter\def\csname LT6\endcsname{\color[rgb]{0,0,0}}%
      \expandafter\def\csname LT7\endcsname{\color[rgb]{1,0.3,0}}%
      \expandafter\def\csname LT8\endcsname{\color[rgb]{0.5,0.5,0.5}}%
    \else
      \def\colorrgb#1{\color{black}}%
      \def\colorgray#1{\color[gray]{#1}}%
      \expandafter\def\csname LTw\endcsname{\color{white}}%
      \expandafter\def\csname LTb\endcsname{\color{black}}%
      \expandafter\def\csname LTa\endcsname{\color{black}}%
      \expandafter\def\csname LT0\endcsname{\color{black}}%
      \expandafter\def\csname LT1\endcsname{\color{black}}%
      \expandafter\def\csname LT2\endcsname{\color{black}}%
      \expandafter\def\csname LT3\endcsname{\color{black}}%
      \expandafter\def\csname LT4\endcsname{\color{black}}%
      \expandafter\def\csname LT5\endcsname{\color{black}}%
      \expandafter\def\csname LT6\endcsname{\color{black}}%
      \expandafter\def\csname LT7\endcsname{\color{black}}%
      \expandafter\def\csname LT8\endcsname{\color{black}}%
    \fi
  \fi
  \setlength{\unitlength}{0.0500bp}%
  \begin{picture}(7200.00,5040.00)%
    \gplgaddtomacro\gplbacktext{%
      \csname LTb\endcsname%
      \put(198,770){\makebox(0,0)[r]{\strut{} 1e-05}}%
      \put(198,1624){\makebox(0,0)[r]{\strut{} 0.0001}}%
      \put(198,2478){\makebox(0,0)[r]{\strut{} 0.001}}%
      \put(198,3331){\makebox(0,0)[r]{\strut{} 0.01}}%
      \put(198,4185){\makebox(0,0)[r]{\strut{} 0.1}}%
      \put(198,5039){\makebox(0,0)[r]{\strut{} 1}}%
      \put(330,550){\makebox(0,0){\strut{} 0}}%
      \put(1156,550){\makebox(0,0){\strut{} 0.5}}%
      \put(1981,550){\makebox(0,0){\strut{} 1}}%
      \put(2807,550){\makebox(0,0){\strut{} 1.5}}%
      \put(3633,550){\makebox(0,0){\strut{} 2}}%
      \put(4458,550){\makebox(0,0){\strut{} 2.5}}%
      \put(5284,550){\makebox(0,0){\strut{} 3}}%
      \put(6109,550){\makebox(0,0){\strut{} 3.5}}%
      \put(6935,550){\makebox(0,0){\strut{} 4}}%
      \put(-968,2904){\rotatebox{-270}{\makebox(0,0){\strut{}\Large $L_2$ Error}}}%
      \put(3632,220){\makebox(0,0){\strut{}\Large$\beta$}}%
    }%
    \gplgaddtomacro\gplfronttext{%
      \csname LTb\endcsname%
      \put(5948,2483){\makebox(0,0)[r]{\strut{}\large BA+LR}}%
      \csname LTb\endcsname%
      \put(5948,2263){\makebox(0,0)[r]{\strut{}\large Kikuchi $\infty$}}%
      \csname LTb\endcsname%
      \put(5948,2043){\makebox(0,0)[r]{\strut{}\large Kikuchi $10^4$}}%
      \csname LTb\endcsname%
      \put(5948,1823){\makebox(0,0)[r]{\strut{}\large Kikuchi $10^5$}}%
      \csname LTb\endcsname%
      \put(5948,1603){\makebox(0,0)[r]{\strut{}\large Kikuchi $10^6$}}%
      \csname LTb\endcsname%
      \put(5948,1383){\makebox(0,0)[r]{\strut{}\large PLM $10^4$}}%
      \csname LTb\endcsname%
      \put(5948,1163){\makebox(0,0)[r]{\strut{}\large PLM $10^5$}}%
      \csname LTb\endcsname%
      \put(5948,943){\makebox(0,0)[r]{\strut{}\large PLM $10^6$}}%
    }%
    \gplbacktext
    \put(0,0){\includegraphics{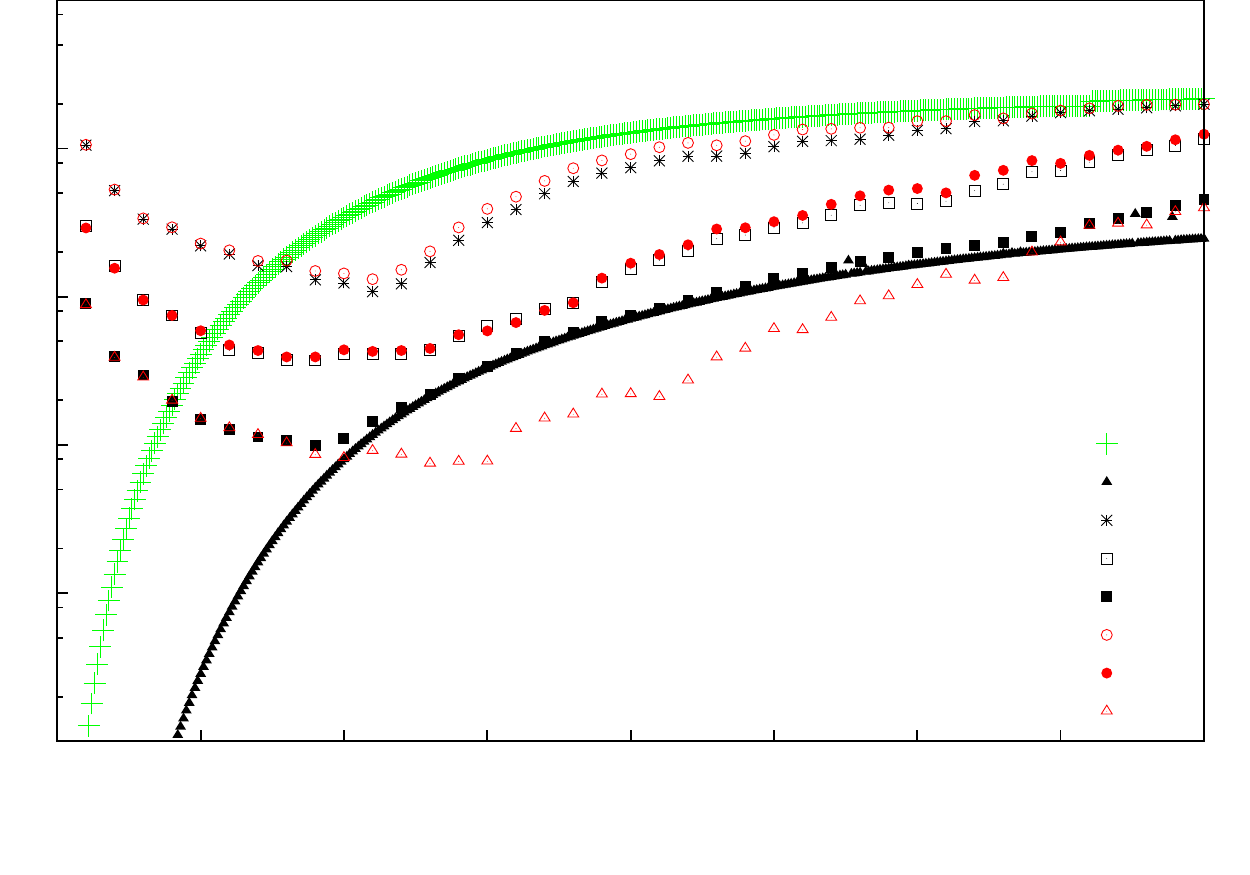}}%
    \gplfronttext
  \end{picture}%
\endgroup

%% file: bip_9.tex
\begingroup
  \makeatletter
  \providecommand\color[2][]{%
    \GenericError{(gnuplot) \space\space\space\@spaces}{%
      Package color not loaded in conjunction with
      terminal option `colourtext'%
    }{See the gnuplot documentation for explanation.%
    }{Either use 'blacktext' in gnuplot or load the package
      color.sty in LaTeX.}%
    \renewcommand\color[2][]{}%
  }%
  \providecommand\includegraphics[2][]{%
    \GenericError{(gnuplot) \space\space\space\@spaces}{%
      Package graphicx or graphics not loaded%
    }{See the gnuplot documentation for explanation.%
    }{The gnuplot epslatex terminal needs graphicx.sty or graphics.sty.}%
    \renewcommand\includegraphics[2][]{}%
  }%
  \providecommand\rotatebox[2]{#2}%
  \@ifundefined{ifGPcolor}{%
    \newif\ifGPcolor
    \GPcolortrue
  }{}%
  \@ifundefined{ifGPblacktext}{%
    \newif\ifGPblacktext
    \GPblacktexttrue
  }{}%
  \let\gplgaddtomacro\g@addto@macro
  \gdef\gplbacktext{}%
  \gdef\gplfronttext{}%
  \makeatother
  \ifGPblacktext
    \def\colorrgb#1{}%
    \def\colorgray#1{}%
  \else
    \ifGPcolor
      \def\colorrgb#1{\color[rgb]{#1}}%
      \def\colorgray#1{\color[gray]{#1}}%
      \expandafter\def\csname LTw\endcsname{\color{white}}%
      \expandafter\def\csname LTb\endcsname{\color{black}}%
      \expandafter\def\csname LTa\endcsname{\color{black}}%
      \expandafter\def\csname LT0\endcsname{\color[rgb]{1,0,0}}%
      \expandafter\def\csname LT1\endcsname{\color[rgb]{0,1,0}}%
      \expandafter\def\csname LT2\endcsname{\color[rgb]{0,0,1}}%
      \expandafter\def\csname LT3\endcsname{\color[rgb]{1,0,1}}%
      \expandafter\def\csname LT4\endcsname{\color[rgb]{0,1,1}}%
      \expandafter\def\csname LT5\endcsname{\color[rgb]{1,1,0}}%
      \expandafter\def\csname LT6\endcsname{\color[rgb]{0,0,0}}%
      \expandafter\def\csname LT7\endcsname{\color[rgb]{1,0.3,0}}%
      \expandafter\def\csname LT8\endcsname{\color[rgb]{0.5,0.5,0.5}}%
    \else
      \def\colorrgb#1{\color{black}}%
      \def\colorgray#1{\color[gray]{#1}}%
      \expandafter\def\csname LTw\endcsname{\color{white}}%
      \expandafter\def\csname LTb\endcsname{\color{black}}%
      \expandafter\def\csname LTa\endcsname{\color{black}}%
      \expandafter\def\csname LT0\endcsname{\color{black}}%
      \expandafter\def\csname LT1\endcsname{\color{black}}%
      \expandafter\def\csname LT2\endcsname{\color{black}}%
      \expandafter\def\csname LT3\endcsname{\color{black}}%
      \expandafter\def\csname LT4\endcsname{\color{black}}%
      \expandafter\def\csname LT5\endcsname{\color{black}}%
      \expandafter\def\csname LT6\endcsname{\color{black}}%
      \expandafter\def\csname LT7\endcsname{\color{black}}%
      \expandafter\def\csname LT8\endcsname{\color{black}}%
    \fi
  \fi
  \setlength{\unitlength}{0.0500bp}%
  \begin{picture}(7200.00,5040.00)%
    \gplgaddtomacro\gplbacktext{%
      \csname LTb\endcsname%
      \put(330,550){\makebox(0,0){\strut{} 0}}%
      \put(1651,550){\makebox(0,0){\strut{} 0.5}}%
      \put(2972,550){\makebox(0,0){\strut{} 1}}%
      \put(4293,550){\makebox(0,0){\strut{} 1.5}}%
      \put(5614,550){\makebox(0,0){\strut{} 2}}%
      \put(6935,550){\makebox(0,0){\strut{} 2.5}}%
      \put(7067,770){\makebox(0,0)[l]{\strut{} 1e-05}}%
      \put(7067,1624){\makebox(0,0)[l]{\strut{} 0.0001}}%
      \put(7067,2478){\makebox(0,0)[l]{\strut{} 0.001}}%
      \put(7067,3331){\makebox(0,0)[l]{\strut{} 0.01}}%
      \put(7067,4185){\makebox(0,0)[l]{\strut{} 0.1}}%
      \put(7067,5039){\makebox(0,0)[l]{\strut{} 1}}%
      \put(8232,2904){\rotatebox{-270}{\makebox(0,0){\strut{}\Large $L_2$ error}}}%
      \put(3632,220){\makebox(0,0){\strut{}\Large $\beta$}}%
    }%
    \gplgaddtomacro\gplfronttext{%
      \csname LTb\endcsname%
      \put(5948,2043){\makebox(0,0)[r]{\strut{}\large BA+LR}}%
      \csname LTb\endcsname%
      \put(5948,1823){\makebox(0,0)[r]{\strut{}\large Kikuchi $\infty$}}%
      \csname LTb\endcsname%
      \put(5948,1603){\makebox(0,0)[r]{\strut{}\large Kikuchi $10^4$}}%
      \csname LTb\endcsname%
      \put(5948,1383){\makebox(0,0)[r]{\strut{}\large Kikuchi $10^5$}}%
      \csname LTb\endcsname%
      \put(5948,1163){\makebox(0,0)[r]{\strut{}\large PLM $10^4$}}%
      \csname LTb\endcsname%
      \put(5948,943){\makebox(0,0)[r]{\strut{}\large PLM $10^5$}}%
    }%
    \gplbacktext
    \put(0,0){\includegraphics{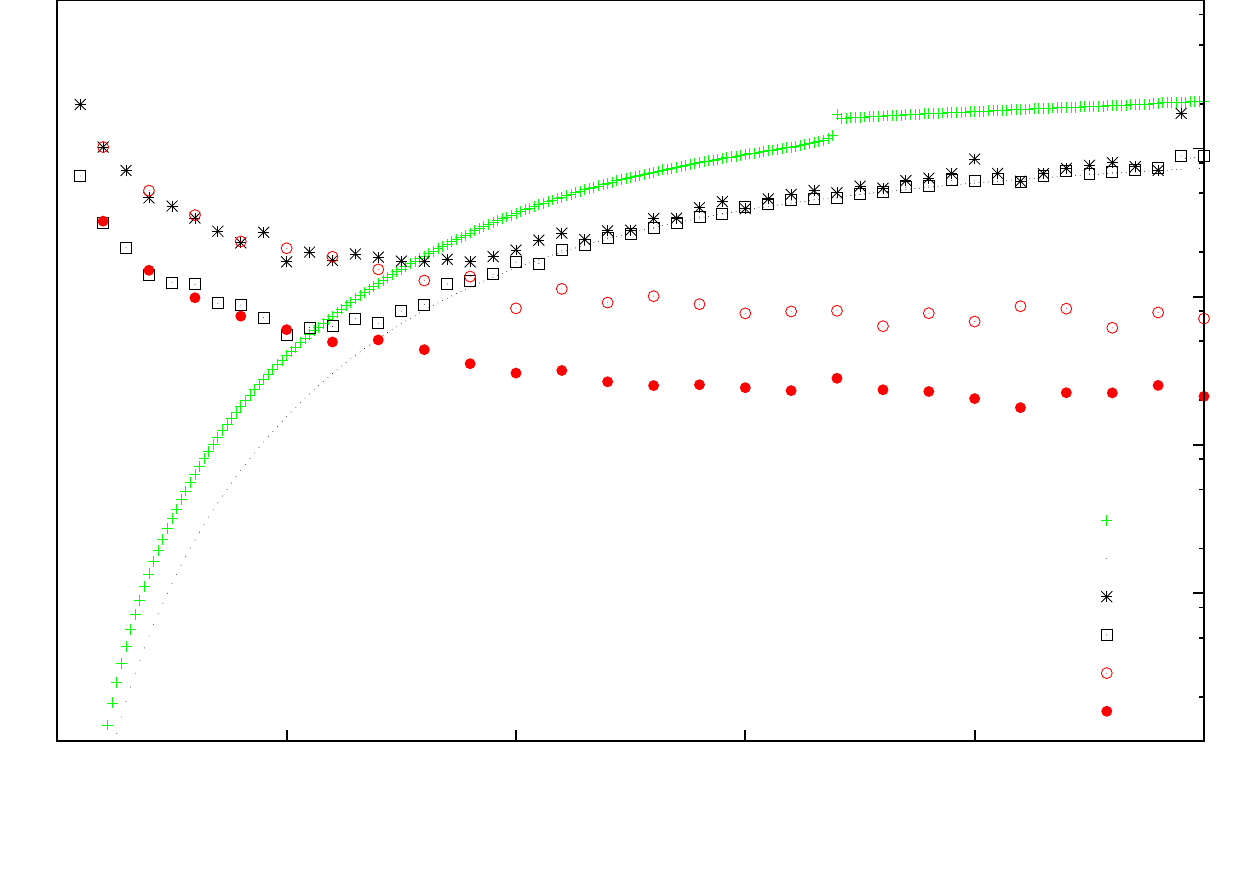}}%
    \gplfronttext
  \end{picture}%
\endgroup